\documentclass[11pt,letter]{article}

\usepackage[affil-it]{authblk}
\usepackage{natbib}
\usepackage{amssymb,amsmath,amsfonts,amstext}
\usepackage{graphicx}
\usepackage{mathtools}
\usepackage{stmaryrd}
\usepackage{color}
\usepackage{fullpage}

\usepackage{amsthm}
\usepackage{esint}
\newcommand{\X}{\mathbf{X}}
\newcommand{\N}{\mathbf{N}}
\newcommand{\F}{\mathbf{F}}
\newcommand{\Fe}{\mathbf{F}^\el}
\newcommand{\Fp}{\mathbf{F}^\pl}
\newcommand{\Fpe}{\mathbf{F}^\pl_{\epsilon}}
\newcommand{\Fee}{\mathbf{F}^\el_{\epsilon}}
\newcommand{\p}{\boldsymbol \varphi}
\newcommand{\pe}{\boldsymbol \varphi_{\epsilon}}
\newcommand{\hatpe}{\boldsymbol {\hat\varphi}_{\epsilon}}

\newtheorem{theorem}{Theorem}[section]
\newtheorem{lemma}[theorem]{Lemma}
\newtheorem{proposition}[theorem]{Proposition}

\def\pl{\mathrm{p}}
\def\pc{\mathrm{d}}
\def\el{\mathrm{e}}

\numberwithin{equation}{section}

\title{\textbf{Derivation of $\F=\Fe\Fp$ as the continuum limit of crystalline slip}}

\author{Celia Reina\footnote{Corrseponding author: creina@seas.upenn.edu}}%\corref{cor1}}
\affil{Department of Mechanical Engineering and Applied Mechanics, University of Pennsylvania, Philadelphia, PA 19104, USA}
\author{Anja Schl\"omerkemper}
\affil{Institut f\"ur Mathematik, Universit\"at W\"urzburg, Emil-Fischer-Str. 40, 97074 W\"urzburg, Germany}
\author{Sergio Conti}
\affil{Institut f\"ur Angewandte Mathematik, Universit\"at Bonn, 53115 Bonn, Germany }

\begin{document}
\maketitle
\begin{abstract}
In this paper we provide a proof of the multiplicative kinematic description of crystal elastoplasticity in the setting of large deformations, i.e.~$\F=\Fe\Fp$, for a two dimensional single crystal. The proof starts by considering a general configuration at the mesoscopic scale, where the dislocations are discrete line defects (points in the two-dimensional description used here) and the displacement field can be considered continuous everywhere in the domain except at the slip surfaces, over which there is a displacement jump. At such scale, as previously shown by two of the authors, there exists unique physically-based definitions of the total deformation tensor $\F$ and the elastic and plastic tensors $\Fe$ and $\Fp$ that do not require the consideration of any non-realizable intermediate configuration and do not assume any a priori relation between them of the form $\F=\Fe\Fp$. This mesoscopic description is then passed to the continuum limit via homogenization i.e., by increasing the number 
of slip surfaces to infinity and reducing the lattice parameter to zero. We show for two-dimensional deformations of initially perfect single crystals that the classical 
continuum 
formulation is recovered in the limit with $\F=\Fe\Fp$, $\det \Fp= 1$ and $\mathbf{G}=\text{Curl}\ \Fp$ the dislocation density tensor.
\end{abstract}

%\begin{keyword}
%    Crystal plasticity \sep
%    Finite kinematics \sep
%    Dislocation density tensor \sep
%    Homogenization
%\end{keyword}

%\end{frontmatter}

%\tableofcontents

\section{Introduction}

Standard continuum models in the setting of large deformations are based on the kinematic assumption $\F=\Fe\Fp$, which decomposes the total deformation gradient $\F$ multiplicatively into the effective deformation induced by elastic and plastic mechanisms, $\Fe$ and $\Fp$ respectively. This approach was first introduced in the 1960's \citep{LeeLiu1967} and has been very successful in many engineering applications \citep{Simo1988, OrtizRepetto1999, NematNasser2004, Marian2013, Abaqus}. However, the lack of a well-grounded justification for this decomposition has led to numerous debates in the literature. The heuristics behind the expression $\F=\Fe\Fp$ is based on the chain rule for the derivative of the deformation mapping $\F=D\boldsymbol\varphi$, when the elastic deformation of the system can be fully relaxed, leading to a purely plastic distortion. At the core of the problem lies the fact that dislocations, when present, physically couple the elastic and plastic field around them, preventing the 
existence of a physical configuration with associated deformation $\F^\el$ (or $\F^\pl$) and thus exclusively elastic (or plastic). This inherent complexity has raised many issues, such as the physical meaning of the individual tensors $\F^\el$ or $\F^\pl$, the existence \citep{
CaseyNaghdi1992, 
Owen2002} and uniqueness  \citep{NematNasser1979, LubardaLee1981, Zbib1993, Naghdi1990, GreenNaghdi1971, CaseyNaghdi1980, Mandel1973, Dafalias1987, Rice1971,Meyers2006} of the decomposition $\F=\F^\el\F^\pl$, or the appropriate measure of dislocation content in the body \citep{Bilby1955, Eshelby1956, Kroner1960, Fox1966, Willis1967, AcharyaBassani2000, CermelliGurtin2001} based in the incompatibilities of elastic or plastic component of the deformation. Other decompositions have also been proposed in the literature, some additive, typically in rate form \citep{NematNasser1979,Zbib1993,Pantelides1994}, and others multiplicative ($\F=\Fp\Fe$, \cite{Clifton1972}, \cite{Lubarda1999}; $\F$ based on the product of three tensors, one elastic and two inelastic, \cite{Lion2000, Dawson2008, Anand2009,ClaytonHartleyMcdowell2014}).

A micromechanical definition of the different tensors  $\F$, $\F^\el$ and $\F^\pl$ has been recently provided \citep{ReinaConti2014} without any a priori assumption on the relationship between them. These definitions are exclusively based on kinematic arguments and physical considerations, are uniquely determined from the microscopic displacement field, and lead to the expected multiplicative decomposition in the continuum limit in the absence of dislocations. From the previous definition of $\Fp$, $\text{Curl}\ \F^\pl$ arises as the natural measure for the dislocation density tensor in the material.

Based on the above micromechanical definitions of $\F$, $\F^\el$ and $\F^\pl$, we provide in this paper a rigorous proof of the continuum kinematic assumption $\F=\F^\el\F^\pl$, with $\det \F^\pl = 1$, for a general dislocation and slip structure. Such generality is achieved in this work by considering microscopic elastoplastic displacement fields of finite energy that belong to a functional space in which displacements are continuous everywhere in the domain except potentially at the slip surfaces over which there is a displacement jump. The continuum limit is then mathematically studied by letting the lattice parameter tend to zero and increasing the number of dislocations and slip surfaces to infinity. We show, for a two-dimensional deformation of an initially perfect single crystal, that the displacement field is indeed continuous in the limit, and that the limit of $\F$ coincides with the product of the limit of $\F^\el$ and $\F^\pl$. The later statement is highly non-trivial since limits 
and products do not generally commute. This complexity, combined with the non-linearity induced by the finite kinematic description, has so far limited most of the mathematical results of coarse-graining of dislocation ensembles to the linearized kinematic setting, see e.g. \cite{ContiOrtiz05,GarronileoniPonsiglione2010}, where the principle of superposition applies and the decomposition can be easily proven to be additive, i.e. $\boldsymbol \varepsilon=\boldsymbol \varepsilon^\el+\boldsymbol \varepsilon^\pl$, or to the case of 
well-separated dislocations \citep{ScardiaZeppieri2012, Mueller2014}. Even at the continuum scale, few mathematical results exist on finite elastoplasticity. Some examples can be found in \cite{Mielke2006,Mielke2009} and \cite{Conti2011}.
A discrete model which does not require the existence of a reference configuration was developed by \cite{LuckhausMugnai2010}.

The paper is organized as follows. We begin in Section \ref{Sec:Description} by introducing the mathematical notation needed for the remainder of the text and reviewing the microscopic definitions of $\F$, $\F^\el$ and $\F^\pl$ following the lines of \cite{ReinaConti2014}. Next, in Section \ref{Sec:Scaling} we discuss the passage of the aforementioned quantities to the continuum limit from a physical perspective, and provide some further guidance on notation in Section \ref{Sec:Notation}. This is followed by Sections \ref{Sec:ProblemSetting} and \ref{Sec:MainResults} where the precise setting of the problem and the mathematical results are shown. The paper is then finalized with some conclusions in Section \ref{Sec:Conclusions}.

\section{Mesoscopic description of elastoplastic deformations} \label{Sec:Description}

\begin{figure}
\begin{center}
    {\includegraphics[width=0.65\textwidth]{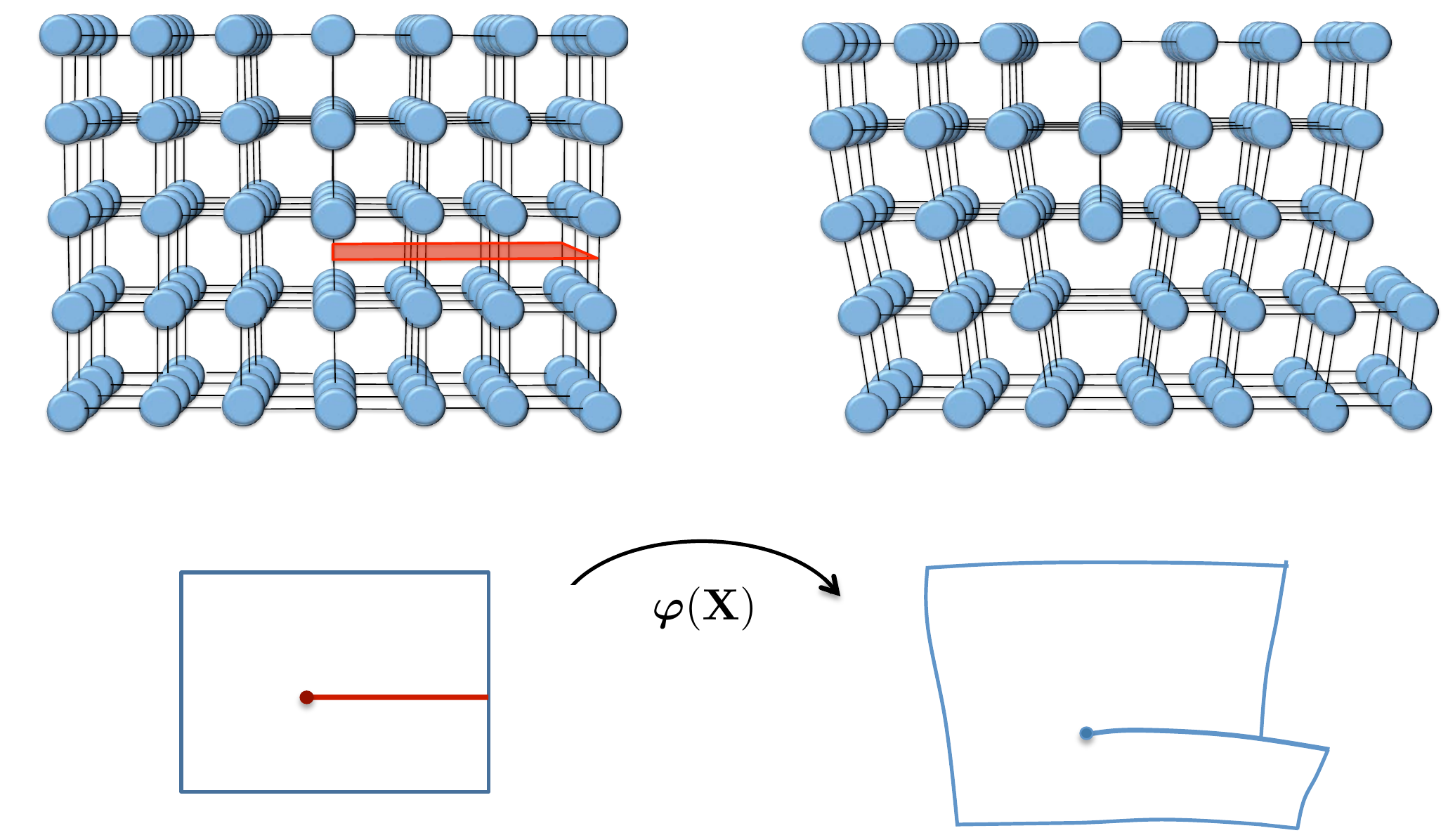}}
    \caption[]{Elastoplastic deformation (right images) of a perfect crystal (left images) at the atomistic (top\footnotemark{})
%reprinted from \cite{ReinaConti2014}  with permission from Elsevier)}    
and mesoscopic scale (bottom) induced by the glide of an edge dislocation from the right surface of the domain till its center. The surface of displacement discontinuity is represented in red in the undeformed configuration and terminates at the dislocation point and the boundary of the domain.}
    \label{Fig:DeformationDescription}
\end{center}
\end{figure}
\footnotetext{Reprinted from Publication `Kinematic description of crystal plasticity in the finite kinematic framework: A micromechanical understanding of F=FeFp', Vol 67, Authors C. Reina and S. Conti, Page No. 43, Copyright (2014), with permission from Elsevier}

The present study is exclusively concerned with two-dimensional elastoplastic deformations induced by dislocation glide (of edge type) and disregards other processes such as phase transformations, diffusion, void nucleation or appearance of microcracks. Mathematically, we describe these deformations with a semi-continuous Lagrangian formulation, where the final configuration $\boldsymbol \varphi(\mathbf{X})$ of every material point $\mathbf{X}$ in the reference domain $\Omega \subset \mathbb{R}^2$, considered here to be a perfect crystal, is described via a deformation mapping $\boldsymbol\varphi:\Omega \rightarrow \boldsymbol\varphi(\Omega) \subset \mathbb{R}^2$. In particular, $\boldsymbol \varphi$ is continuous everywhere in the domain, except at the area swept by the dislocations during their motion (lines in two dimensions), over which there is a displacement jump or discontinuity, c.f.~Fig.~\ref{Fig:DeformationDescription}. These type of deformations belong to the space of special 
functions of bounded variation or $SBV$,
c.f.~\cite{EvansGariepy1991,AmbrosioFuscoPallara2000}; and we further require that $\boldsymbol \varphi$ is one-to-one in order to avoid interpenetration of matter.
%{\color{magenta} Sergio: dropped the footnote, since this is already explained twice below}
%\footnote{The value of the deformation mapping at the jump set is set to be the value at one of its two sides, making it bijective in the full domain.}. 
The deformation gradient $\F$ is then defined as the distributional gradient of $\boldsymbol \varphi$, which for SBV functions is a measure of the form 
\begin{equation} \label{Eq:Definition_F}
\F = D \boldsymbol \varphi = \nabla \boldsymbol \varphi\ \mathcal{L}^2 + \llbracket \boldsymbol \varphi \rrbracket \otimes \mathbf{N}\ \mathcal{H}^1 \lfloor_{\mathcal{J}}. 
\end{equation}

The measure $\F$ consists of an absolutely continuous part, $\nabla \boldsymbol \varphi\ \mathcal{L}^2$, and a singular part, $ \llbracket \boldsymbol \varphi \rrbracket\otimes \mathbf{N}\ \mathcal{H}^1 \lfloor_{\mathcal{J}}$, that has its support over the slip lines or jump set $\mathcal{J}$, where $\mathbf{N}(\mathbf{X})$ is the unit normal at $\mathbf{X} \in \mathcal{J}$ and $\llbracket \boldsymbol \varphi \rrbracket$ is the jump over the slip line, see below.
We denote by $\mathcal{L}^2$ the Lebesgue measure  (which measures area in the plane)
and by $\mathcal{H}^1$ the one-dimensional Hausdorff measure (which measures length). $\nabla \boldsymbol \varphi$ is called the approximate differential of $\boldsymbol \varphi$ and corresponds to the total deformation $\F$ at the points where no slip occurs and the deformation is thus purely elastic. As a result, $\nabla \boldsymbol \varphi$ can be physically identified with the elastic deformation tensor
\begin{equation}\label{Eq:Definition_Fe}
\F^\el =  \nabla \boldsymbol \varphi.
\end{equation} 
The symbol $\nabla$ corresponds to the `standard' gradient where the function is differentiable, in particular, only outside of the jump set. By definition, $\F$ has a vanishing Curl. However, Curl $\nabla \boldsymbol \varphi$ does not vanish in general, and the same occurs for the singular part of the deformation gradient. We note that an additive decomposition of the deformation gradient of the form of \eqref{Eq:Definition_F}, with the first term being identified as the elastic deformation tensor, has been previously considered in the mechanics literature based on local volumetric averages and application of the divergence theorem over continuous regions of the domain \citep{Davison1995}.

\begin{figure}
%\begin{minipage}
\begin{center}
    {\includegraphics[width=0.8\textwidth]{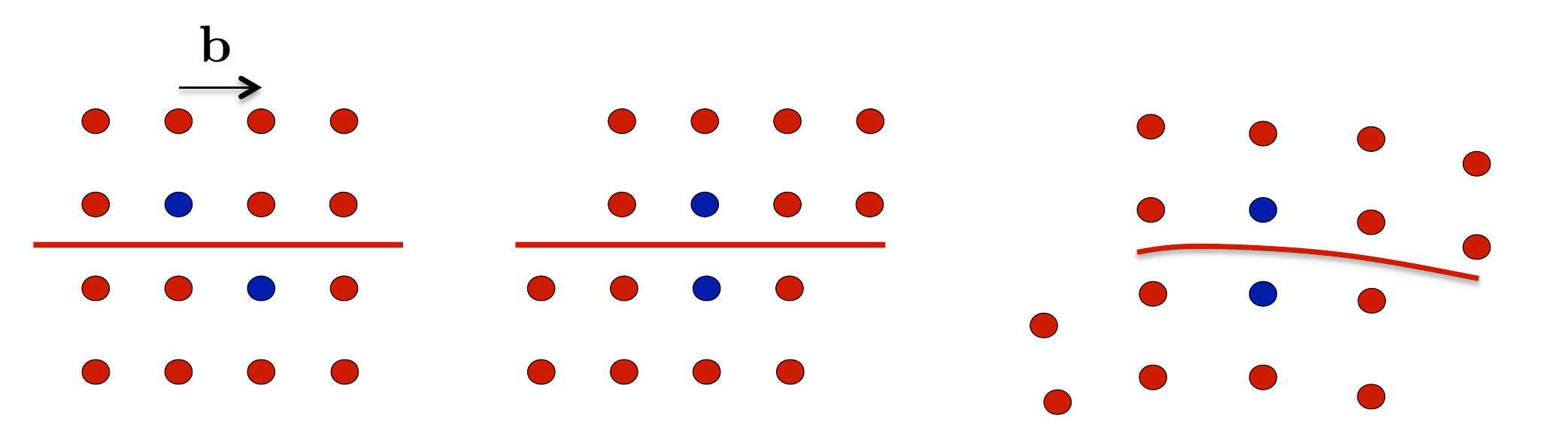}}
    \caption[]{Atomic positions of a perfect crystal (left image) after a  plastic deformation (middle image) and an elastoplastic deformation (right image)\footnotemark{}.}
    \label{Fig:SlipCondition}
\end{center}
%\end{minipage}
\end{figure}
\footnotetext{Reprinted from Publication `Kinematic description of crystal plasticity in the finite kinematic framework: A micromechanical understanding of F=FeFp', Vol 67, Authors C. Reina and S. Conti, Page No. 45, Copyright (2014), with permission from Elsevier.}

For its part, $\llbracket \boldsymbol \varphi \rrbracket (\mathbf{X}) = \boldsymbol \varphi^+(\mathbf{X}) - \boldsymbol \varphi^-(\mathbf{X})$ is the displacement jump at point $\mathbf{X}$ in the slip line, where the $+$ side is the one indicated by the normal $\mathbf{N}$ to the line at point $\mathbf{X}$. Without further restrictions on $\llbracket \boldsymbol \varphi \rrbracket$ and $\mathcal{J}$, the class of deformation mappings considered thus far can also model processes such as cavitation, crack opening or interpenetration of matter. A careful analysis of the kinematics of slip in crystalline materials, c.f.~\cite{ReinaConti2014}, indicates that jump sets in two dimensions necessarily consist of an ensemble of straight segments in the reference configuration that terminate at dislocation points or exit the domain, c.f.~Fig.~\ref{Fig:DeformationDescription}; and that the displacement jump for single slip satisfies, c.f.~Fig.~\ref{Fig:SlipCondition}, 
\begin{equation} \label{Eq:Slip}
\begin{split}
&\boldsymbol \varphi^-(\mathbf{X}) =  \boldsymbol \varphi^+(\mathbf{X}-\mathbf{b}), \quad \forall\ \mathbf{X}, \mathbf{X}-\mathbf{b} \in \mathcal{J}, \\
&\boldsymbol \varphi^+(\mathbf{X})= \boldsymbol \varphi^-(\mathbf{X}+ \mathbf{b}), \quad \forall\ \mathbf{X}, \mathbf{X}+\mathbf{b} \in \mathcal{J},
\end{split}
\end{equation}
where $\mathbf{b}$ is called the Burgers vector and represents the interatomic distance in the direction of slip. 
We remark that the two conditions in (\ref{Eq:Slip}) are equivalent when $\mathbf{X}, \mathbf{X}-\mathbf{b},\mathbf{X}+\mathbf{b} \in \mathcal{J}$.

The jump set $\mathcal{J}$ and the previously introduced vector $\mathbf{b}$ for a single slip system, clearly characterize the plastic deformation with quantities independent of the elastic distortion. Indeed, if we consider for instance a plastic simple shear deformation, such as $\boldsymbol \varphi_1$  in Fig.~\ref{Fig:SimpleShears}, the total deformation is purely plastic and thus equal to $\F^\pl$, and can be expressed, c.f.~Eq.~\eqref{Eq:Definition_F}, as
\begin{equation}
\F=D\boldsymbol \varphi_1 = \F^\pl=\mathbf{I}\ \mathcal{L}^2 + |\mathbf{b}_1| \mathbf{e}_1 \otimes \mathbf{e}_2\ \mathcal{H}^1\lfloor_{\mathcal{J}}.
\end{equation}

\begin{figure}
\begin{center}
    {\includegraphics[width=0.7\textwidth]{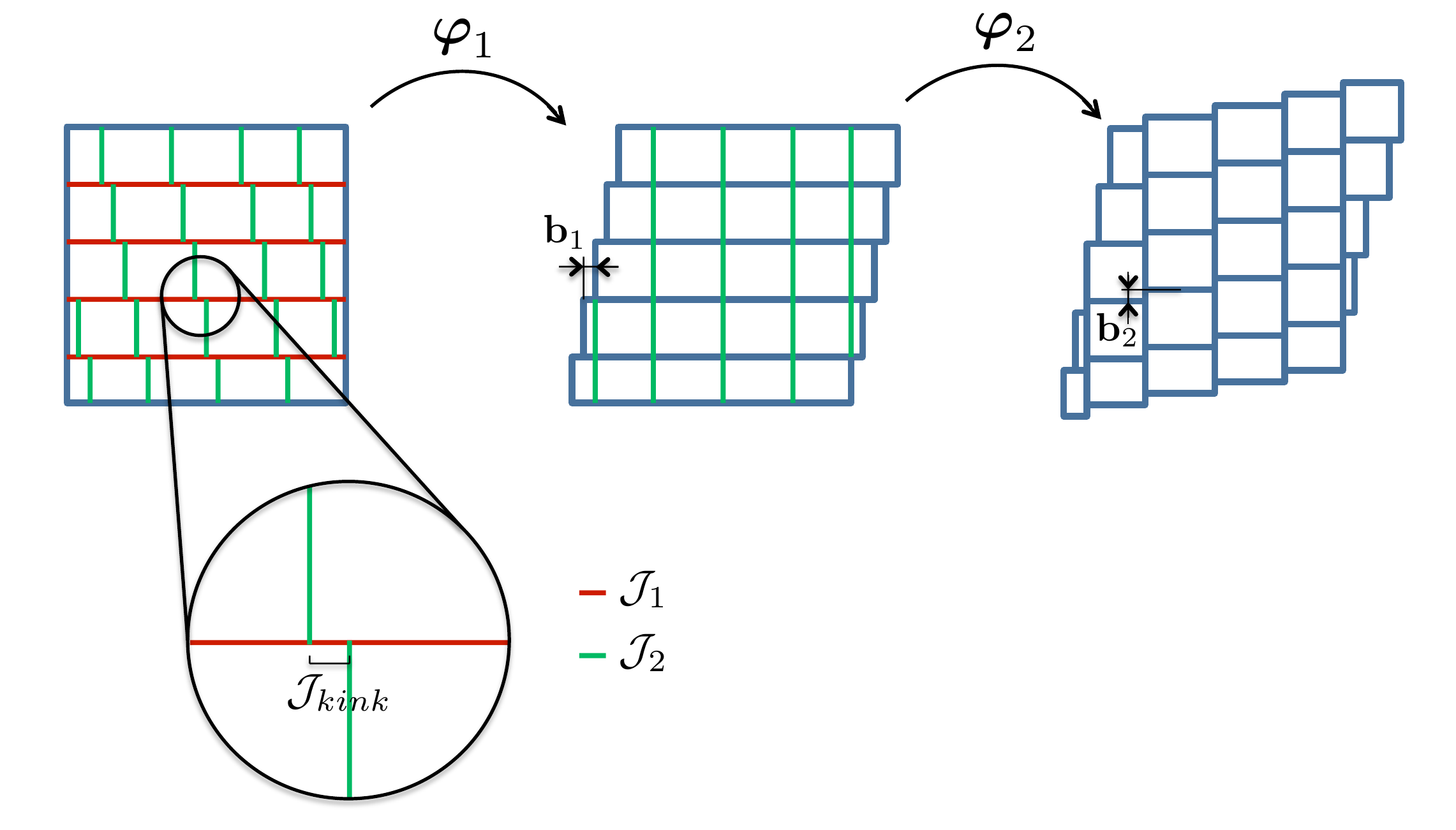}}
    \caption[]{Sequence of plastic deformations involving two orthogonal slip systems. The zoom in the reference configuration is used to label the segments composing the jump set. }
    \label{Fig:SimpleShears}
\end{center}
\end{figure}

Similarly, for the composition of two simple shears where slip lines intersect, c.f.~Fig.~\ref{Fig:SimpleShears} and~\cite{ReinaConti2014}, the plastic deformation tensor is uniquely given from Eq.~\eqref{Eq:Definition_F} as
\begin{equation} \label{Eq:Fp_2simpleshears}
\F=D\left(\boldsymbol \varphi_2 \circ \boldsymbol \varphi_1 \right)=\F^\pl = \mathbf{I}\ \mathcal{L}^2 + |\mathbf{b}_{1} |\ \mathbf{e}_1\otimes \mathbf{e}_2\ \mathcal{H}^1\lfloor_{\mathcal{J}_{1}} + |\mathbf{b}_{2} |\ \mathbf{e}_2\otimes \mathbf{e}_1\ \mathcal{H}^1\lfloor_{\mathcal{J}_{2}} +  |\mathbf{b}_{2} | \mathbf{e}_2\otimes \mathbf{e}_2\ \mathcal{H}^1\lfloor_{\mathcal{J}_\mathrm{kink}}.
\end{equation}
The last term in Eq.~\eqref{Eq:Fp_2simpleshears} relates to the kink that results from the pullback of the second slip lines into the reference configuration and provides the non-commutative term in the finite kinematic formulation. In other words, the pullback is performed in opposite order in which the slip occurred and therefore the resulting kink term is dependent on whether the slip is first performed in the horizontal or vertical direction.

The above results can be easily generalized for an arbitrary sequence of $N_s$ dislocation-free plastic slips, not necessarily homogeneous or belonging to orthogonal slip systems. In that case, $\boldsymbol \varphi$ can be written as $\boldsymbol \varphi= \boldsymbol \varphi_{ N_s} \circ ... \circ \boldsymbol \varphi_{ \nu} \circ ... \circ \boldsymbol \varphi_{ 1}$, where $\mathcal{J}_{\boldsymbol \varphi_{ \nu}}$ (jump set associated to $\boldsymbol \varphi_{ \nu}$) is an ensemble of parallel straight infinite lines with normal $\mathbf{N}_{\nu}$ and 
\begin{equation}
D\boldsymbol \varphi_{ \nu} = \mathbf{I}\ \mathcal{L}^2 +  \mathbf{b}_{ \nu } \otimes \mathbf{N}_{\nu}\ \mathcal{H}^1 \lfloor_{\mathcal{J}_{\boldsymbol \varphi_{ \nu}}},  \quad  \mathbf{b}_{ \nu } \cdot \mathbf{N}_{\nu} = 0, \quad \mathbf{b}_{\nu } \in \mathbb{R}^2\ \text{ constant on each line}.
\end{equation}
The total plastic deformation tensor associated to $\boldsymbol \varphi$ can then be written as
\begin{equation} \label{Eq:Fp_definition}
\F^\pl = \mathbf{I}\ \mathcal{L}^2 + \sum_j \mathbf{b}_{j} \otimes \mathbf{N}_j\ \mathcal{H}^1 \lfloor_{\mathcal{J}_{j}},
\end{equation}
where the total jump set $\mathcal{J}$ in the reference configuration consists on the union of straight segments $\mathcal{J}_j$ (appropriately including all kinks), each of which has an associated Burgers vector $\mathbf{b}_j$. In general, as was the case for the composition of two simple shears, $\mathbf{b}_j\cdot \mathbf{N}_j $ is not necessarily zero. Equation \eqref{Eq:Fp_definition} remains of course valid if an elastic deformation is superposed to the previous plastic distortion, i.e. $\boldsymbol \varphi=\boldsymbol \varphi^\el \circ \boldsymbol \varphi^\pl$, $\F^\pl= D \boldsymbol \varphi^\pl$. In such case, by construction, $\F = D\boldsymbol \varphi = D\left(\boldsymbol \varphi^\el \circ \boldsymbol \varphi^\pl \right) = \nabla \boldsymbol \varphi^\el \circ \boldsymbol \varphi^\pl \ \mathcal{L}^2 + \left( \boldsymbol \varphi^\el\left(\boldsymbol \varphi^{p+} \right)- \boldsymbol \varphi^\el\left(\boldsymbol \varphi^{p-} \right) \right) \otimes \mathbf{N}\ \mathcal{H}^1\lfloor_{\mathcal{J}}$, $\F
^\el = \nabla \boldsymbol \varphi = \nabla \boldsymbol \varphi^\el \circ \boldsymbol \varphi^\pl$, and, if $\Fe$ is sufficiently smooth one obtains that for compatible domains $\F= \F^\el\F^\pl+\mathcal{O}(|\mathbf{b}|^2)$, c.f.~\cite{ReinaConti2014}.

In a general elastoplastic deformation, dislocations are present in the body, coupling the elastic and plastic field. They render impossible a global decomposition of the total deformation mapping into a purely plastic distortion and a subsequent elastic deformation, i.e.~$\boldsymbol \varphi \neq \boldsymbol \varphi^\el \circ \boldsymbol \varphi^\pl$, thus complicating, in principle, the physical definition of the global plastic deformation tensor $\F^\pl$. However, at this mesoscopic scale where the dislocations are individually resolved, subdomains away from the dislocations are, by construction, defect free and it is therefore possible to obtain in each of them a decomposition of the form $\boldsymbol \varphi = \boldsymbol \varphi^\el \circ \boldsymbol \varphi^\pl$ and an associated $\F^\pl$ given by $D\boldsymbol \varphi^\pl$, as in Eq.~\eqref{Eq:Fp_definition}. It will be shown in Section \ref{Sec:ProblemSetting} that $\F^\pl$ is uniquely defined from $\boldsymbol \varphi$ 
regardless of the potentially many decompositions of the form $\boldsymbol \varphi^\el \circ \boldsymbol \varphi^\pl$ in each of these subdomains, thus providing a global definition of $\F^\pl$ almost everywhere in the domain (except at the dislocation cores and close to the boundary). 

It can therefore be concluded that $\F$, $\F^\el$ and $\F^\pl$ are uniquely defined from the mesoscopic deformation mapping $\boldsymbol \varphi$ in general elastoplastic deformations induced by dislocation glide, and are of the form
\begin{equation} \label{Eq:all_F}
\begin{split}
&\F = D \boldsymbol \varphi = \nabla \boldsymbol \varphi\ \mathcal{L}^2 + \sum_j\llbracket \boldsymbol \varphi \rrbracket \otimes \mathbf{N}_j\ \mathcal{H}^1 \lfloor_{\mathcal{J}_j}, \\
&\F^\el =  \nabla \boldsymbol \varphi, \\
&\F^\pl =  \mathbf{I}\ \mathcal{L}^2 + \sum_j \mathbf{b}_{j} \otimes \mathbf{N}_j\ \mathcal{H}^1 \lfloor_{\mathcal{J}_{j}}.
\end{split}
\end{equation}
As will be shown in the following sections, the precise definition of $\F^\pl$ inside the dislocation cores will be irrelevant in the continuous limit. For concreteness, we keep the definition above with $\mathbf{b}_j$ at each slip line inside the cores as the Burgers vector of the corresponding segment outside of the core (this will be made precise in Section \ref{Sec:ProblemSetting}). This choice of $\F^\pl$ associates to each segment composing the jump set a constant Burgers vector, and as derived analytically in \cite{ReinaConti2014}, $\text{Curl}\ \F^\pl=\sum_k \mathbf{b}_k\ \delta_{\mathbf{X}_k}$, where $\{\mathbf{X}_k \}$ is the ensemble of dislocation points and $\mathbf{b}_k$ the corresponding Burgers vector (sum of the Burgers vector of the slip lines terminating at the dislocation point, with the appropriate sign). In other words, $\mathbf{G}=  \text{Curl}\ \F^\pl$ exactly measures the dislocation content in the body and is therefore an appropriate definition of the 
dislocation density tensor.

\section{Scaling towards the continuum limit} \label{Sec:Scaling}

\begin{figure}
\begin{center}
    {\includegraphics[width=0.75\textwidth]{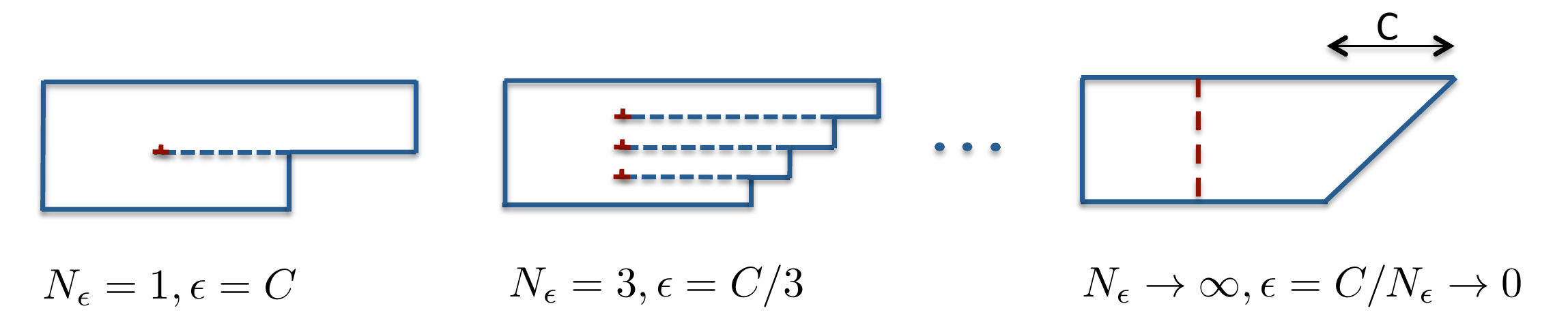}}
    \caption[]{Sequence of elastoplastic deformations as the lattice parameter $\epsilon$ tends to zero, and the number of dislocations $N_{\epsilon}$ and length of the jump set $|\mathcal{J}|$ tends to infinity. }
    \label{Fig:Sequence}
\end{center}
\end{figure}

Ultimately, we are interested in the description of the elastoplastic deformation in the continuum limit. To that regard, we consider sequences of deformations, c.f.~Fig.~\ref{Fig:Sequence}, in which each element of the sequence is characterized by the parameter $\epsilon$ that measures the normalized lattice parameter (ratio between the shortest atomic distance and the characteristic length of the body). In particular, each deformation mapping will be denoted as $\boldsymbol \varphi_{\epsilon}$, and its jump set and the Burgers vector of each segment as $\mathcal{J}= \sum_j  \mathcal{J}_j$ and $\mathbf{b}_{\epsilon j}$ respectively. Similarly, the different deformation tensors, as given by Eqs.~\eqref{Eq:all_F}, will be denoted as $\F_{\epsilon}$, $\F^\el_{\epsilon}$ and $\F^\pl_{\epsilon}$. Furthermore, by crystallographic arguments, the number of slip systems is finite and bounded by some fixed $N_s$; $ |\mathbf{b}_{\epsilon j}| \leq n \epsilon$ for some fixed $n>0$; and planes of the 
same 
slip system are separated by a distance that is equal to or larger than $\epsilon$. From the latter, it is implied that, c.f.~Fig.~\ref{Fig:JumpsetBound}, 
\begin{equation} \label{Eq:Length_jumpset}
\mathcal{H}^1 \left(\mathcal{J}\cap B_r(\mathbf{X}) \right) \leq A \frac{r^2}{\epsilon},
\end{equation}
for some $A$  that depends on $N_s$. We note that the subindex $\epsilon$ is omitted for the jump set $\mathcal{J}$ since the limiting deformation mapping will be continuous, and therefore $\mathcal{J}$ is unequivocally associated to an element of the sequence.

\begin{figure}
\begin{center}
    {\includegraphics[width=0.3\textwidth]{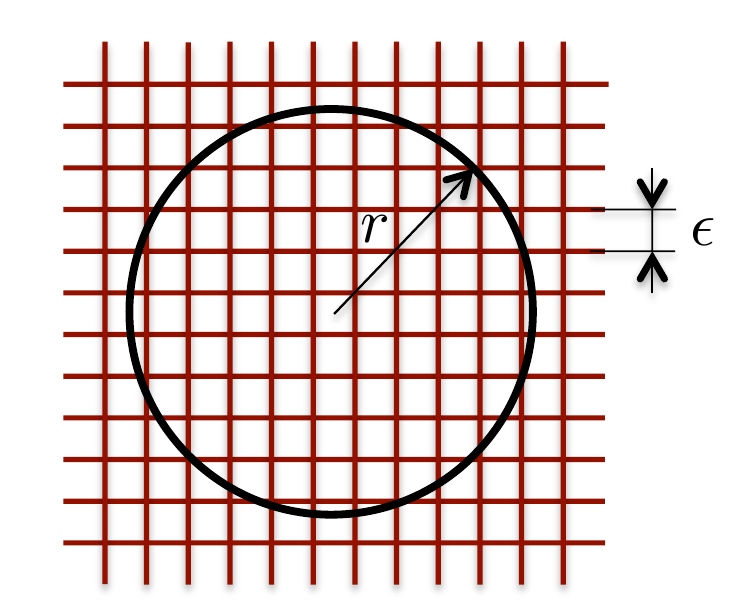}}
    \caption[]{Jump sets corresponding to the same slip system are separated by a distance of at least $\epsilon$. In the case of two orthogonal slip systems, the maximum length of slip lines contained in a ball of radius $r$ is then proportional to $\frac{r^2}{\epsilon}$.}
    \label{Fig:JumpsetBound} 
\end{center}
\end{figure}

The continuum limit is then attained as $\epsilon$ vanishes, and the macroscopic kinematic quantities $\boldsymbol \varphi$, $\F$, $\F^\el$ and $\F^\pl$ are uniquely defined  as the limit of the corresponding mesoscopic quantities, i.e.
\begin{equation}\label{Eq:limits}
\boldsymbol \varphi_{\epsilon} \rightarrow \boldsymbol \varphi, \quad \F_{\epsilon} \rightarrow \F, \quad \F^\el_{\epsilon} \rightarrow \F^\el, \quad \F^\pl_{\epsilon} \rightarrow \F^\pl
\end{equation}
as $\epsilon \rightarrow 0$ in the appropriate topology. This perspective thus delivers a well defined $\F^\el$ and $\F^\pl$, that do not invoque any macroscopic unloading path, which in general does not exist. Only at the mesoscopic scale, as discussed in the previous section, a relaxation process (or decomposition of the deformation mapping into a physically realizable elastic and plastic deformation) is defined everywhere except at the dislocation cores. 

In analogy to a zoom out process in a real material, the number of dislocations $N_{\epsilon}$ and the total length of the jump set $|\mathcal{J}|$ tend, potentially, to infinity as $\epsilon \rightarrow 0$. Equation \eqref{Eq:Length_jumpset} directly implies
\begin{equation}
\mathcal{H}^1\left(\mathcal{J}\right) \epsilon \leq C^*,
\end{equation}
for some $C^*$, and thus $\F^\pl_{\epsilon}$ is at most a finite deformation, as desired. We further assume that the number of dislocations $N_{\epsilon}$ satisfies 
\begin{equation} \label{Eq:ScalingDislocations}
N_{\epsilon}\epsilon \leq C,
\end{equation}
for some $C$. This scaling allows the formation of dislocation walls and prevents the clustering of infinitely many dislocations at one macroscopic material point, as is commonly observed in experiments. To further support the scaling given by Eq.~\eqref{Eq:ScalingDislocations}, consider an ensemble of identical dislocations positioned over a regular square pattern, with distance $\delta$ between dislocations. A simple energetic analysis in the linearized kinematic framework indicates that the system will have finite elastic energy only if $\delta/\sqrt{\epsilon}$ is bounded, and therefore the total amount of dislocations $N_{\epsilon}$ is proportional to $\frac{1}{\delta^2} \leq \frac{C}{\epsilon}$. 

In the following sections, we will show that the macroscopic quantities $\F$, $\F^\el$ and $\F^\pl$ defined in \eqref{Eq:limits} are related in the continuum limit via the standard multiplicative decomposition $\F=\F^\el\F^\pl$ with  $\det \F^\pl=1$, that $\boldsymbol \varphi$ is continuous in the limit and that $\mathbf{G}=\text{Curl}\ \F^\pl$ represents the continuum dislocation density tensor when expressed in the reference configuration. This is done for general sequences of elastoplastic deformations with bounded energy.

\section{Notation} \label{Sec:Notation}
The notation used in the foregoing will follow, when possible, the standard notation used in continuum mechanics. In particular, vectors and tensors are denoted with boldface symbols (the order of the tensor will be clear from the context), and scalars are written in standard font. Similarly $dX=d\mathcal{L}^2$ denotes the differential area  in the reference configuration and $d\mathbf{X}$ a differential vector, so that
\begin{equation}
\int_{\mathbf{X}_1}^{\mathbf{X}_2} \mathbf{f}(\mathbf{X})\cdot d\mathbf{X} = \int_0^1 \mathbf{f} \left( \mathbf{X}_1+t \left(\mathbf{X}_2-\mathbf{X}_1 \right) \right) \cdot \left(\mathbf{X}_2-\mathbf{X}_1 \right) dt.
\end{equation}
The scalar product between two vectors is denoted as $\mathbf{f}\cdot\mathbf{g}$, the full contraction between two tensors as $\mathbf{A}:\mathbf{B}$, the matrix vector multiplication with no symbol $ \mathbf{A} \mathbf{g}$ and $\otimes$ is used to denote the dyadic product between two vectors $\mathbf{f}\otimes \mathbf{g}$. Furthermore, $|\cdot |$ is used to denote the Euclidean norm of the scalar, vector or tensor that the symbol encloses. For sets, we use the notation $|A|=\mathcal{H}^1(A)$ or $|A|=\mathcal{L}^2(A)$, where the dimension of the set $A$ is clear from the context.

The spaces that will be used in the following are the space of measures $\mathcal{M}$, the $L^p$ spaces, the space of functions of bounded variation $BV$ and the space of special functions of bounded variations $SBV$. Further details on these spaces may be found in \cite{AmbrosioFuscoPallara2000}, \cite{EvansGariepy1991} and \cite{Adams2003}.

The deformation mappings $\pe$ and $\pe^\pl$ will be considered uniquely defined at the jump set, by arbitrarily choosing the value at either side of the jump. With these considerations in mind, both mappings are bijective, see Fig.~\ref{Fig:Jumps}. In the compatible regions of the domain, where $\pe$ is of the form $\pe=\pe^\el\circ\pe^\pl$, the choice of $\pe$, $\pe^\pl$ and the domain of $\pe^\el$ need to be consistent with each other. Similar considerations apply for $\F^\el_{\epsilon}(\mathbf{X})=\tilde{\F}^\el_{\epsilon}\circ \pe^\pl$ in the compatible region, where the tilde is used to denote quantities in the image of $\pe^\pl$, i.e. $\tilde{\mathbf{X}}=\pe^\pl(\mathbf{X})$.

Finally, for easiness in the notation, the subindex $k$ will be reserved for dislocations, $j$ for segments of the jump set, and $l$ for domains in the compatible regions away from the dislocation cores.

\begin{figure}
\begin{center}
    {\includegraphics[width=0.75\textwidth]{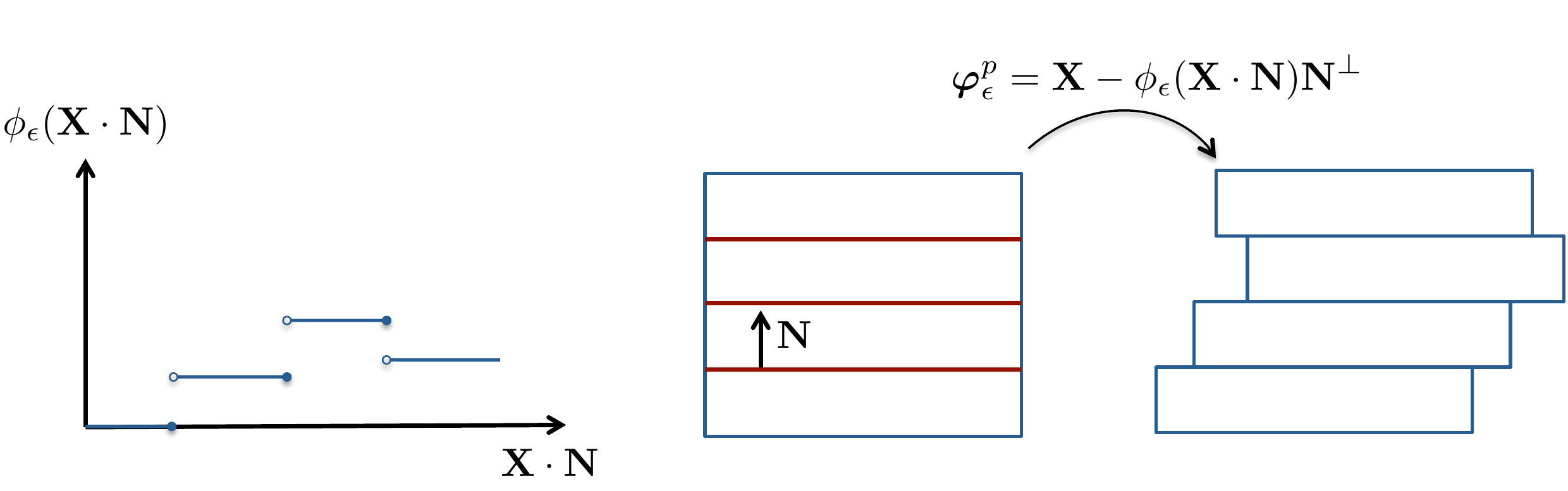}}
    \caption[]{Example of a single slip plastic deformation $\boldsymbol \varphi^\pl_{\epsilon }(\mathbf{X})  = \mathbf{X} - \phi_{\epsilon} (\mathbf{X} \cdot \mathbf{N}) \mathbf{N}^{\bot}$.
    The value of $\phi_\epsilon$ at the jump point is arbitrary, but uniquely defined, so that $\pe^\pl$ is bijective.}
    \label{Fig:Jumps}
\end{center}
\end{figure}

\section{Setting of the problem} \label{Sec:ProblemSetting}
Let $\Omega \subset \mathbb{R}^2$ be an open bounded set with Lipschitz boundary, describing the reference configuration of a defect-free crystalline material. For every $\epsilon > 0$, we consider deformation mappings $\boldsymbol \varphi_{\epsilon}:\Omega \rightarrow \boldsymbol \varphi_{\epsilon}(\Omega) \subset \mathbb{R}^2$, $\boldsymbol \varphi_{\epsilon} \in X_{\epsilon}$ defined below, with finite energy, described by the following functional
\begin{equation}\label{Eq:TotalEnergy}
E_{T\epsilon}(\pe) = E_{\epsilon}(\pe) - \int_{ \Omega} \mathbf{f} \cdot \boldsymbol \varphi_{\epsilon} \, dX - \int_{\partial \Omega} \mathbf{t} \cdot \boldsymbol \varphi_{\epsilon} \,d\mathcal H^1,
\end{equation}
where $\mathbf{f}$ are body forces, $\mathbf{t}$ are the prescribed tractions at the boundary 
(body forces and tractions are assumed to be a system of zero net force and momentum for static equilibrium) 
%($\int_{ \Omega} \mathbf{f} \, dX + \int_{\partial \Omega} \mathbf{t} \, d\mathcal H^1 = 0$ for static equilibrium)  
and
\begin{equation}\label{Eq:ElasticEnergy}
E_{\epsilon}(\pe) = \int_{\Omega\setminus \bigcup \limits_k^{N_{\epsilon}} B_{c \epsilon}(\mathbf{X}_k)} W^\el(\F^\el_{\epsilon}) \, dX + \alpha \int_{\bigcup \limits_k^{N_{\epsilon}} B_{c \epsilon}(\mathbf{X}_k)} W^\el(\F^\el_{\epsilon}) \, dX + \beta  |D\F^\el_{\epsilon}|(\Omega), 
\end{equation}
where $\alpha, \beta >0$  are material parameters and $\F^\el_{\epsilon}$ is the elastic deformation tensor. We note that only the energy terms associated to $E_{\epsilon}$ are needed for the proofs, and that additional terms may be added to Eq.~\eqref{Eq:TotalEnergy} without affecting the results. In particular, a term dependent on the plastic deformation tensor $\F^\pl_{\epsilon}$ is physically justified to provide an energetic compromise between elastic and plastic mechanisms for achieving a deformation compliant with the external loads.

The elastic energy functional given by Eq.~\eqref{Eq:ElasticEnergy} is composed of three terms. The first two represent the elastic energy outside and inside the dislocation cores respectively, where each core is considered to be a ball of radius $c\epsilon$,  around the dislocation points $\mathbf{X}_k$. As usual, we consider continuous elastic energy densities $W^\el(\F^\el_{\epsilon})$ with quadratic growth, i.e.,
\begin{equation} \label{Eq:We_growth}
 \frac{1}{C_L} |\F^\el|^2-C_L \leq W^\el(\F^\el) \leq  C_U |\F^\el|^2+C_U,
\end{equation}
with $C_L,C_U>0$. Continuum mechanics is expected to fail in the core, and the parameter $\alpha$ provides the necessary rescaling of the dislocation core energy to match its physical value, as given, for instance, by atomistic simulations. However, the value of $\alpha$ will be irrelevant for the proofs and may be set to one if desired. For its part, the last term of Eq.~\eqref{Eq:ElasticEnergy} represents the total variation of the measure $D\Fee$ over the set $\Omega$. It is mainly introduced for mathematical reasons, and it is needed to obtain $\F_{\epsilon}\approx \Fee\Fpe$ in the domain, see Proposition \ref{Thm::General_FeFp}. It can be physically interpreted as a non-local elastic deformation term that prevents, for instance, the formation of infinitely thin laminated microstructure with alternate elastic deformations. This type of energy terms are common in the so-called `gradient elasticity' theories \citep{Aifantis2009}.

The functional space $X_{\epsilon}$, defined below, includes the kinematic restrictions of slip described in Section \ref{Sec:Description}, the crystallographic conditions and scalings discussed in Section \ref{Sec:Scaling}, and considers two additional constraints that attend to mathematical simplicity. In particular, dislocations are considered to be separated from each other and from the boundary of the domain by a distance of at least several atomic spacings, and the sequence of potentially activated plastic slip systems will be considered fixed with all slip systems distinct. Although this last assumption may appear artificial and indeed does not account for the full set of compatible elastoplastic deformations that a material may suffer (a weaving structure common in textiles would be out of this scope), we will show later in Lemma \ref{Thm:Decomposition_TripleShear}  that an arbitrary plastic deformation tensor may be written as the composition of three simple shears with prescribed slip direction and 
order of activation. The domain is then separated into the dislocation cores $B_{c\epsilon}(\mathbf{X}_k)$, where the Burgers vector is spread over a ball of radius $m\epsilon<c\epsilon$ to remove the elastic singularity, and $\Omega \setminus \cup_k^{N_{\epsilon}}B_{m\epsilon}(\mathbf{X}_k)$, where the deformation mapping is compatible. The regions of overlap, $B_{c\epsilon}(\mathbf{X}_k)\setminus B_{m\epsilon}(\mathbf{X}_k)$, are  locally compatible and allow a natural extension of the definition of $\F^\pl$ everywhere inside the core region. As will be shown in the following, the specific details of the 
deformation inside each of these cores as well as the definition of $\F^\pl$ in these regions will be irrelevant in the continuum limit, justifying the use of specific functional forms for the deformation in the cores.

In the following we give a precise definition of the space $X_{\epsilon}$. 
It depends on three global constants independent of $\epsilon$, namely, the number of distinct slip systems $N_s$, the normalized maximum length of the Burgers vector for an individual slip $|\mathbf{b}_{\epsilon} / \epsilon |\leq n'$, % ($n$ is chosen as $n=n'N_s$), 
%the normalized core radius $c$
and the scaling associated to the number of dislocations $C$, c.f.~Eq.~\eqref{Eq:ScalingDislocations}.
We define
\begin{equation} \label{Def:Xe}
X_{\epsilon} =\{\boldsymbol \varphi_{\epsilon}  \in SBV(\Omega; \mathbb{R}^2): \boldsymbol \varphi_{\epsilon}, \F^\el_{\epsilon}, \F^\pl_{\epsilon} \text{ satisfy conditions (C1)--(C4)}\}
\end{equation}
for some global constants $n', C, N_s>0$ and some $N_s$ global unit vectors $\mathbf{N}_1$, ... , $\mathbf{N}_{\nu}, ..., \mathbf{N}_{N_s}$, all distinct.  The constants $c,m,L,n>0$ will be defined below from $n',C,N_s$.

\begin{enumerate}
\item[(C1)] There are $N_{\epsilon}$ dislocation points $(\mathbf{X}_k)_k$, with
\begin{equation} \label{Eq:Scaling}
N_{\epsilon} \leq \frac{C}{ \epsilon},
\end{equation}
which describe the dislocations. They are well-separated in the sense that
\begin{equation} \label{Eq:DislocationsSeparation}
\begin{split}
& |\mathbf{X}_{k_1}-\mathbf{X}_{k_2}| > 3 c\epsilon , \quad \forall k_1 \neq k_2, \\
& \text{dist}(\mathbf{X}_k,\partial \Omega) > 2c \epsilon,  \quad \forall k=1, ..., N_{\epsilon},
\end{split}
\end{equation}
where $c=m+ 2L^2$, with $m=2N_s n'$ and $L>0$ as defined later in Lemma \ref{Lemma:up_inverse_Lipschitz}. The value of $m$ is discussed in (C3) and the value of $c$ results from Lemma \ref{Lemma:UniqueDecompositionVarphip}. 

The jump set in each dislocation core $\mathbf{X}_k$, $\mathcal{J} \cap B_{c\epsilon}(\mathbf{X}_k)$, consists of $N_{dk}\leq N_s$ distinct straight segments joining $\mathbf{X}_k$ with a point in $\partial B_{c \epsilon} (\mathbf{X}_k)$.
\item[(C2)] For every convex domain $ \omega \subset \mathbb{R}^2 \setminus \bigcup_{k=1}^{N_{\epsilon}}B_{m\epsilon}(\mathbf{X}_k)$, $\boldsymbol \varphi_{\epsilon}$ 
is a compatible deformation in $\omega\cap\Omega$, in the sense that there are two maps
$\boldsymbol\varphi^\el_{\epsilon}$ and $\boldsymbol \varphi^\pl_{\epsilon}$ such that
\begin{equation}
\boldsymbol \varphi_{\epsilon}(\mathbf{X}) = \left(\boldsymbol \varphi^\el_{\epsilon} \circ \boldsymbol \varphi^\pl_{\epsilon}\right)(\mathbf{X}) \quad \forall \mathbf{X} \in \omega \cap \Omega,
\end{equation}
where  $\boldsymbol \varphi^\el_{\epsilon}:\mathbb R^2\to\mathbb R^2$ is Lipschitz continuous and one-to-one
and $\boldsymbol \varphi^\pl_{\epsilon}:\mathbb R^2\to\mathbb R^2$ can be written as the composition of  $N_s$ single slip deformations, $\boldsymbol \varphi^\pl_{\epsilon} = 
\boldsymbol \varphi^\pl_{\epsilon, N_s} \circ ... \circ \boldsymbol \varphi^\pl_{\epsilon, \nu} \circ ... \circ \boldsymbol \varphi^\pl_{\epsilon, 1}$, where each $\boldsymbol \varphi^\pl_{\epsilon, \nu}$ is of the form
\begin{equation}%\label{eqvarphieps1scalar}
\boldsymbol \varphi^\pl_{\epsilon,\nu}(\mathbf{X}) = \mathbf{X} - \phi_{\epsilon,\nu} (\mathbf{X} \cdot \mathbf{N}_\nu) \mathbf{N}_\nu^{\bot}\,.
\end{equation}
Here $\mathbf{N}_\nu$ is defined after \eqref{Def:Xe} and
each function $\phi_{\epsilon,\nu}:\mathbb{R}\rightarrow \mathbb{R}$ is piecewise constant,
with discontinuity points separated by at least $\epsilon$ and jumps
no larger than $n'\epsilon$ (see Fig.~\ref{Fig:Jumps} for an illustration).
This implies that
\begin{equation}
\begin{split}
&\boldsymbol \varphi^\pl_{\epsilon, \nu} \in SBV_{\operatorname{loc}}(\mathbb{R}^2; \mathbb{R}^2),\\
&\mathcal{J}_{\boldsymbol \varphi^\pl_{\epsilon ,\nu}}  \text{consists of parallel straight infinite lines with normal } \mathbf{N}_{\nu} \text{ separated at least by distance }\epsilon,\\
& D\boldsymbol \varphi^\pl_{\epsilon, \nu} = \mathbf{I}\ \mathcal{L}^2 +  \mathbf{b}_{\epsilon \nu } \otimes \mathbf{N}_{\nu} \mathcal{H}^1 \lfloor_{\mathcal{J}_{\boldsymbol \varphi^\pl_{\epsilon, \nu}}},  \quad  \mathbf{b}_{\epsilon \nu } \cdot \mathbf{N}_{\nu} = 0, \quad \mathbf{b}_{\epsilon \nu } \in \mathbb{R}^2\ \text{ constant on each line}, |\mathbf{b}_{\epsilon \nu}| \leq n'\epsilon\,.
\end{split}
\end{equation}

Let $\omega'=\{\mathbf X\in\omega: \mathrm{dist}(\X,\partial\omega)>L^2\epsilon\}$.
Then the plastic strain $\F^\pl_{\epsilon}$ is defined by
\begin{equation}\label{eqdvarphipleps}
\F^\pl_{\epsilon}\lfloor_{\omega'} = D \boldsymbol \varphi^\pl_{\epsilon }\lfloor_{\omega'}
= \mathbf{I}\ \mathcal{L}^2\lfloor_{\omega'} +
\sum_j \mathbf{b}_{\epsilon j } \otimes \mathbf{N}_j\ \mathcal{H}^1 \lfloor_{\mathcal{J}_j\cap\omega'}\,.
\end{equation}
Here  $\mathcal J_j$ are the finitely many segments which constitute
the jump set of $\pe^\pl$ in $\omega$,
 each of which has a constant jump $\mathbf{b}_{\epsilon j}=\llbracket \boldsymbol \varphi^\pl_{\epsilon} \rrbracket$.

Note that we shall show below (Lemma \ref{Thm:Fp_uniqueness}(ii)) that
$D\pe^\pl$ is uniquely defined in $\omega'$ and that this gives a unique  definition of $\F^\pl_{\epsilon}$ from $\boldsymbol \varphi_{\epsilon}$ on $\Omega' \setminus \left( \bigcup_k^{N_{\epsilon}}B_{(m+L^2)\epsilon}(\mathbf{X}_k)\right)$,
where  $\Omega'=\{\mathbf X\in\Omega: \mathrm{dist}(\X,\partial\Omega)>L^2\epsilon\}$.

%We shall show below, c.f.~Lemma \ref{Thm:Fp_uniqueness}(ii), that $D\pe^\pl$ is uniquely defined from $\pe$ in the smaller set $\omega'=\{\mathbf X\in\omega: \mathrm{dist}(\X,\partial\omega)>L^2\epsilon\}$. The plastic strain $\F^\pl_{\epsilon}$ is then defined by
%\begin{equation}\label{eqdvarphipleps}
%\F^\pl_{\epsilon}\lfloor_{\omega'} = D \boldsymbol \varphi^\pl_{\epsilon }\lfloor_{\omega'}
%= \mathbf{I}\ \mathcal{L}^2\lfloor_{\omega'} +  
%\sum_j \mathbf{b}_{\epsilon j } \otimes \mathbf{N}_j\ \mathcal{H}^1 \lfloor_{\mathcal{J}_j\cap\omega'}\,.
%\end{equation}
%Here  $\mathcal J_j$ are the finitely many segments which constitute the jump set of $\pe^\pl$ in $\omega$, each of which has a constant jump $\mathbf{b}_{\epsilon j}=\llbracket \boldsymbol \varphi^\pl_{\epsilon} \rrbracket$.
%This gives a unique  definition of $\F^\pl_{\epsilon}$ from $\boldsymbol \varphi_{\epsilon}$ on $\Omega' \setminus \left( \bigcup_k^{N_{\epsilon}}B_{(m+L^2)\epsilon}(\mathbf{X}_k)\right)$,
%where  $\Omega'=\{\mathbf X\in\Omega: \mathrm{dist}(\X,\partial\Omega)>L^2\epsilon\}$.
\item[(C3)] For any 
dislocation $\mathbf{X}_k$, 
 there is a map $\boldsymbol \varphi^\el_{\epsilon}$,
Lipschitz continuous and one-to-one, such that
$\boldsymbol \varphi_{\epsilon}$ can be expressed as
\begin{equation}
\boldsymbol \varphi_{\epsilon}(\mathbf{X}) = \left(\boldsymbol \varphi^\el_{\epsilon} \circ \boldsymbol \varphi^{\pc}_{\epsilon}\right)(\mathbf{X}) \quad \forall \mathbf{X} \in B_{c\epsilon}(\mathbf{X}_k).
\end{equation}
The function $\boldsymbol \varphi^\pc_{\epsilon}\in SBV_\mathrm{loc}(\mathbb R^2;\mathbb R^2)$ 
is explicitly defined from the slip-line normals and the Burgers vectors. Precisely,
in each core we are given $N_{dk}$ normals to the jump set $\mathbf{N}_j$, c.f.~(C1), with slips
$\mathbf{b}_{\epsilon j}\in\mathbb R^2$, which obey  $|\mathbf{b}_{\epsilon j}| \leq n' \epsilon$
 and $\mathbf{b}_{\epsilon j}\cdot \mathbf{N}_j = 0$ as in (C2). As a matter of convention, $\mathbf{N}_j^{\perp}$ will denote the counterclockwise 90$^{\circ}$ rotation of $\mathbf{N}_j$, and the normals to the jump sets are chosen so that $\mathbf{N}_j^{\perp}$ points radially outward from the dislocation center, to the direction of the slip line, see Fig.~\ref{Fig:DislocationCore}.

Given these parameters, we set
\begin{equation} \label{Eq:u_dislo}
  \boldsymbol \varphi^{\pc}_{\epsilon}(\mathbf{X}) =  \mathbf{X} - \sum_{j=1}^{N_{dk}} \mathbf{b}_{\epsilon j}^*(\mathbf{X}-\mathbf{X}_k) 
    \frac{\theta_j(\mathbf{X}-\mathbf{X}_k)}{2\pi},
\end{equation}
where $ \mathbf{b}_{\epsilon j}^*(\mathbf{X}) $ is the smoothed radial slip, 
\begin{equation}
\mathbf{b}^*_{\epsilon j }(\mathbf{X}) = 
 \frac{\mathbf X}{|\mathbf X|}\mathbf N_j^\perp\cdot \mathbf{b}_{\epsilon j} 
 \min\left\{\frac{|\mathbf{X}|}{m \epsilon},1\right\}\,.
\end{equation}
Each function  $\mathbf{b}^*_{\epsilon j }(\mathbf{X})$ vanishes at the origin, is Lipschitz continuous and radial,
in the sense that $\mathbf{b}^*_{\epsilon j }(\mathbf{X})$ is parallel to 
$\mathbf{X}$ for all $\mathbf{X}$. Further, on the part of the slip line outside $B_{m\epsilon}$ it coincides with the 
 ``external'' slip $ \mathbf{b}_{\epsilon j}$, in the sense that
\begin{equation} 
\mathbf{b}^*_{\epsilon j }(t \mathbf{N}_j^\perp ) =   \mathbf{b}_{\epsilon j} \hskip1cm \text{ for all } t\ge m\epsilon.
\end{equation}
The angle $\theta_j:\mathbb R^2\to[-2\pi,2\pi]$ is a function
 which obeys $\theta_j(t\mathbf{e}_1)=0$ for all $t>0$, 
 \begin{equation}
  \mathbf X = |\mathbf X| (\cos\theta_j(\mathbf X), \sin\theta_j(\mathbf X)) \hskip1cm\text{ for all } \mathbf X \in\mathbb R^2\,,
 \end{equation}
and is continuous away from $[0,\infty)\mathbf{N}_j^\perp$. Clearly 
$\theta_j(t\mathbf X)=\theta_j(\mathbf X)$ for $t>0$, and $\theta_j$ is unique up to the 
value on the jump.
 
We remark that  $ \boldsymbol \varphi^{\pc}_{\epsilon}$ is a $SBV_\mathrm{loc}$ function, with jump set $\mathbf X_k+\cup_j (0,\infty) \mathbf N_j^\perp$.
Further, we recall that
\begin{equation}\label{eqassndknk}
m=2N_s n',
\end{equation}
c.f.~(C1). This conditions guarantees that $ \boldsymbol \varphi^{\pc}_{\epsilon}$ is bijective. In particular, since all 
 $\mathbf{b}^*_{\epsilon j}$ are radial we have 
 $ \boldsymbol \varphi^{\pc}_{\epsilon}( \mathbf X_k+\mathbb R\mathbf v)\subset \mathbf X_k+\mathbb R\mathbf v$ for all unit vectors $\mathbf v$.
 An explicit computation shows that, for any unit vector $\mathbf v$ and $t\in\mathbb R$, 
 \begin{equation}
  \boldsymbol \varphi^{\pc}_{\epsilon}(\mathbf X_k+t \mathbf v)=\mathbf X_k+
  \mathbf v \left[ t - \mathrm{sgn}(t) \min\left\{\frac{|t|}{m \epsilon},1\right\} \sum_{j=1}^{N_{dk}} \mathbf N_j^\perp\cdot \mathbf{b}_{\epsilon j} \frac{\theta_j(\mathbf{v})}{2\pi}
 \right]\,.
 \end{equation}
This function is guaranteed to be bijective if  $(\boldsymbol \varphi^{\pc}_{\epsilon}
-\mathbf X_k)
\cdot \mathbf{v}$ is a monotone function of $t$, or equivalently, $\partial (\boldsymbol \varphi^{\pc}_{\epsilon}-\mathbf X_k)\cdot \mathbf{v} / \partial t >0 $. Since $|\mathbf{b}_{\epsilon j} |\le n'\epsilon$, such condition will be satisfied by setting $m$ as in \eqref{eqassndknk}.

An explicit computation shows that 
\begin{equation}\label{eqnablapedid}
 |\nabla \pe^{\pc}-\mathbf{I}|\le 2\,.
\end{equation}
Further, outside the inner core $B_{m\epsilon}(\mathbf X_k)$ the jump of $ \boldsymbol \varphi^{\pc}_{\epsilon}$ has the same form as required in (C2), in the sense that
\begin{equation}
 D \boldsymbol \varphi^{\pc}_{\epsilon} \lfloor_{(\mathbb R^2\setminus B_{m\epsilon}(\mathbf X_k))}= 
\nabla \boldsymbol \varphi^{\pc}_{\epsilon}\mathcal L^2 \lfloor_{(\mathbb R^2\setminus B_{m\epsilon}(\mathbf X_k))} +  
  \sum_{j=1}^{N_{dk}} \mathbf{b}_{\epsilon j }\otimes \mathbf N_j \mathcal H^1 \lfloor_{
    (\mathbf X_k+(m\epsilon,\infty) \mathbf N_j^\perp)}\,.
\end{equation}

\item[(C4)]  The total deformation gradient and its elastic and plastic parts are respectively defined as
\begin{align} \label{Eq:F_def}
&\F_{\epsilon} = D\boldsymbol \varphi_{\epsilon} \quad \in \mathcal{M}\left(\Omega;\mathbb{R}^{2 \times 2}\right), \\ \label{Eq:Fe_epsilon}
&\F^\el_{\epsilon} = \nabla \boldsymbol \varphi_{\epsilon} \quad \in L^1\left(\Omega;\mathbb{R}^{2 \times 2}\right), \\ \label{Eq:Fp_epsilon}
&\F^\pl_{\epsilon} = \mathbf{I}\ \mathcal{L}^2 + \sum_j \mathbf{b}_{\epsilon j} \otimes \mathbf{N}_j \  \mathcal{H}^1 \lfloor_{\mathcal{J}_j\cap \Omega'} \quad \in \mathcal{M}\left(\Omega;\mathbb{R}^{2 \times 2}\right),
\end{align}
where $\Fpe$ is set as $\mathbf{I}$ on $\Omega\setminus\Omega'$ and $\Omega'$ is defined as above. By (C2) and (C3), the jumpset $\mathcal{J}$ is a union of segments, i.e.~$\mathcal{J}=\cup_j \mathcal{J}_j$, each with a constant vector $\mathbf{b}_{\epsilon j}$. The value of $\mathbf{b}_{\epsilon j}$ inside the core is taken for each  $j$ equal to that of the corresponding segment in the compatible region $B_{c\epsilon}(\mathbf{X}_k)\setminus B_{m\epsilon}(\mathbf{X}_k)$. Then Curl $\F^\pl_{\epsilon}$ in $\Omega'$ satisfies, c.f.~\cite{ReinaConti2014},
\begin{equation}
\text{Curl }\F^\pl_{\epsilon}\lfloor_{\Omega'} = \sum_j \pm \mathbf{b}_{\epsilon j}\ \mathcal{H}^0\lfloor_{\left(\partial \mathcal{J}_j\right)\cap \Omega'},
\end{equation}
where the $+$ sign corresponds to the endpoint of segment $\mathcal{J}_j$ pointed by $\mathbf{N}^{\bot}$, and the $-$ sign to the opposite endpoint. The support of Curl $\F^\pl_{\epsilon}\lfloor_{\Omega'}$ thus corresponds to the $N_{\epsilon}$ dislocation points.

The fields $\F^\el_{\epsilon}$ that result in a finite energy belong to the space $BV\left(\Omega;\mathbb{R}^{2 \times 2}\right)$  and therefore the traces $\F^{\el+}_{\epsilon}$ and $\F^{\el-}_{\epsilon}$ on each side of the jump set $\mathcal{J}$ are well-defined quantities by Theorem 3.77 of \cite{AmbrosioFuscoPallara2000}. We then define $\F^\el_{\epsilon}$ on the jump set $\mathcal{J}$ as 
\begin{equation}\label{Eq:Fe_J}
\F^\el_{\epsilon}(\mathbf{X}) = \frac{\F^{\el+}_{\epsilon}(\mathbf{X})+\F^{\el-}_{\epsilon}(\mathbf{X})}{2} \qquad \forall\ \mathbf{X} \in \mathcal{J},
\end{equation}
although the precise definition will not affect the limiting results. 
\end{enumerate}

We first show that $\pe^\pl$ is approximately Lipschitz.
\begin{lemma} \label{Lemma:up_inverse_Lipschitz}
Let $\boldsymbol \varphi^\pl_{\epsilon}:\mathbb{R}^2 \rightarrow \mathbb{R}^2$ be a pure-slip deformation as defined in (C2), and set
${L}=2^{N_s}(1+n')^{N_s}$. Then the following holds:
\begin{enumerate}
 \item[(i)] $\boldsymbol \varphi^\pl_{\epsilon}$ is invertible, in the sense that there is
 $\boldsymbol \varphi^{\pl,{-1}}_{\epsilon}:\mathbb{R}^2 \rightarrow \mathbb{R}^2$ such that $\boldsymbol \varphi^{\pl,{-1}}_{\epsilon}(\boldsymbol \varphi^\pl_{\epsilon}(\mathbf{X}))=
 \boldsymbol \varphi^{\pl}_{\epsilon}( \boldsymbol \varphi^{\pl,{-1}}_{\epsilon}(\mathbf{X}))
 = \mathbf{X}$ for all $\mathbf{X}\in\mathbb{R}^2$. 
\item[(ii)]  For all ${\mathbf{X}}, {\mathbf{Y}} \in \mathbb{R}^2$ 
satisfying $|{\mathbf{X}}-{\mathbf{Y}}|\geq  \epsilon$ one has
\begin{equation}\label{eqvarphiplalmostlip}
|\boldsymbol \varphi^{\pl}_{\epsilon}({\mathbf{X}})-\boldsymbol \varphi^{\pl}_{\epsilon}({\mathbf{Y}})| \leq {L} |{\mathbf{X}}-{\mathbf{Y}}|
\end{equation}
and, similarly, for all ${\tilde{\mathbf{X}}}, \tilde{\mathbf{Y}} \in \mathbb{R}^2$ 
satisfying $|\tilde{\mathbf{X}}-\tilde{\mathbf{Y}}|\geq  \epsilon$,
\begin{equation}\label{eqvarphiplinvalmostlip}
|\boldsymbol \varphi^{\pl,{-1}}_{\epsilon}(\tilde{\mathbf{X}})-\boldsymbol \varphi^{\pl,{-1}}_{\epsilon}(\tilde{\mathbf{Y}})| \leq {L} 
|\tilde{\mathbf{X}}-\tilde{\mathbf{Y}}|.
\end{equation}
\item[(iii)] \label{Lemma:up_inverse_Lipschitzepsball}  For any  $\mathbf X\in\mathbb R^2$  and any  $\tilde{\mathbf{X}}\in\mathbb R^2$ one has
\begin{align}
&\boldsymbol \varphi^\pl_{\epsilon}(B_{\epsilon}(\mathbf X)) \subset {B}_{L\epsilon}(\boldsymbol\varphi^\pl_\epsilon(\mathbf X)),  \\
&\boldsymbol \varphi^{\pl,-1}_{\epsilon}(B_{\epsilon}(\tilde{\mathbf X})) \subset {B}_{L\epsilon}(\boldsymbol\varphi^{\pl,-1}_\epsilon(\tilde{\mathbf X})).
\end{align}
\item[(iv)] 
Setting $n=(1+n')^{N_s}$ one has, for all $\mathbf X$ in the jump set of $\boldsymbol \varphi^\pl_{\epsilon}$,
\begin{equation} \label{Eq:b_bound}
|\llbracket \boldsymbol \varphi^\pl_{\epsilon} \rrbracket|(\mathbf{X}) \leq n\epsilon \,.
\end{equation}
\end{enumerate}
\end{lemma}
Here and below we use the overhead tilde  to denote quantities in the intermediate (compatible) configuration $\boldsymbol \varphi^\pl_{\epsilon}(\mathbb R^2)$.

\begin{proof}
Consider a single slip deformation $\boldsymbol \varphi^\pl_{\epsilon, 1}$, c.f.~(C2). By construction, it can be written as 
\begin{equation}\label{eqvarphieps1scalar}
\boldsymbol \varphi^\pl_{\epsilon, 1 }(\mathbf{X}) = \mathbf{X} - \phi_{\epsilon} (\mathbf{X} \cdot \mathbf{N}) \mathbf{N}^{\bot}\,, 
\end{equation}
 with $\phi_{\epsilon}:\mathbb{R}\rightarrow \mathbb{R}$ piecewise constant, see Fig.~\ref{Fig:Jumps}.
Further, the discontinuity points have a separation of at least $\epsilon$ from each other, and the jumps of $\phi_\epsilon$ are no larger than $n'\epsilon$. 
This implies $|\phi_{\epsilon}(s)-\phi_{\epsilon}(t)|\le n'\epsilon \#(\mathcal{J}_{\phi_{\epsilon}}\cap [s,t])\le 
n'  (|t-s|+\epsilon)$ and therefore
\begin{equation}\label{eq1epsl}
  |\boldsymbol \varphi^\pl_{\epsilon, 1 }(\mathbf{X})-\boldsymbol \varphi^\pl_{\epsilon, 1 }(\mathbf{Y})|+\epsilon\le
  (1+n') ( |\mathbf{X}-\mathbf{Y}| + \epsilon) \hskip5mm \text{ for all } \mathbf{X},\mathbf{Y}\in\mathbb{R}^2\,.
\end{equation}
The same holds for all $\boldsymbol \varphi^\pl_{\epsilon, \nu }$. Taking the composition we obtain
\begin{equation}\label{eqcarphipilip}
  |\boldsymbol \varphi^\pl_{\epsilon  }(\mathbf{X})-\boldsymbol \varphi^\pl_{\epsilon }(\mathbf{Y})|+\epsilon\le
  (1+n')^{N_s} ( |\mathbf{X}-\mathbf{Y}| + \epsilon) \hskip5mm \text{ for all } \mathbf{X},\mathbf{Y}\in\mathbb{R}^2\,.
\end{equation}
Setting ${L}=2^{N_s}(1+n')^{N_s}$ concludes the proof of (\ref{eqvarphiplalmostlip}).
Taking the limit $\mathbf Y\to\mathbf X^+$, $\mathbf X\to \mathbf X^-$ proves (\ref{Eq:b_bound}).

To prove (i) we observe that  (\ref{eqvarphieps1scalar}) easily implies  that $\boldsymbol \varphi^\pl_{\epsilon, 1 }$ is invertible,
with inverse
\begin{equation}
\boldsymbol \varphi^{\pl,-1}_{\epsilon, 1 }(\tilde{\mathbf{X}}) = \tilde{\mathbf{X}} + \phi_{\epsilon} (\tilde{\mathbf{X}} \cdot \mathbf{N}) \mathbf{N}^{\bot}\,,
\end{equation}
where  $\phi_\epsilon$ is the same function as in (\ref{eqvarphieps1scalar}). Then $\boldsymbol \varphi^{\pl,-1}_{\epsilon}$ is defined as the composition of the inverses
 of the single slips, and the estimate
\begin{equation}\label{eqcarphipinvlip}
  |\boldsymbol \varphi^{\pl,-1}_{\epsilon  }(\tilde {\mathbf{X}})-\boldsymbol \varphi^{\pl,-1}_{\epsilon }(\tilde{\mathbf{Y}})|+\epsilon\le
  (1+n')^{N_s} ( |\tilde{\mathbf{X}}-\tilde{\mathbf{Y}}| + \epsilon) \hskip5mm \text{ for all } \tilde{\mathbf{X}},\tilde{\mathbf{Y}}\in\mathbb{R}^2
\end{equation}
 follows. This concludes the proof of (i) and (ii). Finally, assertion (iii) follows immediately from 
 (\ref{eqcarphipilip}) and  (\ref{eqcarphipinvlip}).
 \end{proof}

Now we prove that $\F_\epsilon^\pl$ is well defined.
\begin{lemma} \label{Thm:Fp_uniqueness}
Let $\boldsymbol \varphi_{\epsilon}$, $\boldsymbol \varphi^\el_{\epsilon}$, $\boldsymbol \varphi^\pl_{\epsilon}$ and $\F^\pl_{\epsilon}$ as in (C2-C4). Then the following results hold:

(i) The measure $\F^\pl_{\epsilon}$ is uniquely defined from $\boldsymbol \varphi_\epsilon$.

(ii) If 
$\mathbf Y$ is such that $\omega=B_{L^2\epsilon}(\mathbf Y)\subset\Omega \setminus \cup_k B_{m\epsilon}(\mathbf{X}_k)$,
and  $\boldsymbol \varphi_{\epsilon}=\boldsymbol \varphi^\el_{\epsilon} \circ \boldsymbol \varphi^\pl_{\epsilon}$,
$\boldsymbol \varphi_{\epsilon}=\hat{\boldsymbol \varphi}^\el_{\epsilon} \circ \hat{\boldsymbol \varphi}^\pl_{\epsilon}$
are two decompositions as in (C2) on the domain $\omega$, then on the smaller set $B_{\epsilon}(\mathbf Y)$
they coincide up to a translation,
 in the sense that  there is $\mathbf c\in\mathbb R^2$ such that
 $\hat{\boldsymbol \varphi}^\pl_{\epsilon}(\mathbf X)=
{\boldsymbol \varphi}^\pl_{\epsilon}(\mathbf X)-\mathbf c$ for $\mathbf X\in B_{\epsilon}(\mathbf Y)$ and
$\hat{\boldsymbol \varphi}^\el_{\epsilon}(\mathbf{ \tilde X})=
{\boldsymbol \varphi}^\el_{\epsilon}(\mathbf {\tilde X}+\mathbf c)$ for $\mathbf {\tilde X}\in \boldsymbol \varphi^\pl_{\epsilon}(B_{\epsilon}(\mathbf Y))$.
\end{lemma}
 
\begin{proof}
We first show that the second assertion implies the first. 
Recall that $\mathbf F_\epsilon^\pl=\mathbf{I}$ outside 
$\Omega'=\{\mathbf{X}\in \Omega : \text{dist}\left(\mathbf{X},\partial \Omega \right)>L^2\epsilon\}$.
Every $\mathbf Y\in \Omega'\setminus\cup_k B_{(m+L^2)\epsilon}(\mathbf{X}_k)$ has the property stated in (ii), 
therefore  $\mathbf F_\epsilon^\pl$ is unique on this set.
Further, this shows that it is unique in the outer region $B_{c\epsilon}(\mathbf X_k)
\setminus B_{(m+L^2)\epsilon}(\mathbf{X}_k)$ of each core, and therefore by (C3) also in the interior. This concludes the proof of (i).

We now turn to the second assertion. Let $\mathbf Y$ be as stated.
By the previous Lemma, $\pe^\pl\left( B_{\epsilon}(\mathbf{Y}) \right) \subset B_{L\epsilon}\left(\pe^\pl(\mathbf{Y}) \right)$ and $\pe^{\pl,-1}\left( B_{L\epsilon}\left(\pe^\pl(\mathbf{Y}) \right) \right) \subset B_{L^2\epsilon}(\mathbf{Y})$, c.f.~Fig.~\ref{Fig:Unique_sequence}.
We shall show that $D\boldsymbol \varphi^\pl_{\epsilon}=D\hat{\boldsymbol \varphi}^\pl_{\epsilon}$ on $B_\epsilon(\mathbf Y)$  by showing that $\mathcal{J}_{\boldsymbol \varphi^\pl_{\epsilon}}$ and $\mathbf{b}_{\epsilon}$ are uniquely defined from $\boldsymbol \varphi_{\epsilon}$ on such domain.

\textbf{Uniqueness of $\mathcal{J}_{\boldsymbol \varphi^\pl_{\epsilon}}$.}
By definition, $\boldsymbol \varphi_{\epsilon}=\boldsymbol \varphi^\el_{\epsilon}\circ \boldsymbol 
\varphi^\pl_{\epsilon}$, with $\boldsymbol \varphi^\pl_{\epsilon}\in SBV_{\mathrm{loc}}(\mathbb{R}^2;\mathbb{R}^2)$ and $\boldsymbol \varphi^\el_{\epsilon}$ Lipschitz continuous. Therefore, by the chain rule, c.f.~Theorem 3.96 of \cite{AmbrosioFuscoPallara2000}, and since $\nabla \pe^\pl =\mathbf{I}$,
\begin{equation}
D\boldsymbol \varphi_{\epsilon} = \nabla \boldsymbol \varphi^\el_{\epsilon} \circ \boldsymbol \varphi^\pl_{\epsilon} \ \mathcal{L}^2+\left(\boldsymbol \varphi^\el_{\epsilon}(\boldsymbol \varphi^{\pl+}_{\epsilon})- \boldsymbol \varphi^\el_{\epsilon}(\boldsymbol \varphi^{\pl-}_{\epsilon}) \right) \otimes \mathbf{N}\ \mathcal{H}^1\lfloor_{\mathcal{J}_{\boldsymbol \varphi^\pl_{\epsilon}}},
\end{equation}
and $\mathcal{J}_{\boldsymbol \varphi_{\epsilon}} \subseteq \mathcal{J}_{\boldsymbol \varphi^\pl_{\epsilon}}$. For $\mathbf{X} \in \mathcal{J}_{\boldsymbol \varphi^\pl_{\epsilon}}$, $\boldsymbol \varphi^{\pl+}_{\epsilon}(\mathbf{X})\neq \boldsymbol \varphi^{\pl-}_{\epsilon}(\mathbf{X})$ and since $\boldsymbol \varphi^\el_{\epsilon}$ is one-to-one, then $\boldsymbol \varphi^\el_{\epsilon}\left(\boldsymbol \varphi^{\pl+}_{\epsilon}(\mathbf{X})\right) \neq \boldsymbol \varphi^\el_{\epsilon}\left(\boldsymbol \varphi^{\pl-}_{\epsilon}(\mathbf{X})\right)$. 
The function 
$\boldsymbol \varphi_{\epsilon}$ is then not approximately continuous in $\mathbf{X}$ and $\mathcal{H}^1\left(\mathcal{J}_{\boldsymbol \varphi^\pl_{\epsilon}}\setminus \mathcal{J}_{\boldsymbol \varphi_{\epsilon}} \right)=\emptyset$. 
Therefore, $\mathcal{J}_{\boldsymbol \varphi^\pl_{\epsilon}}  = \mathcal{J}_{\boldsymbol \varphi_{\epsilon}}$ and $\mathcal{J}_{\boldsymbol \varphi^\pl_{\epsilon}}$ is uniquely defined from $\boldsymbol \varphi_{\epsilon}$, up to null sets.

\textbf{Uniqueness of $\mathbf{b}_{\epsilon}$.}
We consider two decompositions $\boldsymbol \varphi_{\epsilon}= \boldsymbol \varphi^\el_{\epsilon} \circ
\boldsymbol \varphi^\pl_{\epsilon} =  \hat{\boldsymbol \varphi}^\el_{\epsilon} \circ \hat{\boldsymbol \varphi}^\pl_{\epsilon}$
in $B_{L^2\epsilon}(\mathbf Y)$.
We write $B_{L^2\epsilon}(\mathbf{Y})$ as the union of finitely many 
non overlapping connected open domains $\mathcal{B}$, i.e.~$\omega = \cup \bar{\mathcal{B}}$, such that $\mathcal{B}\cap \mathcal{J} = \emptyset$. For each domain $\mathcal{B}$, $\boldsymbol \varphi^\pl_{\epsilon}$ and $\hat{\boldsymbol \varphi}^\pl_{\epsilon}$ can be written, by construction, as $\boldsymbol \varphi^\pl_{\epsilon}(\mathbf{X})=\mathbf{X}+\mathbf{c}_B$ and $\hat{\boldsymbol \varphi}^\pl_{\epsilon}(\mathbf{X}) = \mathbf{X} + \hat{\mathbf{c}}_{\mathcal{B}}$, with $\mathbf{c}_{\mathcal{B}}$ and $\hat{\mathbf{c}}_{\mathcal{B}}$ constant vectors. Then,
\begin{align}
\boldsymbol \varphi_{\epsilon}(\mathbf{X})=\left(\boldsymbol \varphi^\el_{\epsilon} \circ \boldsymbol \varphi^\pl_{\epsilon} \right)(\mathbf{X}) =  \left(\hat{\boldsymbol \varphi}^\el_{\epsilon} \circ \hat{\boldsymbol \varphi}^\pl_{\epsilon} \right)(\mathbf{X}) =\boldsymbol \varphi^\el_{\epsilon}\left(\mathbf{X}+\mathbf{c}_{\mathcal{B}} \right) = \hat{\boldsymbol \varphi}^\el_{\epsilon}\left(\mathbf{X}+\hat{\mathbf{c}}_{\mathcal{B}} \right).
\end{align}

Let $\mathbf{d}_{\mathcal{B}}=\mathbf{c}_{\mathcal{B}}-\hat{\mathbf{c}}_{\mathcal{B}}$. We proceed to show that all blocks $\mathcal{B}$ such that $\mathcal{B}\cap B_{\epsilon}(\mathbf{Y})\neq \emptyset$ have the same $\mathbf{d}_{\mathcal{B}}$. For that purpose, consider two blocks $\mathcal{B}$ and $\mathcal{B}'$ such that they are in contact in the intermediate configuration defined by $\pe^\pl$, i.e.~$\left(\bar{\mathcal{B}}+\mathbf{c}_{\mathcal{B}} \right)\cap \left( \bar{\mathcal{B}'}+\mathbf{c}_{\mathcal{B}'}\right) \neq \emptyset$.
Then, there exists $\mathbf{X}\in \bar{ \mathcal{B}}$, $\mathbf{X}'\in \bar{ \mathcal{B}'}$, such that $\mathbf{X}+\mathbf{c}_{\mathcal{B}} = \mathbf{X}'+\mathbf{c}_{\mathcal{B}'}$ and $\pe^\el\left( \mathbf{X}+\mathbf{c}_{\mathcal{B}}\right)=\pe^\el(\mathbf{X}'+\mathbf{c}_{\mathcal{B}'})$. 
%Similarly, for $\hat{\p}_{\epsilon}^\pl$, one obtains $\hat{\boldsymbol \varphi}^\el_{\epsilon}\left(\mathbf{X}+\hat{\mathbf{c}}_{\mathcal{B}} \right)= \hat{\boldsymbol \varphi}^\el_{\epsilon}\left(\mathbf{X}'+\hat{\mathbf{c}}_{\mathcal{B}'} \right)$, and since $\boldsymbol \varphi_{\epsilon}=\boldsymbol \varphi^\el_{\epsilon} \circ \boldsymbol \varphi^\pl_{\epsilon} =\hat{\boldsymbol \varphi}^\el_{\epsilon} \circ \hat{\boldsymbol \varphi}^\pl_{\epsilon}$, then 
Then, since $\boldsymbol \varphi_{\epsilon}=\boldsymbol \varphi^\el_{\epsilon} \circ \boldsymbol \varphi^\pl_{\epsilon} =\hat{\boldsymbol \varphi}^\el_{\epsilon} \circ \hat{\boldsymbol \varphi}^\pl_{\epsilon}$,
\begin{equation}
\boldsymbol \varphi^\el_{\epsilon}\left(\mathbf{X}+\mathbf{c}_{\mathcal{B}} \right) = \hat{\boldsymbol \varphi}^\el_{\epsilon}\left(\mathbf{X}+\hat{\mathbf{c}}_{\mathcal{B}} \right)=\boldsymbol \varphi^\el_{\epsilon}\left(\mathbf{X}'+\mathbf{c}_{\mathcal{B}'} \right) = \hat{\boldsymbol \varphi}^\el_{\epsilon}\left(\mathbf{X}'+\hat{\mathbf{c}}_{\mathcal{B}'} \right).
\end{equation}
Since $\hat{\boldsymbol \varphi}^\el_{\epsilon}$ is injective, $\mathbf{X}+\hat{\mathbf{c}}_{\mathcal{B}} = \mathbf{X}'+\hat{\mathbf{c}}_{\mathcal{B}'}$. Subtracting this equation from $\mathbf{X}+\mathbf{c}_{\mathcal{B}} = \mathbf{X}'+\mathbf{c}_{\mathcal{B}'}$, one obtains 
 that $\mathbf{d}_{\mathcal{B}}=\mathbf{d}_{\mathcal{B}'}$ or, equivalently, that $\mathbf{c}_{\mathcal{B}} - \mathbf{c}_{\mathcal{B}'} = \hat{\mathbf{c}}_{\mathcal{B}} -  \hat{\mathbf{c}}_{\mathcal{B}'} $ for every two pair of blocks in contact in $B_{L\epsilon}\left(\pe^\pl(\mathbf{Y}) \right)$. Since its pullback to the reference configuration contains $B_{\epsilon}(\mathbf{Y})$, one obtains $\mathbf{b}_{\epsilon} = \llbracket \boldsymbol \varphi^\pl_{\epsilon} \rrbracket = \llbracket \hat{\boldsymbol \varphi}^\pl_{\epsilon} \rrbracket = \hat{\mathbf{b}}_{\epsilon}$ on $B_{\epsilon}(\mathbf{Y})$ and the result follows.
\end{proof}

The next Lemma shows that the decomposition of $\pe^\pl$ is unique, and will be used to prove $\det \Fp=1$.
\begin{lemma}\label{Lemma:UniqueDecompositionVarphip}
Let $\pe \in X_{\epsilon}$ as defined in \eqref{Def:Xe}. Consider $\mathbf{Y}$ such that $\operatorname{dist}\left(\mathbf{Y},\partial \Omega \right)>2L^2\epsilon$, with $L$ as in Lemma \ref{Lemma:up_inverse_Lipschitz} and $B_{L^2\epsilon}(\mathbf{Y})$ satisfying $B_{L^2\epsilon}(\mathbf{Y})\cap\left(\cup_k B_{m\epsilon}(\mathbf{X}_k) \right)=\emptyset$.  If $\pe^\pl$ and $\hatpe^\pl$ are two deformations on $B_{L^2\epsilon}(\mathbf{Y})$, satisfying $\F^\pl_{\epsilon}= D\pe^\pl=D\hatpe^\pl$, each with its decomposition of the form
 \begin{equation} \label{Eq:Double_definition}
 \begin{split}
 &\boldsymbol \varphi^\pl_{\epsilon} = 
\boldsymbol \varphi^\pl_{\epsilon, N_s} \circ ... \circ \boldsymbol \varphi^\pl_{\epsilon, \nu} \circ ... \circ \boldsymbol \varphi^\pl_{\epsilon, 1}, \\
&\hat{\boldsymbol \varphi}^\pl_{\epsilon} = 
\hat{\boldsymbol \varphi}^\pl_{\epsilon, N_s} \circ ... \circ \hat{\boldsymbol \varphi}^\pl_{\epsilon, \nu} \circ ... \circ \hat{\boldsymbol \varphi}^\pl_{\epsilon, 1},
\end{split}
\end{equation}
then, there are constant vectors $\mathbf{c}_{\nu}$, $\nu=0, .., N_s$ such that
\begin{equation} \label{Eq:ResultDecomposition}
\begin{split}
&\boldsymbol \varphi^{\pl}_{\epsilon, \nu}(\mathbf{X}) = \hat{\boldsymbol \varphi}^{\pl}_{\epsilon, \nu} (\mathbf{X}-\mathbf{c}_{\nu-1})+\mathbf{c}_{\nu},\\ 
&D\boldsymbol \varphi^{\pl}_{\epsilon, \nu}\left(\boldsymbol \varphi^{\pl}_{\epsilon, \nu-1} \circ ... \circ \boldsymbol \varphi^{\pl}_{\epsilon, 1} \right)=D\hat{\boldsymbol \varphi}^{\pl}_{\epsilon, \nu} \left( \hat{\boldsymbol \varphi}^{\pl}_{\epsilon, \nu-1} \circ ... \circ \hat{\boldsymbol \varphi}^{\pl}_{\epsilon, 1} \right)
\end{split}
\end{equation}
in  $B_{\epsilon}(\mathbf{Y})$.
\end{lemma}

\begin{proof}
Consider two possible plastic deformation mappings $\pe^\pl$, $\hat{\boldsymbol \varphi}^\pl_{\epsilon}$ on $B_{L^2\epsilon}(\mathbf{Y})$ and their decomposition as in \eqref{Eq:Double_definition}. By Lemma \ref{Thm:Fp_uniqueness}, $\F^\pl_{\epsilon}=D\boldsymbol \varphi^\pl_{\epsilon}=D\hat{\boldsymbol \varphi}^\pl_{\epsilon}$ on $B_{L^2\epsilon}(\mathbf Y)$, and by integration one obtains $\boldsymbol \varphi^\pl_{\epsilon} =\hat{\boldsymbol \varphi}^\pl_{\epsilon} +\mathbf{c}_{N_s}$, where $\mathbf{c}_{N_s}$ corresponds to a rigid translation. Define $\bar{\boldsymbol \varphi}^\pl_{\epsilon, N_s}(\mathbf{X})= \hat{\boldsymbol \varphi}^\pl_{\epsilon, N_s}(\mathbf{X})+\mathbf{c}_{N_s}$. Then,
\begin{equation}
\boldsymbol \varphi^\pl_{\epsilon, N_s} \circ  \boldsymbol \varphi^\pl_{\epsilon, N_s-1} \circ... \circ \boldsymbol \varphi^\pl_{\epsilon, 1} =  \bar{\boldsymbol \varphi}^\pl_{\epsilon, N_s} \circ \hat{\boldsymbol \varphi}^\pl_{\epsilon, N_s-1}\circ ... \circ  \hat{ \boldsymbol \varphi}^\pl_{\epsilon, 1} \quad \text{on } B_{L^2\epsilon}(\mathbf{Y}).
\end{equation}

\begin{figure} 
\begin{center}
    {\includegraphics[width=0.5\textwidth]{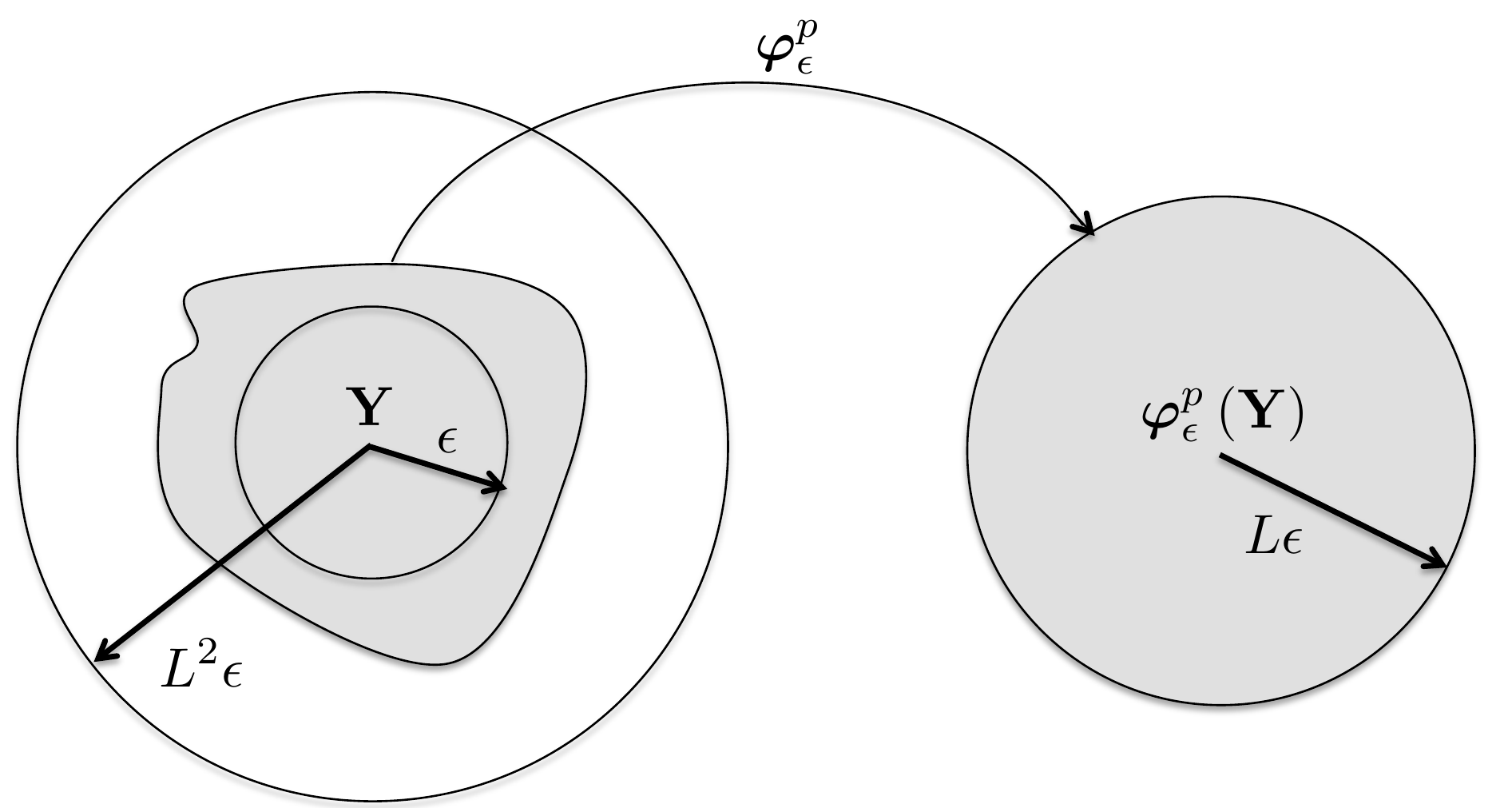}}
    \caption[]{Choice of domains to prove the uniqueness of $\Fpe$ in $B_{\epsilon}(\mathbf{Y})$, for $B_{L^2\epsilon}(\mathbf{Y}) \subset \Omega \setminus \cup_k B_{m\epsilon}(\mathbf{X}_k)$.}
    %Choice of domains for the identification of the decomposition $\boldsymbol \varphi^\pl_{\epsilon} = 
%\boldsymbol \varphi^\pl_{\epsilon, N_s} \circ ... \circ \boldsymbol \varphi^\pl_{\epsilon, \nu} \circ ... \circ \boldsymbol \varphi^\pl_{\epsilon, 1}$ up to rigid translations.}
    \label{Fig:Unique_sequence}
\end{center}
\end{figure}

The last deformation applied in the sequence corresponds to slip along jump sets with normal $\mathbf{N}_{N_s}$. Since the normals $\mathbf{N}_{\nu}$, $\nu=1, ..., N_s$ are all distinct, the jump set with normal $\mathbf{N}_{N_s}$ on configuration $\boldsymbol \varphi^\pl_{\epsilon, N_s-1} \circ ... \circ \boldsymbol \varphi^\pl_{\epsilon, 1}\left(B_{L^2\epsilon}(\mathbf{Y})\right)$ will necessarily traverse regions where $\F^\pl=\mathbf{I}$. The segments in those regions will suffer a translation upon pullback to the reference configuration, and will thus have that same normal $\mathbf{N}_{N_s}$, each with a given Burgers vector. The ensemble of segments with normal $\mathbf{N}_{N_s}$ will, by construction, form a set of straight lines in the deformed configuration with a constant Burgers vector per line, potentially different between lines. 
Thus, $D\boldsymbol \varphi^\pl_{\epsilon, N_s}$ is uniquely characterized from $\F^\pl$  on the set $\p^\pl_{\epsilon}\left(B_{L^2\epsilon} (\mathbf{Y})\right)$, which implies
$D\boldsymbol \varphi^{\pl,-1}_{\epsilon, N_s}= D\bar{\boldsymbol \varphi}^{\pl,-1}_{\epsilon, N_s}$ and
\begin{align}
&\boldsymbol \varphi^{\pl,-1}_{\epsilon, N_s}(\tilde{\mathbf{X}}) = \bar{\boldsymbol \varphi}^{\pl,-1}_{\epsilon, N_s} (\tilde{\mathbf{X}})+ \mathbf{c}_{N_s-1} = \hat{\boldsymbol \varphi}^{\pl,-1}_{\epsilon, N_s} (\tilde{\mathbf{X}}-\mathbf{c}_{N_s})+ \mathbf{c}_{N_s-1}\quad \text{on } D_{N_s},\\
& \boldsymbol \varphi^{\pl}_{\epsilon, N_s}(\mathbf{X})= \bar{\boldsymbol \varphi}^{\pl}_{\epsilon, N_s} (\mathbf{X}-\mathbf{c}_{N_s-1})= \hat{\boldsymbol \varphi}^{\pl}_{\epsilon, N_s} (\mathbf{X}-\mathbf{c}_{N_s-1})+ \mathbf{c}_{N_s} \quad \text{on } \p^{\pl,-1}_{\epsilon,N_s}(D_{N_s}),
\end{align}
where $D_{N_s}=B_{L\epsilon} (\pe^\pl(\mathbf{Y})) \subset \pe^\pl\left(B_{L^2\epsilon} (\mathbf{Y})\right)$ and $\mathbf{c}_{N_s-1}$ and $\mathbf{c}_{N_s}$ are constant vectors. We define $E_{N_s}=\pe^{\pl,-1}\left(D_{N_s} \right)$ so that
\begin{equation}
%\boldsymbol \varphi^\pl_{\epsilon, N_s} \circ  
\boldsymbol \varphi^\pl_{\epsilon, N_s-1}\circ... \circ \boldsymbol \varphi^\pl_{\epsilon, 1} = 
%\boldsymbol \varphi^\pl_{\epsilon, N_s} \circ 
\hat{\boldsymbol \varphi}^\pl_{\epsilon, N_s-1}\circ ... \circ  \hat{ \boldsymbol \varphi}^\pl_{\epsilon, 1} + \mathbf{c}_{N_s-1} \quad \text{on } E_{N_s}.
\end{equation}
Defining $\bar{\boldsymbol \varphi}^\pl_{\epsilon, N_s-1}(\mathbf{X})=\hat{\boldsymbol \varphi}^\pl_{\epsilon, N_s-1}(\mathbf{X})+\mathbf{c}_{N_s-1}$, it is readily obtained
\begin{equation}
  \boldsymbol \varphi^\pl_{\epsilon, N_s-1}\circ  \boldsymbol \varphi^\pl_{\epsilon, N_s-2}\circ... \circ \boldsymbol \varphi^\pl_{\epsilon, 1}=  \bar{\boldsymbol \varphi}^\pl_{\epsilon, N_s-1}\circ \hat{\boldsymbol \varphi}^\pl_{\epsilon, N_s-2} \circ ... \circ  \hat{ \boldsymbol \varphi}^\pl_{\epsilon, 1}  \quad \text{on } E_{N_s}.
\end{equation}

The above procedure may be iterated for $\nu=N_s-1$ successively till $\nu=1$, obtaining for all $\nu$
\begin{align}
&\bar{\boldsymbol \varphi}^\pl_{\epsilon, \nu}(\mathbf{X})=\hat{\boldsymbol \varphi}^\pl_{\epsilon, \nu}(\mathbf{X})+\mathbf{c}_{\nu}, \\
&\boldsymbol \varphi^{\pl}_{\epsilon,\nu}(\mathbf{X})= \bar{\boldsymbol \varphi}^{\pl}_{\epsilon, \nu} (\mathbf{X}-\mathbf{c}_{\nu-1})  = \hat{\boldsymbol \varphi}^{\pl}_{\epsilon, \nu} (\mathbf{X}-\mathbf{c}_{\nu-1})+\mathbf{c}_{\nu} \quad \text{on } \p^{\pl,-1}_{\epsilon, \nu}(D_{\nu}),\\ \label{Eq:Derivative}
&D\boldsymbol \varphi^{\pl}_{\epsilon, \nu}(\mathbf{X})= D\bar{\boldsymbol \varphi}^{\pl}_{\epsilon, \nu} (\mathbf{X}-\mathbf{c}_{\nu-1})=D\hat{\boldsymbol \varphi}^{\pl}_{\epsilon, \nu} (\mathbf{X}-\mathbf{c}_{\nu-1}) \quad \text{on }\p^{\pl,-1}_{\epsilon, \nu}(D_{\nu}),\\
&  \boldsymbol \varphi^\pl_{\epsilon, \nu-1}\circ  \boldsymbol \varphi^\pl_{\epsilon, \nu-2}\circ... \circ \boldsymbol \varphi^\pl_{\epsilon, 1}=  \bar{\boldsymbol \varphi}^\pl_{\epsilon, \nu-1}\circ \hat{\boldsymbol \varphi}^\pl_{\epsilon, \nu-2} \circ... \circ  \hat{ \boldsymbol \varphi}^\pl_{\epsilon, 1}, \quad \text{on } E_{\nu},
\end{align}
where $D_{\nu}=B_{L\epsilon/(2+2n')^{N_s-\nu}} \left((\p^\pl_{\epsilon,\nu}\circ ... \circ \p^\pl_{\epsilon,1})(\mathbf{Y})\right)$, $E_{\nu}=\left(\p^\pl_{\epsilon,\nu}\circ .. \circ 
\p^\pl_{\epsilon,1}\right)^{-1}\left(D_{\nu} \right)$ and $\mathbf{c}_{\nu}$ are constant vectors. 
By (\ref{eq1epsl}), $D_{\nu}\subset
\p^{\pl,-1}_{\nu+1}(D_{\nu+1})$.
Note that for $\nu=2$, the last expression gives $\boldsymbol \varphi^{\pl}_{\epsilon, 1}(\mathbf{X})= \bar{\boldsymbol \varphi}^{\pl}_{\epsilon, 1} (\mathbf{X}) $, and therefore $\mathbf{c}_{0}=0$. 

The composition of two consecutive deformation mappings then gives
\begin{equation}\label{eqiteratcnu1}
\boldsymbol \varphi^{\pl}_{\epsilon, \nu}(\boldsymbol \varphi^{\pl}_{\epsilon, \nu-1}(\mathbf{X})) = \hat{\boldsymbol \varphi}^{\pl}_{\epsilon, \nu}\left(\boldsymbol \varphi^{\pl}_{\epsilon, \nu-1}(\mathbf{X})-\mathbf{c}_{\nu-1}\right) +\mathbf{c}_\nu= \hat{\boldsymbol \varphi}^{\pl}_{\epsilon, \nu}\left(\hat{\boldsymbol \varphi}^{\pl}_{\epsilon, \nu-1}(\mathbf{X}-\mathbf{c}_{\nu-2})\right) +\mathbf{c}_\nu
\end{equation}
and the composition of the first $\nu$ mappings yields
\begin{equation}\label{eqiteratcnu2}
\boldsymbol \varphi^{\pl}_{\epsilon, \nu}\circ \boldsymbol \varphi^{\pl}_{\epsilon, \nu-1} \circ ... \circ \boldsymbol \varphi^{\pl}_{\epsilon, 1} = \hat{\boldsymbol \varphi}^{\pl}_{\epsilon, \nu}\circ \hat{\boldsymbol \varphi}^{\pl}_{\epsilon, \nu-1} \circ ... \circ \hat{\boldsymbol \varphi}^{\pl}_{\epsilon, 1}(\mathbf{X}-\mathbf{c}_0) +\mathbf{c}_{\nu}\quad 
\text{on $E_{1}$},
\end{equation}
where $\mathbf{c}_0=0$. Thus, from \eqref{Eq:Derivative},
\begin{equation}
\begin{split}
D\boldsymbol \varphi^{\pl}_{\epsilon, \nu}\left(\boldsymbol \varphi^{\pl}_{\epsilon, \nu-1} \circ ...  \circ \boldsymbol \varphi^{\pl}_{\epsilon, 1} \right)&=D\hat{\boldsymbol \varphi}^{\pl}_{\epsilon, \nu} \left(\boldsymbol \varphi^{\pl}_{\epsilon, \nu-1} \circ ... \circ \boldsymbol \varphi^{\pl}_{\epsilon, 1} -\mathbf{c}_{\nu-1}\right) =D\hat{\boldsymbol \varphi}^{\pl}_{\epsilon, \nu} \left( \hat{\boldsymbol \varphi}^{\pl}_{\epsilon, \nu-1} \circ... \circ \hat{\boldsymbol \varphi}^{\pl}_{\epsilon, 1} \right) \quad \text{on } E_{1}.
\end{split}
\end{equation}

We note that $B_{\epsilon}(\mathbf{Y}) \subset E_1$, and $E_\nu \subset E_{\nu+1}$, and thus Eqs.~\eqref{Eq:ResultDecomposition} hold for  all $\nu$.
\end{proof}

\begin{lemma} \label{Lemma:SizeControl}
Let $\mathcal{J}$ be the jump set of $\pe \in X_{\epsilon}$. Then, the following results hold:
\begin{enumerate}
\item[(i)] There exists  $A>0$ such that for any ball $B_r(\mathbf X) \subset \mathbb{R}^2$ of radius $r>\epsilon$
\begin{equation} \label{Eq:NoClustering}
|  \mathcal{J}\cap  B_r(\mathbf X)| \leq \frac{A r^2}{\epsilon}.
\end{equation}
\item[(ii)] There exists $C^*>0$ depending only on $\Omega$ such that
\begin{equation} \label{Eq:Scaling_b}
|\mathcal{J}\cap \Omega| \leq \frac{C^*}{\epsilon}.
\end{equation}
\end{enumerate}
\end{lemma}
We recall that by definition $\mathcal J\subset\Omega$.
\begin{proof}
Consider first a ball $B_{\epsilon}\left(\mathbf{Y} \right)$ such that
 $B_{2L\epsilon}\left(\mathbf{Y} \right)\cap \left(\cup_k B_{m\epsilon}(\mathbf{X}_k) \right)=\emptyset$. In this ball, the representation 
 of $\pe^\pl$ from (C2) holds. Since 
 the $\p^\pl_{\epsilon,\nu}$ are piecewise translations and 
 the normals $N_\nu$ are all different, the jump set of $\pe^\pl$ is the union of the counterimages of the jump sets of 
the individual $\p^\pl_{\epsilon,\nu}$, 
\begin{equation*}
 \mathcal{J}_{\pe^\pl}=%\cap B_\epsilon (\mathbf{Y})=
 %\left (
 \mathcal{J}_{\p^\pl_{\epsilon,1}}
 \cup 
 \p^{\pl,-1}_{\epsilon,1}( \mathcal{J}_{\p^\pl_{\epsilon,2}}) 
\cup
 \dots \cup
 (\p^{\pl,-1}_{\epsilon,1}\circ \dots \circ\p^{\pl,-1}_{\epsilon,N_s-1}) ( \mathcal{J}_{\p^\pl_{\epsilon,N_s}})%\right) \cap B_\epsilon(\mathbf{Y})
 \,.
\end{equation*}
For every $\nu$, let $\mathbf Y_\nu=(\p^\pl_{\epsilon,\nu-1}\circ \dots \circ \p^\pl_{\epsilon,1})(\mathbf Y)$. Since the jump 
set of $\p^\pl_{\epsilon,\nu}$ consists of parallel lines with distance at least $\epsilon$, we have
\begin{equation*}
 | \mathcal{J}_{\p^\pl_{\epsilon,\nu}} \cap B_{L\epsilon}(\mathbf Y_\nu)| \le (2L+1)2L\epsilon 
\end{equation*}
and, recalling that $(\p^{\pl}_{\epsilon,\nu-1}\circ \dots \circ\p^{\pl}_{\epsilon,1})(B_\epsilon(\mathbf Y)) \subset
B_{L\epsilon}(\mathbf Y_\nu)$, 
\begin{equation*}
 | (\p^{\pl,-1}_{\epsilon,1}\circ \dots \circ\p^{\pl,-1}_{\epsilon,\nu-1}) (\mathcal{J}_{\p^\pl_{\epsilon,\nu}} ) \cap B_{\epsilon}(\mathbf Y)| \le (2L+1)2L\epsilon \,.
\end{equation*}
Therefore $|\mathcal{J}_{\pe^\pl}\cap B_\epsilon (\mathbf{Y})|\le N_s(2L+1)2L\epsilon $.

Consider now a generic ball  $B_r(\mathbf X)$.
We cover $B_r(\mathbf X)\setminus \left( \cup_k B_{c\epsilon}(\mathbf{X}_k) \right)$ with balls $\omega_l=B_{\epsilon}(\mathbf{Y}_l)$ with finite overlap such that $B_{2L\epsilon}(\mathbf{Y}_l)\cap \left( \cup_k B_{m\epsilon}(\mathbf{X}_k) \right) =\emptyset$. The number of those sets $\omega_l$ is bounded by $A_1 r^2/\epsilon^2$ for some constant $A_1$. Similarly, 
by (\ref{Eq:DislocationsSeparation})
the number of cores intersecting $B_r$ is bounded by $A_2 r^2/\epsilon^2$ for some constant $A_2$, and the total length of the jump set in a core is not larger than $c\epsilon N_s$. Then,
\begin{equation}
\begin{split}
|\mathcal{J}\cap B_r| &\leq \sum_l |\mathcal{J}\cap \omega_l| + \sum_k |\mathcal{J}\cap \left(B_r\cap B_{c\epsilon}(\mathbf{X}_k)\right)| \\
&\leq A_1 \frac{r^2}{\epsilon^2} N_s \left(2L+1\right) 2L\epsilon + A_2 \frac{r^2}{\epsilon^2} N_s c\epsilon \leq A \frac{r^2}{\epsilon},
\end{split}
\end{equation}
which proves (i)

For assertion (ii), since $\Omega$ is a bounded domain, it suffices to cover it with balls of the two types used in (i), or 
to choose in (i) a ball which covers $\Omega$.
\end{proof}

The next Lemma will be used to show that the elastic energy associated to cores vanishes in the continuum limit $\epsilon \rightarrow 0$.
\begin{lemma}\label{Eq:ElasticEnergyCore}
Let $\boldsymbol \varphi^\pc_{\epsilon}$ be as defined in \eqref{Eq:u_dislo}, and $N_{\epsilon}$ and $B_{c\epsilon}(\mathbf{X}_k)$ as defined in (C1) and (C3). Then  there exists $C_c>0$ such that
\begin{equation}
\sum_k^{N_{\epsilon}} \int_{B_{c \epsilon (\mathbf{X}_k)}} |\nabla \boldsymbol \varphi^\pc_{\epsilon}|^2 dX \leq C_c \epsilon.
\end{equation}
\end{lemma}
\begin{proof}
From \eqref{eqnablapedid} and scaling  \eqref{Eq:Scaling}, the result is readily obtained
\begin{equation}
\sum_k^{N_{\epsilon}} \int_{B_{c \epsilon (\mathbf{X}_k)}} |\nabla \boldsymbol \varphi^\pc_{\epsilon}|^2 dX \leq 16 \sum_k^{N_{\epsilon}} |B_{c\epsilon}(\mathbf{X}_k)|  \leq C_c \epsilon\,.
\end{equation}
\end{proof}

In closing this Section we show that the decomposition we assumed for the plastic strain is generically admissible. This is not used in the current argument, but illustrates
the generality of our assumptions.
\begin{lemma}\label{Thm:Decomposition_TripleShear}
Consider an arbitrary $\Fp \in C^{\infty}\left(\Omega;\mathbb{R}^{2\times 2}\right)$ such that 
$\det \Fp=1$,
$F^\pl_{11}\neq 0$ and $F^\pl_{12}\neq F^\pl_{11}$. Then $\Fp$ admits a unique description as the composition of the following three sequential simple shear deformations
\begin{equation} \label{Eq:TripleShear}
\Fp=\left(\mathbf{I}-\gamma\ \mathbf{e}_2\otimes \mathbf{e}_1 \right)\left(\mathbf{I}+\eta\ \mathbf{e}_1\otimes \mathbf{e}_2 \right)\left(\mathbf{I}+\mu\ \mathbf{e}_3^{\bot}\otimes \mathbf{e}_3 \right)
\end{equation}
with $\mathbf{e}_3 = \sqrt{2}/2\left(\mathbf{e}_1+\mathbf{e}_2 \right)$. In particular, 
\begin{align} \label{Eq:mu}
&\mu = 2\frac{F^\pl_{11}-1}{F^\pl_{12}-F^\pl_{11}} \\
&\eta = 1+F^\pl_{12}-F^\pl_{11} \\ \label{Eq:gamma}
&\gamma = \frac{1+F^\pl_{21}-F^\pl_{22}}{F^\pl_{12}-F^\pl_{11}},
\end{align}
where $F^\pl_{11}$, $F^\pl_{12}$, $F^\pl_{21}$ and $F^\pl_{22}$ are the components of $\F^\pl$.
\end{lemma}

\begin{proof}
Expanding the product in \eqref{Eq:TripleShear} one obtains
\[ \F^\pl = \left( \begin{array}{cc}
F^\pl_{11} & F^\pl_{12} \\
F^\pl_{21} & F^\pl_{22}  \end{array} \right) =\left( \begin{array}{cc}
1-\frac{\mu}{2}+\eta \frac{\mu}{2} & -\frac{\mu}{2}+\eta+\eta\frac{\mu}{2} \\
-\gamma +\gamma\frac{\mu}{2}-\frac{\mu}{2} \gamma \eta + \frac{\mu}{2} &\gamma \frac{\mu}{2} -\gamma \eta -\gamma \eta \frac{\mu}{2} + 1 + \frac{\mu}{2}  \end{array} \right).
\] 

From the expressions of $F^\pl_{11}= 1 + \frac{\mu}{2}(\eta -1)$ and  $F^\pl_{12}= \eta + \frac{\mu}{2}(\eta -1)$, 
one obtains $\eta = 1+F^\pl_{12}-F^\pl_{11} $. If $F^\pl_{12}\neq F^\pl_{11}$, then, plugging the result into the equation for $F^\pl_{11} $ we get $F^\pl_{11}=1+\frac{\mu}{2} \left(F^\pl_{12}-F^\pl_{11} \right)$, and so \eqref{Eq:mu}. Finally, the expression for $F^\pl_{21}$ simplifies to 
\begin{equation}
F^\pl_{21} = \frac{\mu}{2}+\gamma [ -\frac{\mu}{2}(\eta -1)-1] =  \frac{\mu}{2}-\gamma F^\pl_{11}.
\end{equation}
Using the previous result for $\mu$, and the fact that $\det \F^\pl=1$, one obtains
\begin{equation}
\gamma F^\pl_{11}= \frac{F^\pl_{11}-F^\pl_{11}F^\pl_{22}+F^\pl_{11}F^\pl_{21}}{F^\pl_{12}-F^\pl_{11}},
\end{equation}
which simplifies to \eqref{Eq:gamma} if $F^\pl_{11}\neq 0$.
\end{proof}
We note that the choice of the orientation of the three slip systems was arbitrary and different orientations may be chosen. The conditions for obtaining a unique representation under such choice will then differ from that of this lemma.

\section{Main results.} \label{Sec:MainResults}
\begin{proposition} \label{Thm::General_FeFp}
Let $\boldsymbol \varphi_{\epsilon} \in X_{\epsilon}$ as defined in \eqref{Def:Xe},  $\F^\el_{\epsilon}$ and $\F^\pl_{\epsilon} $  as defined in \eqref{Eq:Fe_epsilon} and \eqref{Eq:Fp_epsilon} respectively and $L$ as defined in Lemma \ref{Lemma:up_inverse_Lipschitz}. Assume $\sup_{\epsilon} E_{\epsilon} (\pe)< \infty$. Then, there is $C_1$ depending on $\sup_{\epsilon} E_{\epsilon} (\pe)$, such that for any set K compactly contained in $\Omega$ and for any $\epsilon \in (0,1)$ satisfying $\epsilon < \operatorname{dist}(\partial \Omega, K)/(2L^2)$ 
\begin{equation}
|D\boldsymbol \varphi_{\epsilon}-\F^\el_{\epsilon}\F^\pl_{\epsilon}|(K) \leq C_1 \epsilon^{1/2}.
\end{equation}
 
\end{proposition}

\begin{proof}

We decompose the proof into three steps that concern, respectively, the dislocation cores, the compatible subdomains away from the dislocations and, finally, the total domain. These intermediate results will often use a combination of the Boundary Trace theorem (Theorem 3.87 of \cite{AmbrosioFuscoPallara2000}) and  the Poincar\'e inequality (Theorem 3.44 of \cite{AmbrosioFuscoPallara2000}) for $BV$ functions, that we make precise here.  Note that the constant in the Boundary Trace theorem depends on the domain. If the domain is a ball $B_1$ of radius 1, the theorem asserts $\int_{\partial B_1} |f| d\mathcal{H}^1 \leq C_{t1} \|f\|_{BV(B_1)} = C_{t1} \int_{B_1} |f| dX + C_{t1} |Df|(B_1)$ for some constant $C_{t1}>0$. If $B_1$ is replaced by $B_{\epsilon}$, where $B_{\epsilon}$ 
represents a ball of radius $\epsilon$, a scaling argument yields the inequality $\int_{\partial B_{\epsilon}} |f| d\mathcal{H}^1 \leq \tfrac{C_{t1}}\epsilon \int_{B_\epsilon} |f| dX + C_{t1} |Df|(B_\epsilon)
$ for the same constant $C_{t1}$ as before.
In particular, if $f\in BV(B_{\epsilon})$, then
\begin{equation} \label{Eq:Trace_Poincare}
\int_{\partial B_{\epsilon}} |f-\bar{f}| d\mathcal{H}^1 \leq C_{t1} |Df|(B_{\epsilon})+\frac{C_{t_1}}{\epsilon} \int_{B_{\epsilon}} |f-\bar{f}| dX \leq C_{tp} |Df|(B_{\epsilon}) ,
\end{equation}
where $\bar{f}=|B_{\epsilon}|^{-1}\int_{B_{\epsilon}}f dX$ and 
the second inequality follows from Poincar\'e's inequality, see also \cite{AmbrosioFuscoPallara2000}[Remark~3.45] for a corresponding scaling argument. 
The constants $C_{t1}$ and $C_{tp}$ are independent of $\epsilon$ and of course, independent of $f$. 
%Poincar\'e inequality (Theorem 3.44 of \cite{AmbrosioFuscoPallara2000}) for $BV$ functions, that we make precise here. In particular, if $f\in BV(B_{\epsilon})$, where $B_{\epsilon}$ represents a ball of radius $\epsilon$, then, by successive application of the aforementioned theorems, one obtains
%\begin{equation} \label{Eq:Trace_Poincare}
%\int_{\partial B_{\epsilon}} |f-\bar{f}| d\mathcal{H}^1 \leq C_{t1} |Df|(B_{\epsilon})+\frac{C_{t_2}}{\epsilon} \int_{B_{\epsilon}} |f-\bar{f}| dX \leq C_{tp} |Df|(B_{\epsilon}) ,
%\end{equation}
%where $\bar{f}=|B_{\epsilon}|^{-1}\int_{B_{\epsilon}}f dX$ and the constants $C_{t1}$, $C_{t2}$ and $C_{tp}$ are independent of $\epsilon$ and  $f$. 
The above result is also true for $\mathcal{J}\cap B_{\epsilon}$, \begin{equation}  \label{Eq:PoincareSegmentsBall}
\int_{\mathcal{J}\cap B_{\epsilon}} |f-\bar{f}| d\mathcal{H}^1 \leq C_{tp'} |Df|(B_{\epsilon}) ,
\end{equation}
as long as $\mathcal{J}\cap B_{\epsilon}$ is the union of at most $M$ segments, where the constant $C_{tp'}$ depends only on $M$.
 This estimate holds for $f$ on the jump set being defined as either one of the two one-sided traces $f^+$ and $f^-$, or as the average $(f^++f^-)/2$.
To see this, let $C_{p1}$ be the 
constant entering the estimate corresponding to (\ref{Eq:Trace_Poincare})
for the half ball, which does not depend on the orientation of the half ball.
Given a segment $\gamma \subset B_{\epsilon}$, consider the diameter $\gamma_1$  parallel to $\gamma$, and let $\gamma'$ be the projection of 
$\gamma$ onto $\gamma_1$.
By the reverse triangle inequality and the fundamental theorem of calculus, working for notational simplicity in coordinates such that 
$\gamma$ is parallel to $e_1$,
\begin{equation}
\begin{split}
\Big |\int_{\gamma} & |f-\bar{f}| d\mathcal{H}^1- \int_{\gamma'} |f-\bar{f}| d\mathcal{H}^1\Big| =\Big| \int_{[{X_1},{X_2}]} \left(|f-\bar{f}|(X,Y_2)-|f-\bar{f}|(X,Y_1)\right) d\mathcal{H}^1\Big|\\
&\leq  \int_{[{X_1},{X_2}]}  \Big| \left(|f-\bar{f}|(X,Y_2)-|f-\bar{f}|(X,Y_1)\right) \Big| d\mathcal{H}^1 \leq \int_{[{X_1},{X_2}]}  \Big| \left((f-\bar{f})(X,Y_2)-(f-\bar{f})(X,Y_1)\right)\Big| d\mathcal{H}^1\\
&\leq  \int_{[{X_1},{X_2}]}  \Big|  \int_{[{Y_1},{Y_2}]} D_2(f-\bar{f})d\mathcal{H}^1 \Big| d\mathcal{H}^1 \leq \int_{D_{\gamma}} \Big| D_{2}(f-\bar{f}) \Big|dX 
%&=\Big| \int_{[{X_1},{X_2}]\times[{Y_1},{Y_2}]} D_{2}|f-\bar{f}| dX \Big| = \Big| \int_{D_{\gamma}} D_{2}|f-\bar{f}| dX \Big|,
\end{split}
\end{equation}
where $D_2$ denotes the distributional derivative in direction orthogonal to segment $\gamma$, and 
$(X_1,Y_2)$, $(X_2,Y_2)$ are the endpoints of $\gamma$, 
$(X_1,Y_1)$, $(X_2,Y_1)$  the endpoints of $\gamma'$, see  Fig.~\ref{Fig:TracePoincare}. Therefore, %since $|D_2 |f-\bar{f}|| = |D_2( f-\bar{f})|$,
\begin{equation}
\int_{\gamma}|f-\bar{f}| d\mathcal{H}^1 \leq \int_{\gamma'}|f-\bar{f}| d\mathcal{H}^1  +  |Df | (D_{\gamma}) \leq \int_{\gamma_1}|f-\bar{f}| d\mathcal{H}^1  +  |Df | (B_\epsilon)\leq\left(C_{p1}+1  \right) |Df|  (B_{\epsilon}),
\end{equation}
so that for $M$ segments, \eqref{Eq:PoincareSegmentsBall} holds with $C_{tp'}=M(C_{p1}+1)$.

\begin{figure}
\begin{center}
    {\includegraphics[width=0.35\textwidth]{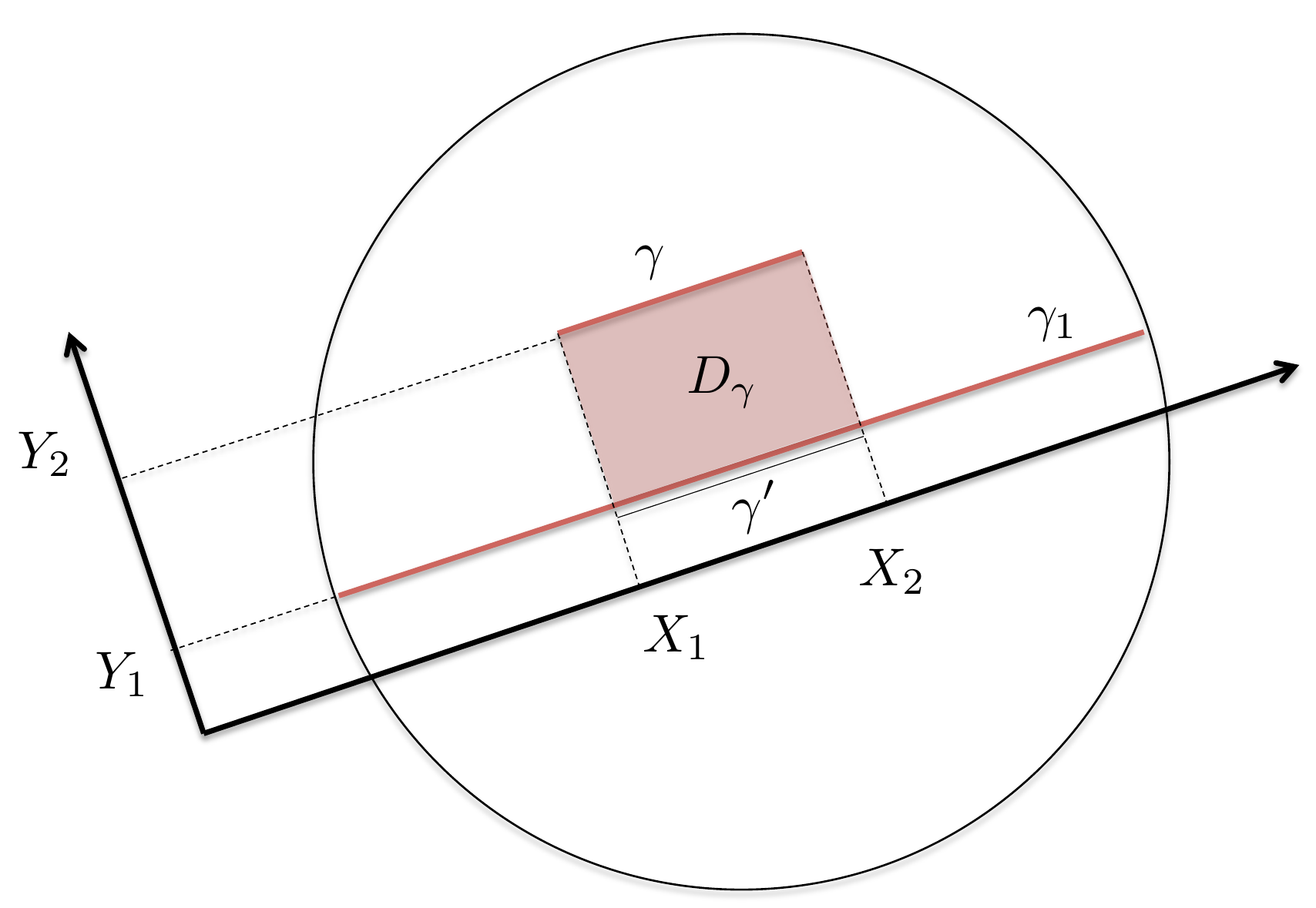}}
    \caption[]{Sketch of the strategy to bound the absolute value of a function over a segment in a ball, with the variation of the function over the ball.}
    \label{Fig:TracePoincare}
\end{center}
\end{figure}

\begin{enumerate}

\item[(S1)] In this first step we show
\begin{align} \label{Eq:FeFp_Step1}
&|D \boldsymbol \varphi_{\epsilon}| \left( \bigcup \limits_k^{N_{\epsilon}} B_{c \epsilon} (\mathbf{X}_k)\right) \leq C_2 \epsilon^{1/2}, \\ \label{Eq:FeFp_Step1b}
&|\F^\el_{\epsilon}\F^\pl_{\epsilon}| \left( \bigcup \limits_k^{N_{\epsilon}} B_{c \epsilon} (\mathbf{X}_k)\right) \leq C_3 \epsilon^{1/2}.
\end{align}
We start with (\ref{Eq:FeFp_Step1}).
The displacement jump at $\mathbf{X} \in \mathcal{J}_j\cap B_{c\epsilon}(\mathbf{X}_k)$ is
\begin{equation}\label{eqjumpphiepsX}
\llbracket \boldsymbol \varphi_{\epsilon}(\mathbf{X}) \rrbracket = \boldsymbol \varphi^+_{\epsilon}(\mathbf{X}) - \boldsymbol \varphi^-_{\epsilon}(\mathbf{X}) = \left( \boldsymbol \varphi^+_{\epsilon}(\mathbf{X}) - \boldsymbol \varphi^+_{\epsilon}\left(\mathbf{X}_k \right) \right) -\left( \boldsymbol \varphi^-_{\epsilon}(\mathbf{X}) - \boldsymbol \varphi^-_{\epsilon}\left(\mathbf{X}_k \right) \right)
\end{equation}
where we have used $\boldsymbol \varphi^+_{\epsilon}\left(\mathbf{X}_k \right)=\boldsymbol \varphi^-_{\epsilon}\left(\mathbf{X}_k \right)$, c.f.~(C3) in \eqref{Def:Xe}. 
We further note that between $\mathbf{X}_k$ and $\mathbf{X}$ either above or below $\mathcal{J}_j$ we do not cross any other jump set, 
 and therefore $D\boldsymbol \varphi_{\epsilon}  = \F^{\el}_{\epsilon}\mathcal L^2$ on both sides. Since 
 $\F^{\el}_{\epsilon}$ has a trace on each side of the jump set, we obtain 
\begin{equation}\label{gsfgsf}
 \boldsymbol \varphi^+_{\epsilon}(\mathbf{X}) - \boldsymbol \varphi^+_{\epsilon}\left(\mathbf{X}_k \right) 
 =\int_{\mathbf{X}_k}^{\mathbf{X}} \F^{\el+}_{\epsilon} \, d\mathbf{X}' \,,
\end{equation}
 and the same on the other side. 
 %To prove this, it suffices to replace $ \boldsymbol \varphi^+_{\epsilon}(\mathbf{X})$ by the average over a small one-sided
 %neighborhood of $\mathbf X$, and the same at $\mathbf X_k$ 
  Inserting in (\ref{eqjumpphiepsX}),
 \begin{equation}
 \boldsymbol \varphi^+_{\epsilon}(\mathbf{X}) - \boldsymbol \varphi^-_{\epsilon}\left(\mathbf{X}\right) 
 =\int_{\mathbf{X}_k}^{\mathbf{X}} (\F^{\el+}_{\epsilon}-\F^{\el-}_{\epsilon}) \, d\mathbf{X}' \,,
\end{equation}
which implies $|\llbracket \boldsymbol \varphi_{\epsilon}(\mathbf{X}) \rrbracket|\le |D \F^{\el}_{\epsilon}|(B_{c \epsilon} (\mathbf{X}_k))$.
\begin{figure}
\begin{center}
    {\includegraphics[width=0.35\textwidth]{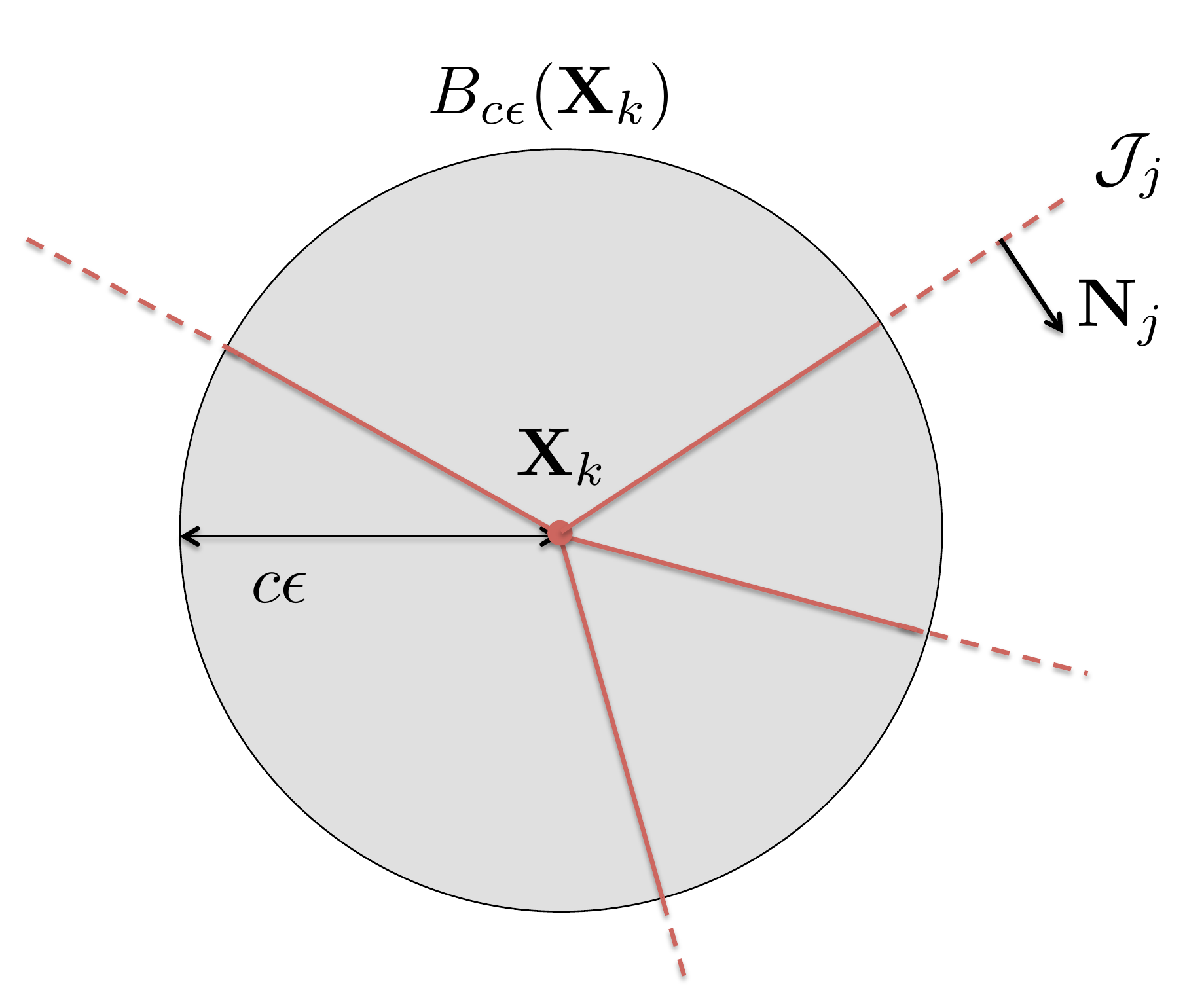}}
    \caption[]{Example of a dislocation core $B_{c\epsilon} (\mathbf{X}_k)$. The dislocation is centered at $\mathbf{X}_k$ and is the boundary of a finite number of slip lines $\mathcal{J}_j$, $1 \leq j \leq N_{dk}\leq N_s$.}
    \label{Fig:DislocationCore}
\end{center}
\end{figure}
Integrating over one of the slip lines inside the core,  which have length $c \epsilon$, and then summing over all of the slip lines, we obtain
\begin{equation} 
\sum_j^{N_{dk}} \int_{\mathcal{J}_j \cap B_{c\epsilon }(\mathbf{X}_k)} |\llbracket \boldsymbol \varphi_{\epsilon} \rrbracket| \, d \mathcal{H}^1  \leq  N_s c \epsilon |D\F^\el_{\epsilon}| \left( B_{c \epsilon }(\mathbf{X}_k) \right),  
\end{equation}
where we have used $N_{dk} \leq N_s$, c.f.~(C1) in \eqref{Def:Xe}. Summing now over all dislocation cores, we obtain 
\begin{equation} \label{Eq:eq3}
\sum_k^{N_{\epsilon}} \sum_j^{N_{dk}} \int_{\mathcal{J}_j\cap B_{c\epsilon }(\mathbf{X}_k)}|\llbracket \boldsymbol \varphi_{\epsilon} \rrbracket| \, d \mathcal{H}^1  \leq  N_s c  \epsilon
|D\F^\el_{\epsilon}| \left(\bigcup \limits_k^{N_{\epsilon}} B_{c \epsilon} (\mathbf{X}_k)\right)
\le C_4 \epsilon E_\epsilon\,.
\end{equation}
We turn to the absolutely continuous part of the gradient of $\boldsymbol \varphi_\epsilon$, which by H\"older's inequality can be estimated as
\begin{equation}\label{eqcontrfepsl1}
\sum_k^{N_{\epsilon}} \int_{B_{c \epsilon }(\mathbf{X}_k)} |\nabla \boldsymbol \varphi_{\epsilon}|\, dX=
\sum_k^{N_{\epsilon}} \int_{B_{c \epsilon }(\mathbf{X}_k)} |\F^\el_{\epsilon}|\, dX\le
\sqrt\pi c\epsilon N_\epsilon^{1/2} \left( \sum_k^{N_{\epsilon}}  \int_{B_{c \epsilon }(\mathbf{X}_k)} |\F^\el_{\epsilon}|^2\, dX\right)^{1/2}\,.
\end{equation}
Recalling the quadratic growth of the energy, c.f.~\eqref{Eq:We_growth}, and the definition  \eqref{Eq:ElasticEnergy} we see that the parenthesis is bounded.
Since $N_\epsilon\le C/\epsilon$ (recall \eqref{Eq:Scaling}), this proves \eqref{Eq:FeFp_Step1}.

We now turn to \eqref{Eq:FeFp_Step1b}. By \eqref{Eq:Fe_J} and the bound $|\mathbf{b}_{\epsilon j}|\leq n' \epsilon \leq n \epsilon$, c.f.~(C2) and (C3) in \eqref{Def:Xe},
\begin{align} 
|\F^\el_{\epsilon} \F^\pl_{\epsilon}|& \left(  B_{c \epsilon }(\mathbf{X}_k) \right) \leq  \int_{ B_{c \epsilon }(\mathbf{X}_k)} | \F^\el_{\epsilon}| \, dX + n \epsilon   \sum_j^{N_{dk}} \int_{\mathcal{J}_j\cap B_{c\epsilon }(\mathbf{X}_k)}| \F^\el_{\epsilon}| \, d\mathcal{H}^1  \,.
\end{align}
In the second term we use (\ref{Eq:PoincareSegmentsBall}), so that it
becomes
\begin{align} 
n \epsilon   \sum_j^{N_{dk}} \int_{\mathcal{J}_j\cap B_{c\epsilon }(\mathbf{X}_k)}| \F^\el_{\epsilon}| \, d\mathcal{H}^1  \le
n \epsilon   \sum_j^{N_{dk}} \int_{\mathcal{J}_j\cap B_{c\epsilon }(\mathbf{X}_k)}| \overline \F^\el_{\epsilon}| \, d\mathcal{H}^1  +C_{5}
 n \epsilon  |D\F^\el_{\epsilon}|(B_{c\epsilon }(\mathbf{X}_k))\,,
\end{align}
where $ \overline \F^\el_{\epsilon}$ denotes the average of $\F^\el_{\epsilon}$ over $B_{c\epsilon }(\mathbf{X}_k)$, which obeys
\begin{equation}\label{eqaveragefeps}
|\bar{\F}^{\el}_{\epsilon }(\mathbf{X}_k)| = \frac{1}{\pi c^2\epsilon^2} \left|\int_{B_{c\epsilon}(\mathbf{X}_k)} \F^\el_{\epsilon} \,dX\right|
\le \frac{1}{\pi c^2\epsilon^2}\int_{B_{c\epsilon}(\mathbf{X}_k)}  \left|\F^\el_{\epsilon}\right| \,dX\,.
\end{equation}
Therefore
\begin{align} \label{eqs1final}
|\F^\el_{\epsilon} \F^\pl_{\epsilon}|& \left(  B_{c \epsilon }(\mathbf{X}_k) \right) \leq  
(1+\frac{n N_s c}{\pi c^2})
\int_{ B_{c \epsilon }(\mathbf{X}_k)} | \F^\el_{\epsilon}| \, dX +  C_{5}
 n \epsilon  |D\F^\el_{\epsilon}|(B_{c\epsilon }(\mathbf{X}_k))\,.
\end{align}
After summing over $k$, the first term is bounded as in (\ref{eqcontrfepsl1}), the second as in (\ref{Eq:eq3}). This concludes
the proof of  \eqref{Eq:FeFp_Step1b}.

\item[(S2)] Next, we cover the domain $ K \setminus \bigcup \limits_k^{N_{\epsilon}}B_{c\epsilon} (\mathbf{X}_k)$ with finitely many balls $B_{\epsilon}(\mathbf Y_l)$, with $\mathbf Y_l \in K$, such that the larger balls $B_{L^2\epsilon}(\mathbf Y_l)$ do not intersect the dislocations cores $B_{m\epsilon} (\mathbf{X}_k)$ and have finite overlap. Such a covering is possible in view of Eq.~\eqref{Eq:DislocationsSeparation} and, by the definition of $K$, $B_{L^2\epsilon}(\mathbf Y_l)\subset \Omega' \subset \Omega$.

We shall now show that for any $l$ one has
\begin{equation}\label{eqestimateS2oneball}
 |D\boldsymbol \varphi_{\epsilon}-\F^\el_{\epsilon}\F^\pl_{\epsilon}| \left(B_{\epsilon}(\mathbf Y_l) \right)   \leq C_6 \epsilon |D\mathbf F^\el_\epsilon|(B_{h\epsilon}(\mathbf Y_l) )\,,
\end{equation}
with $h=1+Ln \leq L^2$. Summing over all the balls will give the estimate on  $ K \setminus \bigcup \limits_k^{N_{\epsilon}}B_{c\epsilon} (\mathbf{X}_k)$.

Fix $l$ and let $\boldsymbol \varphi_{\epsilon}^\el$ and $\boldsymbol \varphi_{\epsilon}^\pl$ be as in (C2),
applied to the set $B_{2L^2\epsilon}(\mathbf Y_l)$. Then,
$\boldsymbol\varphi_\epsilon=\boldsymbol\varphi^\el_\epsilon\circ\boldsymbol\varphi^\pl_\epsilon$. Taking the gradient we obtain
$\mathbf F^\el_\epsilon=\nabla \boldsymbol\varphi_\epsilon^\el \circ \boldsymbol\varphi_\epsilon^\pl$ and
\begin{equation}
D\boldsymbol\varphi_\epsilon=\mathbf F^\el_\epsilon \mathcal{L}^2 +  \sum_j \llbracket \boldsymbol \varphi_{\epsilon}\rrbracket\otimes \mathbf{N}_j\ \mathcal{H}^1 \lfloor_{\mathcal{J}_j} \hskip1cm \text{ in $B_{h\epsilon}(\mathbf Y_l)$}\,.
\end{equation}
Therefore, recalling (\ref{eqdvarphipleps}),
\begin{equation}\label{eqdvarphifefpeps}
 |D\boldsymbol \varphi_{\epsilon}-\F^\el_{\epsilon}\F^\pl_{\epsilon}| \left(B_{\epsilon}(\mathbf Y_l) \right)   
 =\int_{B_{\epsilon}(\mathbf Y_l) \cap \mathcal J}
 |\llbracket \boldsymbol \varphi_{\epsilon}\rrbracket- {\F}^\el_{\epsilon} \mathbf{b}_{\epsilon}| 
 d\mathcal H^1\,,
\end{equation}
where $\mathbf{b}_{\epsilon}$ takes the constant value $\mathbf{b}_{\epsilon j}$ on each segment $\mathcal{J}_j \subset B_{\epsilon}(\mathbf Y_l) \cap \mathcal J$.
For $\mathbf X\in B_{\epsilon}(\mathbf Y_l) \cap \mathcal J_j$ 
we write $\tilde{\mathbf{X}}^+ = \boldsymbol \varphi^{\pl+}_{\epsilon}(\mathbf{X})$, $\tilde{\mathbf{X}}^- = \boldsymbol \varphi^{\pl-}_{\epsilon}(\mathbf{X})$ and $\tilde{\mathbf{X}}^+ - \tilde{\mathbf{X}}^- = \mathbf{b}_{\epsilon j}$.
Then,
\begin{equation} \label{Eq:eq4}
\llbracket \boldsymbol \varphi_{\epsilon}\rrbracket(\mathbf{X}) = \boldsymbol \varphi^\el_{\epsilon}(\tilde{\mathbf{X}}^+)-\boldsymbol \varphi^\el_{\epsilon}(\tilde{\mathbf{X}}^-) = \int_{\tilde{\mathbf{X}}^-}^{\tilde{\mathbf{X}}^+} \tilde{\nabla} \boldsymbol \varphi^\el_{\epsilon} \, d\tilde{\mathbf{X}}= \int_0^1{\F}^\el_{\epsilon}(\boldsymbol\varphi^{\pl,-1}_\epsilon(
\tilde{\mathbf{X}}^-+ t \mathbf{b}_{\epsilon j}))\mathbf{b}_{\epsilon j}\, dt\,.
\end{equation}
The value of $\F^\el_\epsilon$ on the jump set was defined in (\ref{Eq:Fe_J}).
Since $S=[\tilde{\mathbf{X}}^+,\tilde{\mathbf{X}}^-]$ is a segment of length no larger than $n\epsilon$,
by Lemma \ref{Lemma:up_inverse_Lipschitz}, 
$\boldsymbol\varphi_\epsilon^{\pl,-1}(S) \subset B_{h\epsilon}(\mathbf{Y}_l)$. Furthermore, $\boldsymbol\varphi_\epsilon^{\pl,-1}(S)$ consists of the union of at most $M=(1+n)^{N_s}$
segments. Indeed, each $\boldsymbol\varphi_{\epsilon,\nu}^{\pl,-1}$ is piecewise a translation, and a segment of length no larger than
 $n\epsilon$ crosses at most $\lfloor n\rfloor +1$ discontinuity points, hence gets decomposed into at most
 $\lfloor n\rfloor +1$ segments, each of length no larger than $n\epsilon$ (plus possibly a finite number of points, which may be ignored). Iterating at most $N_s$ times gives the result. %The segments are disjoint because  $\boldsymbol\varphi_\epsilon^{\pl}$ is bijective. 
 Then, by \eqref{Eq:PoincareSegmentsBall} one obtains
\begin{equation}\label{eqS2trace}
  \int_{\boldsymbol\varphi_\epsilon^{\pl,-1}(S)} |\mathbf F^\el_\epsilon- \overline {\mathbf F}^\el_\epsilon| \, 
  d\mathcal H^1
    \le C_{7} |D\mathbf F^\el_\epsilon | \left( B_{ h\epsilon}(\mathbf Y_l)\right),
\end{equation}
where $\overline {\mathbf F}^\el_\epsilon$ is the average of $\F^\el_{\epsilon} $ over $B_{h\epsilon}(\mathbf Y_l)$,
\begin{equation}
\overline {\mathbf F}^\el_\epsilon = \frac{1}{|{B}_{h\epsilon }({\mathbf{Y}}_l)|} \int_{{B}_{h\epsilon }({\mathbf{Y}}_l) } {\F}^\el_{\epsilon}  \, d{X}\,.
\end{equation}
Therefore (\ref{Eq:eq4}) and (\ref{eqS2trace}) yield
\begin{equation} 
\left|\llbracket \boldsymbol \varphi_{\epsilon}\rrbracket (\mathbf{X})
-\overline {\mathbf F}^\el_\epsilon\mathbf{b}_{\epsilon j}\right|\le 
  \int_{\boldsymbol\varphi_\epsilon^{\pl,-1}(S)} |\mathbf F^\el_\epsilon- \overline {\mathbf F}^\el_\epsilon| 
 \,  d\mathcal H^1
\le C_{7} |D\mathbf F^\el_\epsilon | \left( B_{ h\epsilon}(\mathbf Y_l)\right)\,.
\end{equation}
We write  (\ref{eqdvarphifefpeps})  as
\begin{equation} 
 |D\boldsymbol \varphi_{\epsilon}-\F^\el_{\epsilon}\F^\pl_{\epsilon}| \left(B_{\epsilon}(\mathbf Y_l) \right)   
 \le\int_{B_{\epsilon}(\mathbf Y_l) \cap \mathcal J}
 |\llbracket \boldsymbol \varphi_{\epsilon}\rrbracket- \overline{\F}^\el_{\epsilon} \mathbf{b}_{\epsilon}| \, 
 d\mathcal H^1
+ |\mathbf{b}_{\epsilon}|  \int_{B_{\epsilon}(\mathbf Y_l) \cap \mathcal J  } |\mathbf F^\el_\epsilon- \overline {\mathbf F}^\el_\epsilon|\,  d\mathcal H^1\,.
\end{equation}
Since, by \eqref{Eq:PoincareSegmentsBall}, 
\begin{equation}\label{eqfreasfr}
  \int_{B_{\epsilon}(\mathbf Y_l) \cap \mathcal J  } |\mathbf F^\el_\epsilon- \overline {\mathbf F}^\el_\epsilon| d\mathcal H^1  \le C_{8}  |D\mathbf F^\el_\epsilon | \left( B_{ h\epsilon}(\mathbf Y_l)\right),
\end{equation}
this concludes the proof of \eqref{eqestimateS2oneball}. We remark that by (C2) and the proof of Lemma \ref{Lemma:SizeControl},  $|B_{\epsilon}(\mathbf Y_k) \cap \mathcal J |$ consists of a finite union of segments bounded independently of $\epsilon$.

\item[(S3)] We combine the results for the dislocation cores in (S1) and each $B_{\epsilon}(\mathbf Y_l)$ in (S2), which, by construction, have a finite overlap. Then, 
\begin{equation} \label{Eq:S3Proof}
 |D\boldsymbol \varphi_{\epsilon}-\F^\el_{\epsilon}\F^\pl_{\epsilon}| \left(K\right)   \leq C_1 \epsilon^{1/2}.
\end{equation}
\end{enumerate}
\end{proof}

\begin{theorem} \label{Thm:Compactness}
Let $\boldsymbol \varphi_{\epsilon} \in X_{\epsilon}$, $\F_{\epsilon}, \F^\el_{\epsilon}$ and $\F^\pl_{\epsilon}$ as defined in Section \ref{Sec:ProblemSetting}, and $\sup_{\epsilon} E_{\epsilon}(\boldsymbol \varphi_{\epsilon}) < \infty$. Then there exist
$\mathbf F^\pl$, $\mathbf F^\el$, $\boldsymbol\varphi$ and a
subsequence (not relabeled) such that 
\begin{align}  \label{Eq:Fp_compactness_M}
&\F^\pl_{\epsilon} \overset{*}{\rightharpoonup} \F^\pl \qquad \text{in } \mathcal{M} \left(\Omega; \mathbb{R}^{2 \times 2}\right), \\ \label{Eq:Fe_compactness_L2}
&\F^\el_{\epsilon} \rightharpoonup \F^\el \qquad \text{in } L^2 \left(\Omega; \mathbb{R}^{2 \times 2}\right),\\ \label{Eq:Fe_compactness_BV}
&\F^\el_{\epsilon} \overset{*}{\rightharpoonup} \F^\el \qquad \text{in }  BV \left(\Omega; \mathbb{R}^{2 \times 2}\right), \\ \label{Eq:Fe_compactness_L1}
&\F^\el_{\epsilon} \rightarrow \F^\el \qquad \text{in } L^1 \left(\Omega; \mathbb{R}^{2 \times 2}\right), \\ \label{Eq:varphi_compactness}
&\boldsymbol \varphi_{\epsilon} - \mathbf{v}_{\epsilon} \rightarrow \boldsymbol \varphi \qquad \text{in }  L^1_{\mathrm{loc}}\left(\Omega; \mathbb{R}^2 \right), \\ \label{Eq:F_compactness}
&\F_{\epsilon} = D\boldsymbol \varphi_{\epsilon}\overset{*}{\rightharpoonup} D\boldsymbol \varphi=\F \qquad \text{in } \mathcal{M}_{\mathrm{loc}}\left(\Omega;\mathbb{R}^{2 \times 2} \right),\\
&\operatorname{Curl}\ \F^\pl_{\epsilon }\lfloor_{\Omega'}  \overset{*}{\rightharpoonup}\operatorname{Curl}\ \F^\pl \qquad \text{in } \mathcal{M}\left(\Omega;\mathbb{R}^{ 2} \right)
\end{align}
as $\epsilon \rightarrow 0$, where  $\mathbf{v}_{\epsilon}$ is a sequence of rigid translations and $\Omega'=\{\mathbf{X}\in \Omega : \operatorname{dist}\left(\mathbf{X},\partial \Omega \right)>L^2\epsilon\}$. The support of the measure $\operatorname{Curl }\F^\pl_{\epsilon }\lfloor_{\Omega'} $ thus exclusively consists of dislocation points. Furthermore $\boldsymbol \varphi$ is approximately continuous $\mathcal{H}^1$-almost everywhere.

\end{theorem}

We recall that $\mu_{\epsilon}$ is said to weakly* converge to $\mu$ in $\mathcal{M}_{\mathrm{loc}}(\Omega)$, if 
\begin{equation}
\lim_{\epsilon \rightarrow 0}\int_{\Omega} f d\mu_{\epsilon} = \int_{\Omega} f d\mu
\end{equation}
for all $f \in C_c(\Omega)$, c.f.~Definition 1.58 of \cite{AmbrosioFuscoPallara2000}.

\begin{proof}

{\bf  Convergence of $\F^\pl_{\epsilon}$.}

By the definitions of $\F^\pl_{\epsilon}$ and $\mathcal{J}$, and the bound on $\mathbf{b}_{\epsilon j}$  \eqref{Eq:b_bound}, $|\F^\pl_{\epsilon}|(\Omega) \leq \sqrt2
|\Omega| + \sum_j |\mathbf{b}_{\epsilon j}| |\mathcal{J}_j| \leq \sqrt2|\Omega| + n \epsilon |\mathcal{J}| $. Therefore $\sup_{\epsilon} |\F^\pl_{\epsilon}|(\Omega) < \infty$ and \eqref{Eq:Fp_compactness_M} follows by weak* compactness (Theorem 1.59 of \cite{AmbrosioFuscoPallara2000}).

{\bf Convergence of $\F^\el_{\epsilon}$.}

By the growth condition of the elastic energy \eqref{Eq:We_growth}, $\sup_{\epsilon} \|\F^\el_{\epsilon}\|_{L^2(\Omega)} < \infty$. Then, \eqref{Eq:Fe_compactness_L2} follows by weak compactness (Theorem 1.36 of \cite{AmbrosioFuscoPallara2000}). Furthermore, since $\Omega$ is  bounded and the energy imposes $\sup_{\epsilon}|D\F^\el_{\epsilon}|(\Omega) < \infty$, we obtain $\sup_{\epsilon} \|\F^\el_{\epsilon}\|_{BV(\Omega)} < \infty$. Hence \eqref{Eq:Fe_compactness_BV} and \eqref{Eq:Fe_compactness_L1} follow by $BV$ compactness (Theorem 3.23 of \cite{AmbrosioFuscoPallara2000}).

{\bf Convergence of $\boldsymbol \varphi_{\epsilon}$ and $D\boldsymbol \varphi_{\epsilon}$.}

Consider a fixed closed ball $B$ compactly contained in $\Omega$, and define $\mathbf{v}_{\epsilon}=\frac{1}{|B|}\int_B \pe \, dX$. Next, choose $K'\subset \subset \Omega$ and define $K=K'\cup B$. Then for $\epsilon<\text{dist}(K,\partial \Omega)/(2L^2)$, we show that $|D\boldsymbol \varphi_{\epsilon}|\left(K\right) < \infty $ by decomposing the domain into the dislocation cores and the remainder of the domain, c.f.~proof of Proposition \ref{Thm::General_FeFp}. By \eqref{Eq:FeFp_Step1},
\begin{equation} \label{Eq:eq7}
|D\boldsymbol \varphi_{\epsilon}| \left(\bigcup \limits_k^{N_{\epsilon}} B_{c\epsilon (\mathbf{X}_k)}\right)\leq C_2 \epsilon^{1/2}.
\end{equation}

Next, we cover $K\setminus \bigcup_k^{N_{\epsilon}}B_{c\epsilon}(\mathbf{X}_k)$ with balls of radius $\epsilon$, $B_{\epsilon}(\mathbf{Y}_l)$, $l\in \mathbb{N}$ with finite overlap, 
as in Proposition \ref{Thm::General_FeFp} (S3) and we further define
\begin{equation} \label{Eq:eq6}
\bar{\F}^{\el}_{\epsilon }(\mathbf{Y}_l) = \frac{1}{\pi \epsilon^2} \int_{B_{\epsilon}(\mathbf{Y}_l)} \F^\el_{\epsilon} \, dX.
\end{equation}
By  Proposition \ref{Thm::General_FeFp}, $|D\boldsymbol \varphi_{\epsilon} - \F^\el_{\epsilon} \F^\pl_{\epsilon}| \left(K\right) \le  C_1 \epsilon^{1/2}$. The bound
on $|\F^\el_{\epsilon} \F^\pl_{\epsilon}|(K)$ is obtained with a similar strategy as in Proposition \ref{Thm::General_FeFp}. For each ball, 
since $|\mathbf b_{\epsilon j}|\le n\epsilon$ 
we have
\begin{align}
  | \F^\el_{\epsilon} \F^\pl_{\epsilon}|  \left(B_{\epsilon}(\mathbf{Y}_l)\right) 
  \le \int_{B_{\epsilon}(\mathbf{Y}_l)} | \F^\el_{\epsilon}|\, dX + n\epsilon \int_{\mathcal{J} \cap B_{\epsilon}(\mathbf{Y}_l)} 
   |\F^\el_{\epsilon}| \, d\mathcal{H}^1\,,
\end{align}
and $\mathcal{J} \cap B_{\epsilon}(\mathbf{Y}_l)$ consists of finitely many segments. 
Using \eqref{Eq:PoincareSegmentsBall} as in  (\ref{eqfreasfr}), and estimating the average via
 (\ref{eqaveragefeps}) as in (\ref{eqs1final}) gives
\begin{align}
  | \F^\el_{\epsilon} \F^\pl_{\epsilon}|  \left(B_{\epsilon}(\mathbf{Y}_l)\right) 
  \le C_9 \int_{B_{\epsilon}(\mathbf{Y}_l)} | \F^\el_{\epsilon}|\, dX + C_{10} n\epsilon
  |D\F^\el_{\epsilon}|(  B_{\epsilon}(\mathbf{Y}_l))\,.
\end{align}

Summing over all balls, 
\begin{align}\label{eqfefpbyl}
  | \F^\el_{\epsilon} \F^\pl_{\epsilon}| \left(K\setminus \left(\cup_kB_{c\epsilon}(\mathbf{X}_k)\right)\right)
  \le C_{11} \int_{\Omega} | \F^\el_{\epsilon}| dX + C_{12}\epsilon
  |D\F^\el_{\epsilon}|( \Omega)
\end{align}
and therefore, recalling (\ref{Eq:eq7}),
\begin{equation}\label{eqconvdphifinal} 
\begin{split}
|D\boldsymbol \varphi_{\epsilon}| \left(K\right) & \leq 
|D\boldsymbol \varphi_{\epsilon}| \left(\cup_k B_{c\epsilon (\mathbf{X}_k)}\right)+
|D\boldsymbol \varphi_{\epsilon} - \F^\el_{\epsilon} \F^\pl_{\epsilon}| \left(K\right)  +  | \F^\el_{\epsilon} \F^\pl_{\epsilon}|  \left(K \setminus \cup_k B_{c\epsilon (\mathbf{X}_k)} \right)  \\
&\leq C_2 \epsilon^{1/2}+ C_1  \epsilon^{1/2} + C_{11} \int_{\Omega} |\F^\el_{\epsilon}| \, dX+ \epsilon C_{12}  |D\F^\el_{\epsilon}| \left( \Omega \right) .
\end{split}
\end{equation}

Therefore $\sup_{\epsilon} |D\boldsymbol \varphi_{\epsilon}| \left(K\right)< \infty$. Then, by the Poincar\'e theorem for $BV$ functions, c.f.~Theorem 3.44 of \cite{AmbrosioFuscoPallara2000}, $\sup_{\epsilon} \|\boldsymbol \varphi_{\epsilon}-\bar{\boldsymbol \varphi}_{\epsilon }\|_{L^1(K)} < \infty$, where $\bar{\boldsymbol \varphi}_{\epsilon }$ is the average of $\boldsymbol \varphi_{\epsilon}$ over $K$. Since $B \subset K$ it follows that $|\mathbf{v}_{\epsilon}- \bar{\boldsymbol \varphi}_{\epsilon }|$
is bounded as $\epsilon \rightarrow 0$. Then, as a result, by $BV$-compactness, c.f.~Theorem 3.23 of \cite{AmbrosioFuscoPallara2000}) and the definition of weak* convergence in $BV$ (Definition 3.11 of \cite{AmbrosioFuscoPallara2000}), \eqref{Eq:varphi_compactness} and \eqref{Eq:F_compactness} follow.

{ \bf Continuity of the limiting solution}

From the previous results,  $\boldsymbol \varphi_{\epsilon} -\mathbf{v}_{\epsilon} \rightarrow \boldsymbol \varphi$ in $L^1_{\mathrm{loc}}(\Omega)$. Then, by lower semicontinuity of the variation measure, c.f.~page 172 of \cite{EvansGariepy1991},

\begin{equation}% \label{Eq:eq8}
|D\boldsymbol \varphi |(B_r) \leq \liminf_{\epsilon \rightarrow 0} |D\boldsymbol \varphi_{\epsilon}|(B_r) \quad \forall B_r \subset \subset \Omega,
\end{equation}
where $B_r$ denotes a ball of radius $r$.   We write $B_{2r}$ for the ball with the same center and radius $2r$, and assume $B_{2r}\subset\subset\Omega$.

Let us consider a set $K\subset \subset \Omega$, such that $B_{2r} \subset\subset K$, and take $\epsilon < \text{dist}(\partial \Omega,K)/(2L^2)$.
Then, defining $K_{\epsilon}=K\setminus \left(\cup_k B_{c\epsilon}(\mathbf{X}_k) \right)$, we obtain by Proposition \ref{Thm::General_FeFp}, \eqref{Eq:eq7} and the definition of $\F^\pl_{\epsilon}$
\begin{equation}
\begin{split}%\label{Eq:ContinuityCalculation}
|D\boldsymbol \varphi_{\epsilon}| (B_r) &\leq |D\boldsymbol \varphi_{\epsilon}| (B_r\cap K_{\epsilon})+ |D\boldsymbol \varphi_{\epsilon}| \left(B_r\cap \bigcup \limits_k^{N_{\epsilon}} B_{c \epsilon}(\mathbf{X}_k)\right) \\
&\leq |D\boldsymbol \varphi_{\epsilon}-\F^\el_{\epsilon}\F^\pl_{\epsilon}| (B_r\cap K_{\epsilon})+|\F^\el_{\epsilon}\F^\pl_{\epsilon}| (B_r\cap K_{\epsilon})+ |D\boldsymbol \varphi_{\epsilon}| \left(B_r\cap \bigcup \limits_k^{N_{\epsilon}} B_{c \epsilon}(\mathbf{X}_k)\right) \\
&\leq C_{13} \epsilon^{1/2}+ |\F^\el_{\epsilon}\F^\pl_{\epsilon}|(B_r\cap K_{\epsilon}) \,.
%\\
%&\leq C_{16} \epsilon^{1/2}+\|\F^\el_{\epsilon}\|_{L^1(B_r\cap K_{\epsilon})} + n \epsilon|\F^\el_{\epsilon}|\left((B_r\cap K_{\epsilon})\cap \mathcal{J}\right).
\end{split}
\end{equation}
To estimate the last term we cover $B_r\cap K_{\epsilon}$ with balls of radius $\epsilon$ which do not intersect the cores,
as done in (\ref{eqfefpbyl}). The balls have finite overlap and are all contained in the larger ball $B_{2r}$.
Therefore
\begin{equation}
\begin{split}%\label{Eq:ContinuityCalculation}
|D\boldsymbol \varphi_{\epsilon}| (B_r)
&\leq C_{13} \epsilon^{1/2}+C_{14} \|\F^\el_{\epsilon}\|_{L^1(B_{2r}\cap K_{\epsilon})} + C_{15} \epsilon|D\F^\el_{\epsilon}|\left(B_{2r}\right).
\end{split}
\end{equation}
Passing to the limit, since (\ref{Eq:Fe_compactness_L1}) implies
$\|\F^\el_{\epsilon}\|_{L^1(B_{2r})}\to 
\|\F^\el\|_{L^1(B_{2r})}$, we obtain
\begin{equation}%\label{Eq:eq8}
|D\boldsymbol \varphi |(B_r) \leq \liminf_{\epsilon \rightarrow 0} |D\boldsymbol \varphi_{\epsilon}|(B_r)
\le C_{14} \|\F^\el\|_{L^1(B_{2r})} \le C_{14} 2 \sqrt\pi r  \|\F^\el\|_{L^2(B_{2r)}}\,.
\end{equation}
Therefore
\begin{equation}
|D\boldsymbol \varphi|(B_r) \leq C_{16} r \|\F^\el  \|_{L^2\left(B_{2r} \right) }.
\end{equation}
Since, by \eqref{Eq:Fe_compactness_L2}, $\F^\el \in L^2(\Omega)$,
\begin{equation} \label{Eq:Continuity}
\liminf_{r\rightarrow 0} \frac{|D\boldsymbol \varphi|(B_r)}{r}=0.
\end{equation}
As a result, c.f.~Proposition 3.92 of \cite{AmbrosioFuscoPallara2000}, the jump set of $D\boldsymbol \varphi$ has $\mathcal{H}^1$ measure zero, and $\boldsymbol \varphi$ is approximately 
continuous $\mathcal{H}^1$-almost everywhere. We note that it may contain a Cantor part.

{\bf Convergence of Curl $\F^\pl_{\epsilon}$}

By the scaling \eqref{Eq:Scaling} and the bound on the Burgers vector \eqref{Eq:b_bound}, Curl $\F^\pl_{\epsilon}\lfloor_{\Omega'}$ is equibounded. Therefore, by weak* compactness (Theorem 1.59  of \cite{AmbrosioFuscoPallara2000}),
\begin{equation}
\text{Curl}\ \F^\pl_{\epsilon}\lfloor_{\Omega'} \overset{*}{\rightharpoonup} \mathbf{G} \quad \text{in } \mathcal{M}\left(\Omega; \mathbb{R}^2 \right).
\end{equation}

Since $\F^\pl_\epsilon$ also converges in $\mathcal{M}$, we have $\mathbf{G}= \text{Curl}\ \F^\pl$. Indeed, for every $\phi \in C_c^{\infty}(\Omega)$, one can choose $\epsilon$ such that $\text{supp}\ \phi \subset \subset \Omega'$. Then, c.f~Appendix A,
\begin{equation}
\lim_{\epsilon \rightarrow 0} \int_{\Omega} \phi\ d\left( \text{Curl}\ \Fpe- \text{Curl}\ \Fp \right) = \lim_{\epsilon \rightarrow 0} \int_{\Omega'} \phi\ d\left( \text{Curl}\ \Fpe- \text{Curl}\ \Fp \right) = \lim_{\epsilon \rightarrow 0} - \int_{\Omega'} D\phi \times d\left(\Fpe-\Fp \right) = 0,
\end{equation}
where the cross product is considered between $D\phi$ and each column of $\left(\Fpe-\Fp \right)$.

\end{proof}

\begin{theorem} \label{Thm:F=FeFp}
Let $\boldsymbol \varphi_{\epsilon}$, $\F_{\epsilon} = D\boldsymbol \varphi_{\epsilon}$, $\F^\el_{\epsilon}$, $\F^\pl_{\epsilon}$ as defined in Section \ref{Sec:ProblemSetting},  and $\F=D\boldsymbol \varphi$, $\F^\el$ and $\F^\pl$ as in Theorem \ref{Thm:Compactness}. Further, $\sup_{\epsilon} E_{\epsilon}(\pe) < \infty$. Then, 
\begin{equation}\label{eqffefp}
\F = \F^\el\F^\pl.
\end{equation}
Additionally, $\F^\pl \in L^{\infty}\left(\Omega;\mathbb{R}^{2\times2} \right)$, $\F \in L^2\left(\Omega;\mathbb{R}^{2\times2} \right)$ and $\boldsymbol \varphi \in 
 W^{1,2} \left(\Omega; \mathbb{R}^2 \right)$. 
\end{theorem}

\begin{proof}
In the following, we show $D\boldsymbol \varphi_{\epsilon} \overset{*}{\rightharpoonup} \F^\el\F^\pl$ in the sense of distributions. By  Theorem \ref{Thm:Compactness} and uniqueness of the limit, this implies (\ref{eqffefp}).  

We begin by defining $\boldsymbol \mu_{\epsilon}$ as the mollification of $\F^\pl_{\epsilon}$ with a kernel 
$\eta_{\epsilon} \in C^{\infty}_c(B_{\epsilon})$ satisfying $\int_{B_{\epsilon}} \eta_{\epsilon}\, dX=1$  and $\eta_\epsilon(X)=\eta_\epsilon(-X)$, i.e.  
\begin{equation}
\boldsymbol \mu_{\epsilon}(\mathbf{X}) = (\F^\pl_{\epsilon} * \eta_{\epsilon})(\mathbf{X}) = \int_{B_{\epsilon}}  \eta_{\epsilon} (\mathbf{X}-\mathbf{Y})\, d\F^\pl_{\epsilon} (\mathbf{Y}).
\end{equation}
Then, for any $K\subset \subset \Omega$, and $\epsilon<\text{dist}(\partial \Omega, K)/(2L^2)$, one has, by $\eta_{\epsilon} \leq C_{17}\frac{1}{\epsilon^2} \chi_{B_{\epsilon}}$, \eqref{Eq:NoClustering} and \eqref{Eq:b_bound}, that $\|\boldsymbol \mu_{\epsilon}\|_{L^{\infty}(K)}$ is equibounded, the bound does not depend on $K$.
Therefore  there exist  $\boldsymbol\mu\in L^\infty(\Omega;\mathbb{R}^{2\times2})$ and
a subsequence satisfying $\boldsymbol \mu_{\epsilon} \overset{*}{\rightharpoonup}\boldsymbol  \mu$ in $L^\infty(K;\mathbb{R}^{2\times2})$ for all $K\subset\subset\Omega$. 
Next, we show that $\boldsymbol \mu = \F^\pl$. Let $\phi \in C^{\infty}_c(\Omega)$. By  choosing  $K$ such that supp$(\phi) \subset \subset K \subset\subset \Omega$, and choosing $\epsilon < \text{dist}(\partial \Omega, K)/(2L^2)$,
\begin{equation}
\begin{split}
\int_{\Omega} \phi  \left(d\F^\pl_{\epsilon} - d\boldsymbol \mu_{\epsilon }\right)&=\int_{K} \phi  \left(d\F^\pl_{\epsilon} - d\boldsymbol \mu_{\epsilon }\right) = \int_{K} \phi \left(  d\F^\pl_{\epsilon} - d\left(\F^\pl_{\epsilon} *\eta_{\epsilon } \right)\right) = \int_{K} \left( \phi - \phi * \eta_{\epsilon}\right) d\F^\pl_{\epsilon} \\
&\leq \sup_{\epsilon} |\phi - \phi * \eta_{\epsilon} | |\F^\pl_{\epsilon}|(\Omega) \rightarrow 0, \ \text{as } \epsilon \rightarrow 0,
\end{split}
\end{equation} 
since $|\F^\pl_{\epsilon}|(\Omega)$ equibounded, c.f.~proof of Theorem \ref{Thm:Compactness}, and $ \phi * \eta_{\epsilon} \rightarrow \phi$ uniformly as $ \epsilon \rightarrow 0$. Therefore $\boldsymbol \mu = \F^\pl$ and $\F^\pl \in  L^{\infty}\left(\Omega;\mathbb{R}^{2\times2} \right)$. Furthermore,  writing for simplicity, with abuse of notation, $D\boldsymbol \varphi_{\epsilon} dX$ and 
$\F^\pl_\epsilon dX$ for $dD\boldsymbol \varphi_{\epsilon}$ and 
$d\F^\pl_\epsilon$,
\begin{equation}
\begin{split}
\Big| \int_{\Omega} \phi \left(D\boldsymbol \varphi_{\epsilon} - \F^\el\F^\pl \right) dX \Big| &= \Big| \int_{K} \phi \left(D\boldsymbol \varphi_{\epsilon} - \F^\el\F^\pl \right) dX \Big|\\
& \leq \int_{K} |\phi \left(D\boldsymbol \varphi_{\epsilon} - \F^\el_{\epsilon}\F^\pl_{\epsilon} \right)| dX + \Big| \int_{K} \phi \left(\F^\el_{\epsilon}\F^\pl_{\epsilon} - \F^\el_{\epsilon}\boldsymbol  \mu_{\epsilon} \right) dX \Big| \\
&\phantom{\leq}+ \int_{K} |\phi \left(\F^\el_{\epsilon}\boldsymbol \mu_{\epsilon} - \F^\el\boldsymbol  \mu_{\epsilon} \right)| dX+  \Big| \int_{K} \phi \left(\F^\el\boldsymbol  \mu_{\epsilon} - \F^\el\F^\pl \right) dX \Big|.
\end{split}
\end{equation}

The first term, $ \int_{K} |\phi \left(D\boldsymbol \varphi_{\epsilon} - \F^\el_{\epsilon}\F^\pl_{\epsilon} \right)| dX$, tends to zero  by Proposition \ref{Thm::General_FeFp}. The third term, \newline $\int_{K} |\phi \left(\F^\el_{\epsilon}\boldsymbol \mu_{\epsilon} - \F^\el \boldsymbol \mu_{\epsilon} \right)| dX$, tends to zero by the strong $L^1(K)$ convergence of $\F^\el_{\epsilon}$, c.f.~\eqref{Eq:Fe_compactness_L1}, and $\boldsymbol \mu_{\epsilon} \in L^{\infty}(K)$. The last term, $\Big| \int_{K} \phi \left(\F^\el \boldsymbol \mu_{\epsilon} - \F^\el\F^\pl \right) dX \Big|$, tends to zero since $ \phi \F^\el \in L^1(K)$ and $\boldsymbol \mu_{\epsilon}$ weakly$*$ converges to $\F^\pl$ in $L^{\infty}(K)$. 
Finally, in the second term, $ \Big|\int_{K} \phi \left(\F^\el_{\epsilon}\F^\pl_{\epsilon} - \F^\el_{\epsilon} \boldsymbol \mu_{\epsilon} \right) dX \Big|$, 
we add and subtract the identity, use $\mathbf{I}*\eta_{\epsilon}=\mathbf{I}$ and a change of variables in the second term, to get
\begin{equation}
\begin{split}
 \int_K \phi \left(\F^\el_{\epsilon}\F^\pl_{\epsilon} - \F^\el_{\epsilon}\boldsymbol \mu_{\epsilon} \right) dX  &=\int_{K} \Big[ \phi \F^\el_{\epsilon}\left(\F^\pl_{\epsilon}-\mathbf{I}\right) - \phi \F^\el_{\epsilon}\left(\left(\F^\pl_{\epsilon}-\mathbf{I}\right)*\eta_{\epsilon} \right) \Big] dX  \\
 &= \int_K \Big[ \phi \F^\el_{\epsilon}\left(\F^\pl_{\epsilon}-\mathbf{I}\right) - \left(\left(\phi \F^\el_{\epsilon} \right)*\eta_{\epsilon} \right)\left(\F^\pl_{\epsilon}-\mathbf{I}\right) \Big] dX \\
 &= \int_{K\cap\mathcal J} \Big[ \phi \F^\el_{\epsilon} - \left(\phi \F^\el_{\epsilon} \right)*\eta_{\epsilon} \Big] 
{\mathbf b}_\epsilon \otimes \mathbf N d\mathcal H^1\,.
\end{split}
\end{equation}
 We cover the jump set $\mathcal{J}$ with balls $B_{\epsilon}(\mathbf{Y}_
l)$ such that the double balls $B_{2\epsilon}(\mathbf{Y}_
l)$ are contained in $\Omega$ and have finite overlap.
We use
(\ref{Eq:PoincareSegmentsBall}) in each ball, both for $ \phi  \F^\el_{\epsilon}$ and   for $ \left( \phi  \F^\el_{\epsilon} \right)*\eta_{\epsilon} $, 
the chain rule, the quadratic growth of the elastic energy and the energy bound, 
\begin{equation}
\begin{split}
 \Big|\int_{K} \phi \left(\F^\el_{\epsilon}\F^\pl_{\epsilon} - \F^\el_{\epsilon}\boldsymbol \mu_{\epsilon} \right) dX \Big| 
  & \leq    \sum_l  n\epsilon \int_{B_{\epsilon}(\mathbf{Y}_l)\cap \mathcal{J}} \Big| \left( \phi \F^\el_{\epsilon} - \left( \phi  \F^\el_{\epsilon} \right)*\eta_{\epsilon} \right) \Big| d\mathcal{H}^1  \\
 &\leq  \sum_l  n\epsilon C_{18}  | D\left( \phi  \F^\el_{\epsilon} \right) |(B_{2\epsilon}(\mathbf{Y}_l)) \\
 & \leq C_{19} n \epsilon \left(\|\phi \|_{L^{\infty}(\Omega)}  |D\F^\el_{\epsilon}| (\Omega) + \|D\phi \|_{L^{\infty}(\Omega)} \int_{\Omega} |\F^\el_{\epsilon}|\, d\mathbf{X}\right) \leq C_{20} \epsilon.
\end{split}
\end{equation}
 Therefore $D\boldsymbol \varphi_{\epsilon}$ converges distributionally to $\F^\el\F^\pl$.
By uniqueness of the limit, then $\F=\F^\el\F^\pl$. Further, since $\F^\el \in L^{2}\left(\Omega;\mathbb{R}^{2\times2} \right)$, and  $\F^\pl \in L^{\infty}\left(\Omega;\mathbb{R}^{2\times2} \right)$, it immediately follows that  $\F \in L^{2}\left(\Omega;\mathbb{R}^{2\times2} \right)$. By the Poincar\'e inequality,  $\boldsymbol \varphi \in L^2 \left(\Omega; \mathbb{R}^2 \right)$.

\end{proof}

\begin{theorem}
Let $\Fp$ be as in Theorem \ref{Thm:Compactness}, $\sup_{\epsilon} E_{\epsilon} (\pe)< \infty$. Then $\det \Fp =1$.
\end{theorem}

\begin{proof}
Consider the set $\Omega'_2=\{\mathbf{X}\in \Omega: \text{dist}(\mathbf{X},\partial \Omega)\geq 2L^2\epsilon\}$, with $L$ as defined in Lemma \ref{Lemma:up_inverse_Lipschitz}. Next, we cover $\Omega'_2 \setminus \cup_k B_{c\epsilon}(\mathbf{X}_k)$ with balls $B_{\epsilon}(\mathbf{Y}_l)$, such that 
$\mathbf{Y}_l\in\Omega'_2$ and $|\mathbf{Y}_l-\mathbf{X}_k|>c\epsilon$ for all $l,k$. Note that this cover is possible by the constraint on the minimum distance between dislocations, c.f.~Eq.~\eqref{Eq:DislocationsSeparation}, and our choice of $\Omega'_2$. 
 We fix a mollification  kernel $\eta_{\alpha}\in C_c^{\infty}(B_{\alpha})$, where $\alpha=\epsilon/10$.
In   $B_{2L^2\epsilon}(\mathbf{Y}_l)$ we choose a decomposition as in (C2),
 $\boldsymbol \varphi^\pl_{\epsilon} = 
\boldsymbol \varphi^\pl_{\epsilon, N_s} \circ ...  \circ \boldsymbol \varphi^\pl_{\epsilon, 1}$. We define, in the smaller ball  $B_\epsilon(\mathbf{Y}_l)$, the regularized plastic deformation by
\begin{equation}
\p^\pl_{\epsilon \alpha} = \left( \p^\pl_{\epsilon, N_s} * \eta_{\alpha} \right) \circ ... \circ \left(\p^\pl_{\epsilon, \nu} * \eta_{\alpha} \right) \circ ... \circ \left(\p^\pl_{\epsilon, 1} * \eta_{\alpha} \right) \,.
\end{equation} 
We observe that $\p^\pl_{\epsilon \alpha} $ is smooth, 
and
$\det D(\p^\pl_{\epsilon,\nu}  * \eta_{\alpha})=1$ implies
$\det D\p^\pl_{\epsilon \alpha} =1$ everywhere. Further,
since the jumps of each $\p^\pl_{\epsilon,\nu}$ are bounded by $n'\epsilon$ and separated by at least $\epsilon$, c.f.~(C2),
we obtain  $|\p^\pl_{\epsilon,\nu}\ast\eta_\alpha-\p^\pl_{\epsilon,\nu}|\le n'\epsilon$ and 
 $|D\left(\p^\pl_{\epsilon, \nu} * \eta_{\alpha} \right)|\le c n'\epsilon/\alpha$. Iterating
 we deduce that $|\p^\pl_{\epsilon \alpha}-\p^\pl_{\epsilon}| \leq C_{21} \epsilon $ and $|D\p^\pl_{\epsilon \alpha}|\le C_{22}$. 

We now show that 
$D\p^\pl_{\epsilon \alpha}$ is uniquely defined, independently on the choice of $l$ and the decomposition of 
$\p^\pl_{\epsilon}$. Let 
$\hat{\boldsymbol \varphi}^\pl_{\epsilon} = 
\hat{\boldsymbol \varphi}^\pl_{\epsilon, N_s} \circ ...  \circ \hat{\boldsymbol \varphi}^\pl_{\epsilon, 1}$
be a different decomposition. 
By Lemma \ref{Lemma:UniqueDecompositionVarphip}, there are constants $\mathbf c_\nu$ such that
$\boldsymbol \varphi^{\pl}_{\epsilon, \nu}(\mathbf{X}) = \hat{\boldsymbol \varphi}^{\pl}_{\epsilon, \nu} (\mathbf{X}-\mathbf{c}_{\nu-1})+\mathbf{c}_{\nu}$, which implies
$(\boldsymbol \varphi^{p}_{\epsilon, \nu}* \eta_{\alpha} )(\mathbf{X}) = (\hat{\boldsymbol \varphi}^{\pl}_{\epsilon, \nu}* \eta_{\alpha} ) (\mathbf{X}-\mathbf{c}_{\nu-1})+\mathbf{c}_{\nu}$. Iterating as in 
(\ref{eqiteratcnu1}-\ref{eqiteratcnu2}) gives $D\p^\pl_{\epsilon \alpha} =D\hat \p^\pl_{\epsilon \alpha}$.
We define $\F^\pl_{\epsilon \alpha}$ as $D\p^\pl_{\epsilon\alpha}$ on $\cup_l B_{\epsilon}(\mathbf{Y}_l)$ and zero elsewhere.

Next, we choose a constant $1/2  < \rho <1$, and define $\Omega''=\{\mathbf{X}\in \Omega: \text{dist}(\mathbf{X},\partial \Omega)\geq \epsilon^{\rho}\}$. 
We now construct a cover of $\Omega''$ with balls of radius of order $\epsilon^\rho$.
Precisely, we 
first select balls $B_{\epsilon^{\rho}}(\mathbf{X}_k)$ centered at the dislocation points, and then finitely many 
$\mathbf Y_l\in H=
\{\mathbf{X}\in \Omega: \text{dist}(\mathbf{X},\partial \Omega)\geq \epsilon^{\rho}/2\}\setminus \cup_k B_{\epsilon^{\rho}/2}(\mathbf{X}_k)$ such that $H\subset
\cup_l B_{\epsilon^{\rho}/4}(\mathbf{Y}_l)$ (this is possible since $H$ is compact).
Let $\phi_k'\in C^\infty_c(B_{\epsilon^{\rho}/2}(\mathbf{X}_k))$, $\phi_l\in C^\infty_c(B_{\epsilon^{\rho}/4}(\mathbf{Y}_l))$ be a 
partition of unity on  $\{\mathbf{X}\in \Omega: \text{dist}(\mathbf{X},\partial \Omega)\geq \epsilon^{\rho}/2\}$.
We define
$\theta_k'=\phi_k' \ast \eta_{\epsilon^\rho/4}$ and 
$\theta_l=\phi_l \ast \eta_{\epsilon^\rho/4}$. Then $\sum_k \theta_k'+\sum_l\theta_l=1$ on $\overline{\Omega''}$,
 $\theta_k'\in C^\infty_c(B_{\epsilon^{\rho}}(\mathbf{X}_k))$, $\theta_l\in C^\infty_c(B_{\epsilon^{\rho/}2}(\mathbf{Y}_l))$,
 and $|D\theta_k'|, |D\theta_l|\le  C_{23}\epsilon^{-\rho}$.

\begin{figure}
\begin{center}
    {\includegraphics[width=0.65\textwidth]{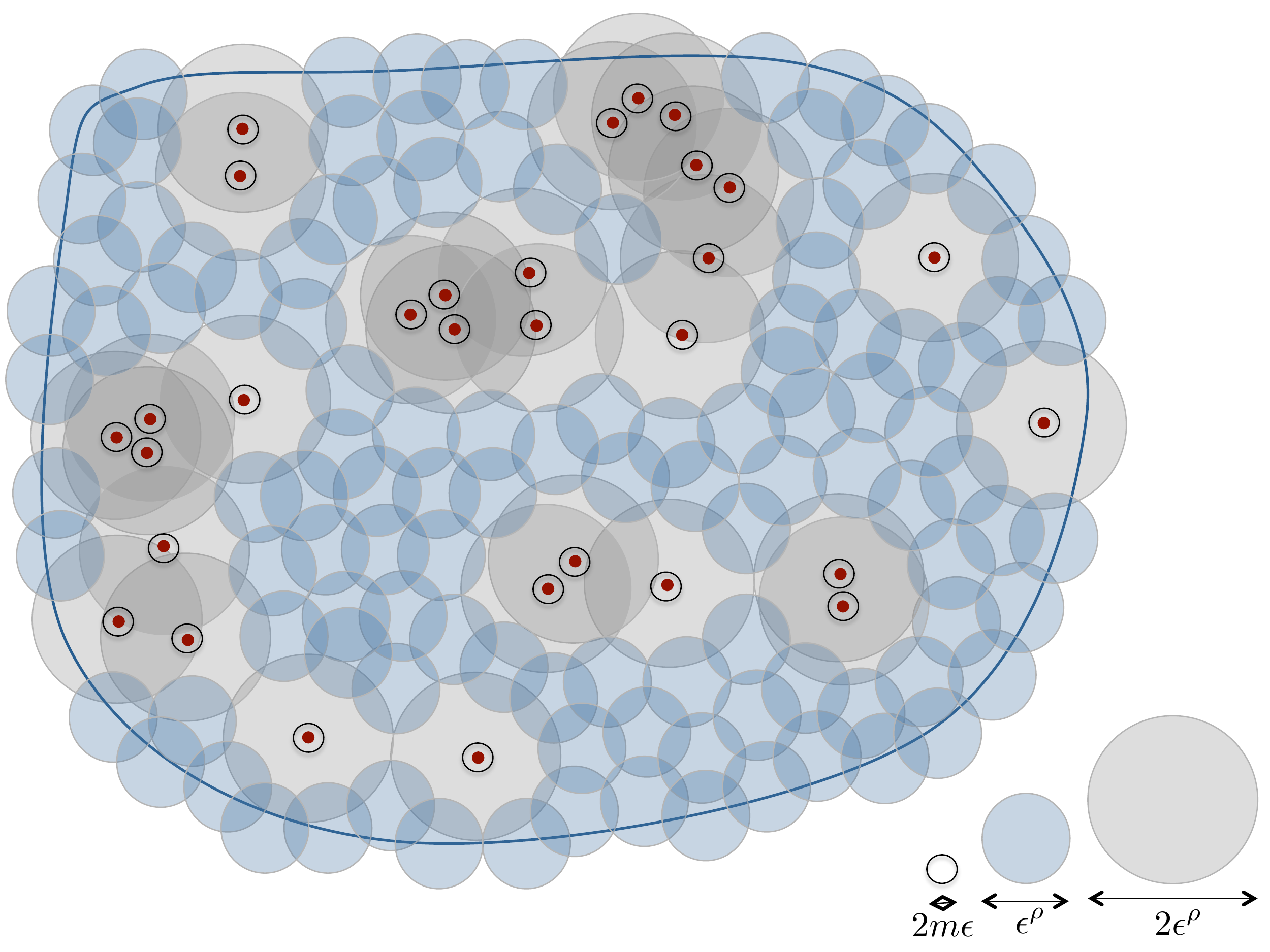}}
    \caption[]{Cover of $\Omega''$ with balls $B_{\epsilon^{\rho}}$ centered at the dislocation points $\mathbf{X}_k$ (grey balls), and balls $B_{\epsilon^{\rho}/2}$ with finite overlap (blue balls) that do not overlap with the dislocation cores $B_{m\epsilon}(\mathbf{X}_k)$.}
    \label{Fig:Delaunay}
\end{center}
\end{figure}

For $\phi\in C^{\infty}_c(\Omega)$, consider $K$ such that $\text{supp}(\phi) \subset \subset K \subset \subset \Omega$ and take  $\epsilon < \left(\text{dist}(K,\partial \Omega) \right)^{(1/\rho)}$,
 so that $K\subset\Omega''$. Then, using the previously obtained relation $\sum_l \theta_l + \sum_k \theta'_k=1$
 one obtains,
 writing as above $\F^\pl_{\epsilon} dX$ for $d\F^\pl_{\epsilon}$,  and $\omega_l=B_{\epsilon^{\rho}/2}(\mathbf{Y}_l)$
\begin{equation}
\begin{split}
\lim_{\epsilon \rightarrow 0} \Big | \int_{\Omega} (\F^\pl_{\epsilon \alpha}-\F^\pl_{\epsilon}) \phi \, dX \Big |&= \lim_{\epsilon \rightarrow 0}  \Big | \int_{K} (\F^\pl_{\epsilon \alpha}-\F^\pl_{\epsilon}) \phi \, dX \Big |\\
& \leq \lim_{\epsilon \rightarrow 0} \sum_l  \Big | \int_{\omega_l} \theta_l(\F^\pl_{\epsilon \alpha}-\F^\pl_{\epsilon}) \phi \, dX \Big | + \lim_{\epsilon \rightarrow 0}\sum_k \Big |  \int_{B_{\epsilon^{\rho}}(\mathbf{X}_k)} \theta'_k(\F^\pl_{\epsilon \alpha}-\F^\pl_{\epsilon}) \phi \, dX \Big | \,.
\end{split}
\end{equation}
In the first term we integrate by parts and 
use the previously obtained relations 
 $|\p^\pl_{\epsilon \alpha}-\p^\pl_{\epsilon}| \leq C_{21} \epsilon $ and $|D\theta_l| \leq C_{23}\epsilon^{-\rho}$, 
\begin{equation}
\begin{split}
\lim_{\epsilon \rightarrow 0} \sum_l  \Big | \int_{\omega_l} \theta_l(\F^\pl_{\epsilon \alpha}-\F^\pl_{\epsilon}) \phi \, dX \Big |
&= \lim_{\epsilon \rightarrow 0}  \sum_l  \Big |\int_{\omega_l} \theta_l(D\p^\pl_{\epsilon \alpha}-D\p^\pl_{\epsilon}) \phi \, dX \Big | \\
& = \lim_{\epsilon \rightarrow 0} \sum_l \Big | \int_{\omega_l} (\p^\pl_{\epsilon \alpha}-\p^\pl_{\epsilon})D (\theta_l\phi)  \, dX   \Big |\\
& \leq \lim_{\epsilon \rightarrow 0} C_{24}\epsilon \left( \epsilon^{-\rho} \|\phi \|_{L^{\infty}} + \| D\phi\|_{L^{\infty}} \right) = 0.
\end{split}
\end{equation}
In the second one we use $|\theta'_k|\le 1$ to estimate
\begin{equation}
\begin{split}
\lim_{\epsilon \rightarrow 0}\sum_k \Big |  \int_{B_{\epsilon^{\rho}}(\mathbf{X}_k)} \theta'_k(\F^\pl_{\epsilon \alpha}-\F^\pl_{\epsilon}) \phi \, dX \Big | 
&\le
\|\phi\|_{L^\infty}  \lim_{\epsilon \rightarrow 0}
\sum_k\left( \|\F^\pl_{\epsilon \alpha}\|_{L^\infty} |B_{\epsilon^{\rho}}(\mathbf{X}_k)| +
|\F^\pl_{\epsilon}|(B_{\epsilon^{\rho}}(\mathbf{X}_k))\right)\,.
\end{split}
\end{equation}
By Lemma \ref{Lemma:SizeControl}, $|\F^\pl_{\epsilon}|(B_{\epsilon^{\rho}}(\mathbf{X}_k))\le 
C_{25} \epsilon^{2\rho}$, and recalling that  $\|\F^\pl_{\epsilon \alpha}\|_{L^\infty}$ is bounded,
$N_\epsilon\le C/\epsilon$  and  $\rho>1/2$
we conclude 
\begin{equation}
\begin{split}
\lim_{\epsilon \rightarrow 0}\sum_k \Big |  \int_{B_{\epsilon^{\rho}}(\mathbf{X}_k)} \theta'_k(\F^\pl_{\epsilon \alpha}-\F^\pl_{\epsilon}) \phi \, dX \Big | 
&\le
\|\phi\|_{L^\infty}  \lim_{\epsilon \rightarrow 0} C_{26}
N_\epsilon \epsilon^{2\rho}=0\,.
\end{split}
\end{equation}

\begin{figure}
\begin{center}
    {\includegraphics[width=0.65\textwidth]{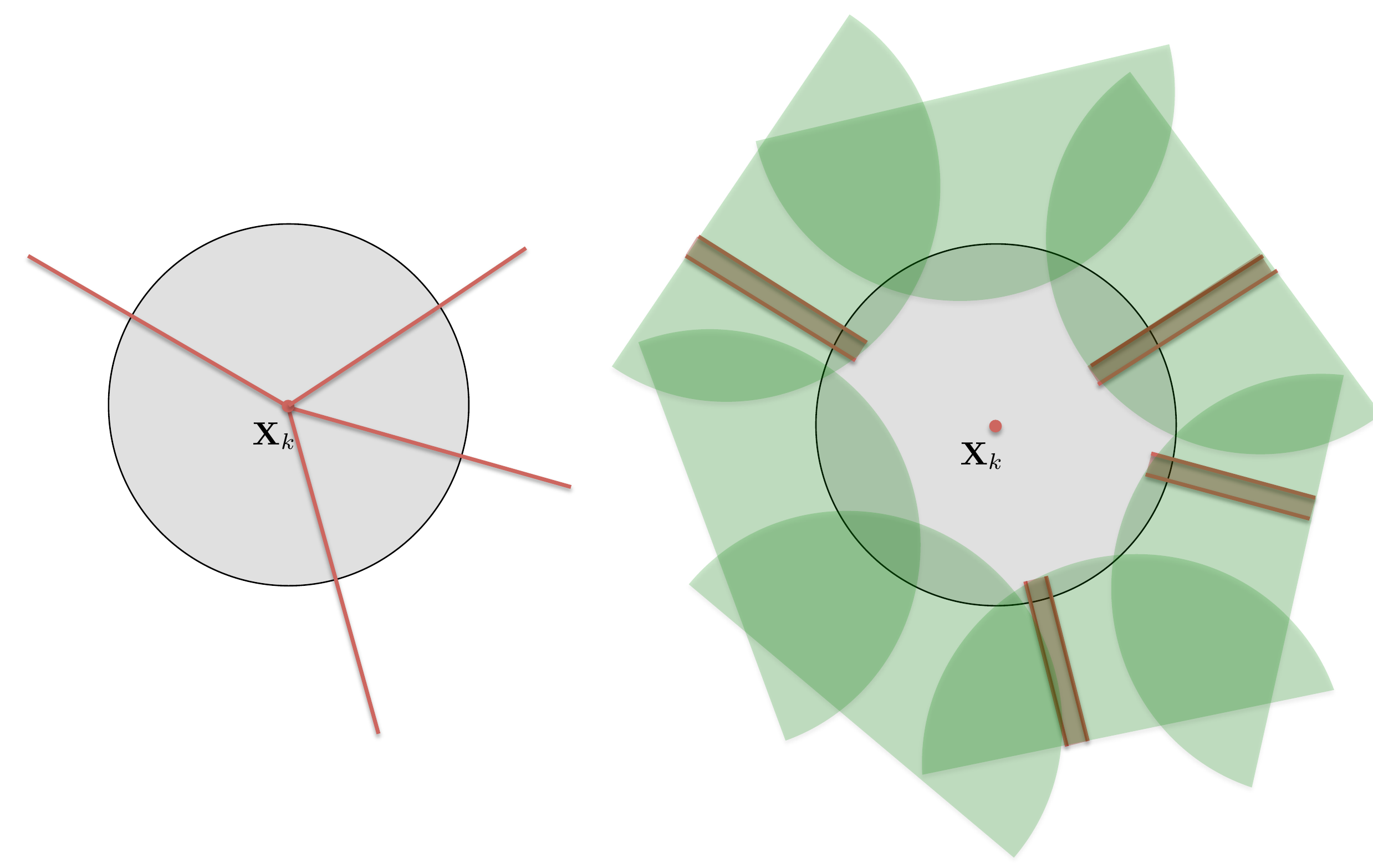}}
    \caption[]{The figure shows the support of $\Fpe$ in red (left image) and $\Fp_{\epsilon\alpha}$ (right image) near a given dislocation core $B_{c\epsilon}(\mathbf{X}_k)$. The support of $\text{Curl} \Fpe$ is $\mathbf{X}_k$, whereas the support of $\text{Curl} \Fp_{\epsilon \alpha}$ will be $N_{dk} \leq N_s$ diffuse dislocations located at the boundary of the domains $B_{\epsilon}$ in which a global definition of $\Fp_{\epsilon \alpha}$ was obtained.}
    \label{Fig:Fp_approx}
\end{center}
\end{figure}
 
As a result $\F^\pl_{\epsilon \alpha}\overset{*}\rightharpoonup \F^\pl$ in $\mathcal{M}$ and therefore in $L^\infty$ as $\epsilon$ tends to zero. Moreover, by the definition of $\p^\pl_{\epsilon \alpha}$, $\text{Curl } \F^\pl_{\epsilon \alpha} =0 $ in $\Omega'_2 \setminus \cup_k B_{c\epsilon}(\mathbf{X}_k)$. 
Furthermore, $|\text{Curl } \F^\pl_{\epsilon \alpha}|(\Omega'_2) \leq N_{\epsilon}  N_s n \epsilon < \infty$ 
since each of the $N_{\epsilon}$ dislocation points has at most $ N_s$ radial jump sets and each of them has a Burgers vector bounded by $n\epsilon$, see Fig.~\ref{Fig:Fp_approx}. We recall that  the determinant of a 2 by 2 tensor can be expressed as the scalar product of two vectors whose Curl and Div (divergence) are controlled by the Curl of the tensor. Indeed, 
\begin{equation}
\det \F = (F_{11}, F_{12})\times(F_{21},F_{22}) = (F_{11}, F_{12})\cdot(F_{22}, -F_{21}),
\end{equation}
with Curl $(F_{11}, F_{12})=F_{12,1}-F_{11,2}$, and Div $(F_{22}, -F_{21}) = F_{22,1} -F_{21,2} = \text{Curl}\ (F_{21}, F_{22})$. 
Therefore the div-curl lemma \citep{Evans1990} implies $\det \F^\pl_{\epsilon \alpha} \overset{*}\rightharpoonup \det \F^\pl$ in  $L^\infty$. 

Next we show that $\det \F^\pl_{\epsilon \alpha} \to 1$ in $L^1$.
Since $\det \F^\pl_{\epsilon \alpha}=1$ on the union of the $\omega_l$, and it is zero elsewhere,
\begin{equation}
\|\det\F^\pl_{\epsilon \alpha}-1\|_{L^1(\Omega_2')}\le \sum_k |B_{\epsilon^\rho}(\mathbf X_k)| \le \pi N_\epsilon \epsilon^{2\rho}\to0
 \end{equation}
since $\rho>1/2$. Therefore  $\det \F^\pl = 1$ a.e., which is the sought-after result.
\end{proof}

In closing we remark that the determinant of the gradient of a $SBV$ function such as $\p_\epsilon$ can be given
 a distributional interpretation. This formulation is not used in the proof above, but illustrates the 
 idea behind the statement $\det \mathbf F^\pl=1$.
Indeed, in Appendix A we show that if  $\p_{\epsilon} \in  L^\infty\cap SBV(\Omega;\mathbb R^2)$  and $\phi \in C_c^{\infty}(\Omega)$, the determinant can be interpreted as
\begin{equation}\label{Eq:Det_SBV}
\begin{split}
\langle\det D\pe, \phi\rangle &=  - \frac{1}{2}\int_{\Omega} \pe^T\  \text{cof } \nabla \pe\ D\phi \, dX + \frac{1}{2}\int_{\mathcal{J}} \left(\pe \times \llbracket\pe \rrbracket \right) \cdot \left( \N \times D\phi \right)\, d\mathcal{H}^1\,. \\
& =- \frac{1}{2} \int_{\Omega}\left(\pe \otimes D\phi \right) : d\text{cof }D\pe \ ,
\end{split}
\end{equation}
where the precise definition of $\pe$ on $\mathcal{J}$ is irrelevant since $\pe^+ \times \llbracket \pe \rrbracket = \pe^- \times \llbracket \pe \rrbracket$. One can therefore choose for $\pe(\mathbf{X})$ on $\mathbf{X} \in\mathcal{J}$ any convex combination of $\pe^+(\mathbf X)$ and $\pe^-(\mathbf X)$.

\section{Conclusions} \label{Sec:Conclusions}
The kinematic assumption $\F=\F^\el\F^\pl$ pertaining to the modeling of elastoplasticity in the setting of large deformations was introduced in the 1960's based on heuristic arguments and has become standard in the continuum mechanics community. However, the lack of a micromechanical justification has raised some skcepticism over the years about its validity, in particular questioning the existence and uniqueness of the multiplicative decomposition.

In this paper, we provide for the first time, as far as the authors' knowledge, a rigorous proof of the expression $\F=\F^\el\F^\pl$ for a general elastoplastic deformation of a single crystal. The proof is based on the coarse graining of a mesoscopic description of the deformation, where the dislocations and slip surfaces are individually resolved and the displacement field can be treated as continuous in all the domain except at the surfaces over which dislocations have glided. At such scale, as previously shown by two of the authors \citep{ReinaConti2014}, there exists physically-based definitions of the different tensors $\F_{\epsilon}$, $\F^\el_{\epsilon}$ and $\F^\pl_{\epsilon}$ that are uniquely defined from the microscopic deformation gradient $\boldsymbol \varphi_{\epsilon}$ and that do not make use of any multiplicative relationship between them for their definition. Instead, $\pe$ is described in the compatible regions away from the dislocations as the composition of an elastic and plastic 
deformation, i.e. $\pe=\pe^\el\circ \pe^\pl$. The corresponding continuum quantities at the macroscopic scale  $\boldsymbol \varphi$, $\F$, $\F^\el$ and $\F^\pl$ are defined respectively as the limit of  $\boldsymbol \varphi_{\epsilon}$, $\F_{\epsilon}$, $\F^\el_{\epsilon}$ and $\F^\pl_{\epsilon}$ as the lattice parameter $\epsilon$ tends to zero in the appropriate topology. We show that the limiting deformation mapping $\boldsymbol \varphi$ is indeed continuous, that $\F=\F^\el\F^\pl$ holds in the limit with $\det \F^\pl =1$, and that Curl $\Fp$ represents the dislocation density tensor when it is expressed in the reference configuration.

The results presented in this paper are of great generality, although limited so far to two dimensional deformations due to the mathematical complexity involved in the proofs. In particular, the assumptions made consist of standard growth conditions for the elastic energy, dislocations separated from each other and from the boundary by a distance of the order of several atomic spacings, a finite variation of the elastic deformation and a precise sequence of slips assumed in the compatible regions of the domain.  We remark that this latter assumption is only needed for the proof of $\det \Fp=1$ and it has only been introduced for mathematical simplicity. It is yet to be studied in further detail whether or not some of these assumptions are necessary for the final results to hold.
\\

\begin{appendix}
\section*{Appendix A}
{\bf Distributional Curl of a second order tensor.}
Let $\phi \in C^{\infty}_c(\Omega)$, and $\Fp$ smooth. Then using indicial notation and integration by parts,
\begin{equation}
 \int_{\Omega} \phi\ \left(\text{Curl}\ F^{\pl}\right)_{i} dX = \int_{\Omega} \phi\  e_{3kj}  \left(F^{\pl}\right)_{ij,k} dX = -\int_{\Omega}  e_{3kj}\ \phi_{,k}  \left(F^{\pl}\right)_{ij} dX =  - \int_{\Omega} \left(D\phi \times \Fp\right)_i \, dX,
\end{equation}
where $e_{ijk}$ is the Levi-Civita symbol. The equality of the first and last term is taken as the definition for the distributional Curl.

{\bf Determinant of $D\boldsymbol \varphi$, with $\boldsymbol \varphi \in$ SBV.}
 Let $\phi \in C^{\infty}_c(\Omega)$, and $\boldsymbol \varphi$ smooth. Then, the weak and strong form of the determinant coincide and read
\begin{equation} 
\int_{\Omega} \left(\det D\boldsymbol \varphi\right) \phi \, dX = - \frac{1}{2}\int_{\Omega}  \boldsymbol \varphi^T\left( \text{cof}\ D\boldsymbol \varphi \right) D\phi \, dX =  - \frac{1}{2}\int_{\Omega} \left( \boldsymbol \varphi \otimes D\phi\right) : \left( \text{cof}\ D\boldsymbol \varphi \right)  \, dX.
\end{equation}
Indeed,
\begin{equation}
\begin{split}
 \boldsymbol \varphi^T\left( \text{cof}\ D\boldsymbol \varphi \right) D\phi &= \left( \begin{array}{c c}
 \varphi_1 &  \varphi_2  \end{array}  \right) \left( \begin{array}{cc}
 \varphi_{2,2} & -  \varphi_{2,1} \\
-  \varphi_{1,2} &  \varphi_{1,1} \end{array} \right)  \left( \begin{array}{c}
\phi_{,1}  \\
\phi_{,2} \end{array} \right) =  \left( \begin{array}{c c}
 \varphi_1 &  \varphi_2  \end{array}  \right)  \left( \begin{array}{c}
\varphi_{2,2}\phi_{,1}-\varphi_{2,1}\phi_{,2}  \\
- \varphi_{1,2}\phi_{,1}+\varphi_{1,1}\phi_{,2} \end{array} \right)\\
&=  \varphi_1 \varphi_{2,2}\phi_{,1}- \varphi_1  \varphi_{2,1}\phi_{,2} -\varphi_2  \varphi_{1,2}\phi_{,1}+\varphi_2 \varphi_{1,1}\phi_{,2}.
\end{split}
\end{equation}
Therefore, after integrating by parts,
\begin{equation}
\begin{split}
 - \int_{\Omega}  \boldsymbol \varphi^T &\left( \text{cof}\ D\boldsymbol \varphi \right) D\phi \, dX \\
 &= \int_{\Omega} \left( \varphi_{1,1}\varphi_{2,2} +\varphi_1 \varphi_{2,21} -  \varphi_{1,2}\varphi_{2,1}- \varphi_1 \varphi_{2,12}-\varphi_{2,1} \varphi_{1,2}- \varphi_2  \varphi_{1,21} +\varphi_{2,2} \varphi_{1,1} + \varphi_2 \varphi_{1,12} \right) \phi \, dX \\
 & =2 \int_{\Omega}  \left( \varphi_{1,1}\varphi_{2,2}-\varphi_{1,2}\varphi_{2,1}\right)\phi \, dX=2 \int_{\Omega}\left( \det D\boldsymbol \varphi\right) \phi \, dX.
\end{split}
\end{equation}

For $\boldsymbol \varphi \in SBV$, $\text{cof}\ D\boldsymbol \varphi = \text{cof}\ \nabla \boldsymbol \varphi \mathcal{L}^2+ \text{cof}\left(\llbracket \boldsymbol \varphi \rrbracket \otimes \mathbf{N}\right)\ \mathcal{H}^1 \lfloor_{\mathcal{J}}$, with
\begin{equation}
\begin{split}
\text{cof}\left(\llbracket \boldsymbol \varphi \rrbracket \otimes N\right) &= \text{cof} \left( \begin{array}{cc}
\llbracket \varphi_1 \rrbracket N_1  & \llbracket \varphi_1 \rrbracket  N_2 \\
\llbracket  \varphi_2 \rrbracket N_1  & \llbracket \varphi_2 \rrbracket N_2  \end{array} \right) = \left( \begin{array}{cc}
\llbracket \varphi_2 \rrbracket N_2  & -\llbracket \varphi_2 \rrbracket  N_1 \\
-\llbracket \varphi_1 \rrbracket N_2  & \llbracket \varphi_1 \rrbracket N_1  \end{array} \right) \\
&= \left( \begin{array}{c}
-\llbracket \varphi_2 \rrbracket  \\
\llbracket \varphi_1 \rrbracket   \end{array} \right) \otimes \left( \begin{array}{c}
-N_2  \\
N_1   \end{array}\right) = \llbracket \boldsymbol \varphi \rrbracket^{\bot} \otimes \mathbf{N}^{\bot},
\end{split}
\end{equation}
\begin{equation}
\begin{split}
\boldsymbol \varphi^T\ \text{cof} \left(\llbracket \boldsymbol \varphi \rrbracket \otimes N\right) D\phi &= \left(\boldsymbol \varphi \cdot \llbracket \boldsymbol \varphi\rrbracket^{\bot} \right) \left(\mathbf{N}^{\bot} \cdot D\phi \right)\\
&=  \left( \left( \begin{array}{c}
\varphi_1  \\
\varphi_2   \end{array}\right) \cdot
 \left( \begin{array}{c}
-\llbracket\varphi_2 \rrbracket  \\
\llbracket \varphi_1 \rrbracket   \end{array} \right)  \right)\left( \left( \begin{array}{c}
-N_2  \\
N_1   \end{array}\right) \cdot
 \left( \begin{array}{c}
\phi_{,1}  \\
\phi_{,2}   \end{array} \right)  \right) \\
&= \left( - \boldsymbol \varphi \times \llbracket \boldsymbol \varphi \rrbracket \right) \cdot \left( \mathbf{N} \times D\phi\right),
\end{split}
\end{equation} 
therefore giving \eqref{Eq:Det_SBV}.
\end{appendix}

\section*{Acknowledgements}

C. Reina acknowledges the NSF grant CMMI - 1401537. 
S. Conti acknowledges support from  the  Deutsche Forschungsgemeinschaft
through SFB 1060 {\em ``The Mathematics of emergent effects''}, project A5.
All the authors further acknowledge support from the Hausdorff Center for Mathematics. 

%\bibliographystyle{plainnat}
%\bibliography{Biblio}

\end{document}